\documentclass[generic,preprint]{imsart}
\usepackage{natbib}
\usepackage{graphicx}
\usepackage{array}
\usepackage{multirow}
\usepackage{amsthm}
\usepackage{amsmath}
\usepackage{amsfonts}
\usepackage{amssymb}
\usepackage{pgf}
\usepackage{bm}
\usepackage{paralist}
\usepackage{pdflscape}
\usepackage{rotating}
\usepackage{scalefnt}
\usepackage{xr}
\usepackage{verbatim}
\usepackage{placeins}
\usepackage{subcaption}
\usepackage{float}

\usepackage{hyperref}
\usepackage{tikz}
\usepackage{ifthen}
\usetikzlibrary{arrows,shapes}
\usepackage{xr}

\RequirePackage[OT1]{fontenc}
\RequirePackage{amsthm,amsmath}
\RequirePackage[colorlinks,citecolor=blue,urlcolor=blue]{hyperref}
\startlocaldefs
\numberwithin{equation}{section}
\theoremstyle{plain}
\newtheorem{theorem}{Theorem}
\newtheorem{proposition}{Proposition}
\newtheorem{lemma}{Lemma}
\newtheorem{assumption}{Assumption}
\newtheorem{corollary}{Corollary}

\newcommand{\beq}{\begin{equation}}
\newcommand{\eeq}{\end{equation}}
\newcommand{\bea}{\begin{eqnarray}}
\newcommand{\eea}{\end{eqnarray}}
\newcommand{\bit}{\begin{itemize}}
\newcommand{\eit}{\end{itemize}}
\newcommand{\ben}{\begin{enumerate}}
\newcommand{\een}{\end{enumerate}}
\newcommand{\bpm}{\begin{pmatrix}}
\newcommand{\epm}{\end{pmatrix}}
\newcommand{\bbm}{\begin{bmatrix}}
\newcommand{\ebm}{\end{bmatrix}}

\endlocaldefs

\begin{document}

\begin{frontmatter}
\title{Sequential monitoring for cointegrating regressions}
\runtitle{Monitoring cointegration}

\begin{aug}
\author{\fnms{Lorenzo} \snm{Trapani}$^{\dag}$\ead[label=e1]{lorenzo.trapani@nottingham.ac.uk}}\hskip .2cm
\author{\fnms{Emily} \snm{Whitehouse}$^{*}$\ead[label=e2]{emily.whitehouse@newcastle.ac.uk}}
\runauthor{L. Trapani and E. Whitehouse}

\affiliation{University of Nottingham\thanksmark{m1} and Newcastle University \thanksmark{m2}}

\address{$^{\dag}$University of Nottingham\\
\printead{e1}\\
}

\address{$^{*}$Newcastle University\\
\printead{e2}\\
}
\end{aug}

\begin{abstract}
We develop monitoring procedures for cointegrating regressions, testing the null of no breaks against the alternatives that there is either a change in the slope, or a change to non-cointegration. After observing the regression for a calibration sample $m$, we study a CUSUM-type statistic to detect the presence of change during a monitoring horizon $m+1,...,T$. Our procedures use a class of boundary functions which depend on a parameter, $0 \leq \eta \leq \frac{1}{2}$, whose value affects the delay in detecting the possible break. Technically, these procedures are based on almost sure limiting theorems whose derivation is not straightforward. We therefore define a monitoring function which - at every point in time - diverges to infinity under the null, and drifts to zero under alternatives. We cast this sequence in a randomised procedure to construct an \textit{i.i.d.} sequence, which we then employ to define the detector function. Our monitoring procedure rejects the null of no break (when correct) with a small probability, whilst it rejects with probability one over the monitoring horizon in the presence of breaks.
\end{abstract}


\begin{keyword%
}
\kwd{cointegration}
\kwd{structural change}
\kwd{sequential monitoring}
\kwd{randomized tests}
\end{keyword}

\end{frontmatter}

\newpage

\section{Introduction\label{introduction}}

In this paper, we study the following cointegrating regression 
\begin{equation}
y_{i}=\beta ^{\prime }x_{i}+\epsilon _{i}\text{, }1\leq i\leq T,
\label{model-1}
\end{equation}%
where $\left( y_{i},x_{i}^{\prime }\right) ^{\prime }$ is a $(p+1)\times 1$, 
$I\left( 1\right) $ vector and $\epsilon _{i}\ $is a stationary innovation.
In particular, we investigate the issue of monitoring (\ref{model-1}), after
a calibration period of length $m$, during which our maintained assumptions are
that \textit{(i)} (\ref{model-1}) is a cointegrating relationship and 
\textit{(ii)} the slope $\beta $ is constant. From $i=m+1$ onwards, we check
whether the relationship in (\ref{model-1}) remains constant, or whether
either the slope $\beta $ changes, or (\ref{model-1}) becomes a
non-cointegrating regression (or both).

The timely detection of structural change is arguably of great importance in
the context of any regression model: whilst there is an extensive literature
on the general topic of on-line detection of changes (see e.g. %
\citealp{csorgo1997} for a survey), in the econometrics literature this
issue has received some limited attention since the contribution by %
\citet{chu}. Recent articles that study this topic have focused on linear
regression models (\citealp{lajos04}, \citealp{aue2006}, \citealp{lajos07}, %
\citealp{kap-monitor}), large factor models (\citealp{bt1}), and also
cointegrating regressions (\citealp{steland}, \citealp{wied}, %
\citealp{wagner}). In particular, \citet{wied} consider, essentially, the
same problem as in our paper; namely, they propose several statistics for
the on-line detection of structural breaks in a model like (\ref{model-1}),
considering both the possibility of a change in the slope $\beta$ and a
change to a non-cointegrating regression.

From a methodological viewpoint, we use a residual-based detector to test
for the null hypothesis of no change over the monitoring horizon $m+1\leq
i\leq T$. Note that this corresponds to a closed-end procedure (\citealp{aue12}), since monitoring - as can be expected to happen in practice - stops after $T$.
The family of detectors
which we propose here are based on the sum of squared residuals. Simulations show that our monitoring scheme
has excellent finite sample properties, with low occurrence of false detections and very good
power versus both alternatives under consideration. Other detectors are also
possible (see, for example, the various statistics considered in %
\citealp{homm}, albeit in a different context). 

From a technical point of view, as pointed out by \citet{lajos04} and \citet{lajos07}, the detectors employed in monitoring procedures depend upon a parameter,
henceforth denoted as $\eta$, which can vary in the interval $\left[ 0,\frac{%
1}{2}\right] $. Constructing test statistics when $\eta =0$ (see e.g. \citealp{chu}) requires, as a
technical tool, weak convergence, and therefore one can employ a huge
variety of results which are well-known in the literature (see e.g. the book
by \citealp{billingsley}). On the other hand, the choice $\eta =0$ is known
to often yield inferior results, in particular resulting in a longer delay
in detection of a break (\citealp{aue2004}). In order to overcome this issue, it is usually
recommended to choose $\eta >0$ (\citealp{lajos07}). However, from a
technical viewpoint, using $\eta >0$ requires having stronger forms of
convergence than weak convergence, with fewer results available (we refer to
the textbook by \citealp{csorgo1997} for an excellent treatment of the
subject). For example, to the best of our knowledge we
are not aware of strong approximations like the ones derived in \citet{kmt1}
and \citet{kmt2} for convergence to stochastic integrals, where usually
\textquotedblleft weak\textquotedblright\ results are used instead (see %
\citealp{chanwei}; and \citealp{phillipsweak}). In light of this, we only
rely on (almost sure) rates, and we develop a family of statistics -
computed at each $m+1\leq i\leq T$ - which diverge to positive infinity
under the null of no break, whilst they drift to zero in the presence of
breaks. We then randomize such statistics at each point in time $i $: the
outcome of our randomisation is a sequence of random variables which, under
the null of no break, are \textit{i.i.d.} with finite moments up to any
order, whilst they diverge to infinity in the presence of a break. Finally,
we employ the newly generated sequence in order to construct the same
detectors as in \citet{lajos04} and \citet{lajos07}, being able to rely on
the theory spelt out in those papers. Using randomisation is helpful when
the properties of a certain statistic are not known, or depend on nuisance
parameters: in this respect, it might be envisaged that randomisation serves
a similar purpose to the bootstrap or to self-normalisations (see \citealp{dette2019likelihood} for an example of self-normalisation in the context of monitoring). In the econometric literature,
randomisation has been employed in a wide variety of contexts, including
testing for forecasting superiority (\citealp{corradi2006}), stationarity (%
\citealp{bandi2014}), finiteness of moments (\citealp{trapani16}), boundary
problems (\citealp{HT16}) and determining the number of common factors in a
large factor models (\citealp{trapani17}). In our context, however, we do
not employ randomisation to produce a randomised test, but to construct a
\textquotedblleft well-behaved\textquotedblright\ sequence which, in turn,
can be employed to define an easy-to-study test statistic. In this respect,
our contribution uses the same approach as \citet{bt1}, who study monitoring
for structural change in the context of a large, stationary factor models.
By relying solely on rates, we require quite mild assumptions; all the theory can be based on using a
standard OLS\ estimator, with no need for more specialised estimators like,
say, the FM-OLS estimator (\citealp{phillips-hansen}) or a Dynamic OLS
estimator (\citealp{saikkonen}); and, finally, we do not need to rely on the accuracy of the long-run
variance estimator.

The remainder of the paper is organised as follows. In Section \ref{theory},
we provide the relevant assumptions, and then report theoretical results on
estimation and the monitoring procedure. Extensions to e.g. the case of deterministics are in Section \ref{extensions}. In Section \ref{montecarlo} we
demonstrate the performance of our monitoring procedure through both a Monte
Carlo simulation exercise (Section \ref{simulations}) and an empirical
application to US housing market data (Section \ref{empirics}). Section \ref%
{conclusions} concludes. Proofs of the main results are in Section \ref{proofs}. All technical lemmas and some proofs are relegated to
the Supplement. 

Throughout the paper we use the notation $c_{0}$, $c_{1}$,... to denote
positive and finite constants, that do not depend on the sample size; their
value is allowed to change from line to line. We use the expression
\textquotedblleft a.s.\textquotedblright\ as short-hand for
\textquotedblleft almost surely\textquotedblright ; the ordinary limit is
denoted as \textquotedblleft $\rightarrow $\textquotedblright . Finally, for
a vector $a $ and a matrix $A$, $\left\Vert a\right\Vert $ and $\left\Vert
A\right\Vert $ represent the Euclidean norm. Other notation is introduced
later on in the paper.

\section{Theory\label{theory}}

We begin with introducing some notation and the main assumptions that should
hold under the null of no break (Section \ref{a-h0}); we then move to
discuss the two alternative hypotheses which we consider, namely a change in
the slope and/or a change to a non-cointegrating equation (Section \ref%
{monitor}). Finally, in Section \ref{sequential}, we discuss the relevant
CUSUM process, and the randomisation algorithm. 

\subsection{Main assumptions\label{a-h0}}

Recall (\ref{model-1})%
\begin{equation*}
y_{i}=\beta ^{\prime }x_{i}+\epsilon _{i},
\end{equation*}%
which we assume to be valid during the calibration period $1 \leq i \leq m$, with 
\begin{equation}
x_{i}=x_{i-1}+u_{i}.  \label{model-x}
\end{equation}%
We also define the long run variances of $u_{i}$ and $\epsilon _{i}$ as 
\begin{align}
\Sigma _{u}& =\lim_{m\rightarrow \infty }E\left( \frac{1}{\sqrt{m}}%
\sum_{i=1}^{m}u_{i}\right)\left( \frac{1}{\sqrt{m}}\sum_{i=1}^{m}u_{i}%
\right)^{\prime },  \label{sigma-u} \\
\sigma _{\epsilon }^{2}& =\lim_{m\rightarrow \infty }Var\left( \frac{1}{%
\sqrt{m}}\sum_{i=1}^{m}\epsilon _{i}\right) .  \label{sigma-e}
\end{align}
We consider the following assumption.

\begin{assumption}
\label{as-2}It holds that: (i) $\epsilon _{i}$ and $u_{i}$ have mean zero
with (a) $E\left\vert \epsilon _{i}\right\vert ^{2}<\infty $ for $1\leq
i\leq T$, and $0<\sigma _{\epsilon }^{2}<\infty $; and (b) $\Sigma _{u}$ is
positive definite with $\left\Vert \Sigma _{u}\right\Vert $; (ii) $%
E\left\Vert x_{0}\right\Vert ^{2}<\infty $ and 
\begin{equation}
\sup_{1\leq i\leq t}\left\Vert x_{i}-W_{x}\left( i\right) \right\Vert
=O_{a.s.}\left( t^{1/2-\delta ^{\prime }}\right) ,  \label{sip-x}
\end{equation}%
for some $0<\delta ^{\prime }<\frac{1}{2}$, where $W_{x}\left( i\right) $ is
a $p$-dimensional Wiener process with increments of variance $\Sigma _{u}$;
(iii) $E\left\Vert \sum_{i=1}^{t}x_{i}\epsilon _{i}\right\Vert ^{2}\leq
c_{0}t^{2}$, for all $1\leq t\leq T$; (iv) $E\left\Vert
\sum_{i=1}^{t}x_{i}x_{i}^{\prime }\right\Vert \leq c_{0}t^{2}$, for all $%
1\leq t\leq T$.
\end{assumption}

Assumption \ref{as-2}\textit{(i)} is a standard second moment condition
which is required to hold under the null of no change, and also when the
slope $\beta $ changes but (\ref{model-1}) remains a cointegrating
relationship. Note that, by part \textit{(i)}(b), we rule out cointegration among the regressors. Part \textit{(ii)} of the assumption, in essence, states that
a strong approximation exists for the partial sums process $x_{i}$. This is
a high-level assumption, which could be replaced by more primitive
requirements on the existence of moments for the innovation $u_{i}$, and
some form of weak dependence. It can be envisaged, as far as moments are
concerned, that at least $E\left\Vert u_{i}\right\Vert ^{2+\delta }<\infty $
is required for some $\delta >0$; thence, (\ref{sip-x}) would follow
immediately if $u_{i}$ is \textit{i.i.d.} (see \citealp{kmt1} and %
\citet{kmt2} for the univariate case, and \citealp{gotze} for the
multidimensional one), and also under fairly general forms of weak
dependence such as the case of stationary causal processes including linear
models, Volterra series and models with conditional heteroskedasticity (see %
\citealp{wu2005}, and \citealp{berkesliuwu}). Interestingly, in the
literature it is relatively common to assume a weak Invariance Principle to
hold in lieu of assuming weak dependence (see e.g. Assumption 2 in %
\citealp{wied}). Part \textit{(ii)} of Assumption \ref{as-2} serves exactly
the same purpose, except for the fact that in our paper we need almost sure
rates. Parts \textit{(iii)} and \textit{(iv)} could also be derived under
more primitive conditions on moments, serial dependence, and possible
correlation between $u_{i}$ and $\epsilon _{i}$. For example, the results
could be shown by standard arguments in the case of $u_{i}$ and $\epsilon
_{i}$ being \textit{i.i.d.} and independent of each other; in this case,
existence of second moments would suffice. Part \textit{(iv)} can be shown
under more general forms of dependence, e.g. in the case of linear processes
by exploiting the results in \citet{solo}. Also, in part \textit{(iii)}, the
requirement of independence between $u_{i}$ and $\epsilon _{i}$ is not
necessary: again under the assumption of linear processes, for example, it
could be shown (see, \textit{inter alia}, \citealp{durlauf}, and %
\citealp{phillips-hansen}) that this part of the assumption can hold also in
the presence of endogeneity.

\subsection{Hypotheses of interest and the construction of the monitoring
procedure\label{monitor}}

We base our on-line monitoring on the theory developed in \citet{lajos04}
and \citet{lajos07}. We assume that the data are collected for an initial
calibration period of size $m$ where no break occurs; this can be viewed as
the historic sample available to the researcher. We then define the (length
of) the monitoring horizon $T_{m}$ as $T_{m}=T-m$. Thus, if $T$ represents
the total period considered, $m$ is the amount of time elapsed until the
beginning of the monitoring period. In essence, $m$ is going to be the
sample size used by the researcher for estimation. Choosing $T_{m}$ - that
is, choosing where to stop the monitoring - is an important issue in
sequential analysis, since it can be argued that monitoring comes at a cost
(see the original paper by \citealp{wald}); in this paper, we allow for $%
T_{m}\rightarrow \infty $, under the assumption that monitoring is costless
- this assumption can be realistic when analysing economic series, although
not in other contexts (see e.g. the comments in \citealp{chu}).

\subsubsection{Alternative hypotheses of interest\label{subsub}}

Under the null hypothesis of our monitoring scheme, (\ref{model-1}) is a
cointegrating relationship for the whole monitoring horizon, and the slope $%
\beta $ does not change; rewriting (\ref{model-1}) as
\begin{equation*}
y_{i}=\beta_{i} ^{\prime }x_{i}+\epsilon _{i},
\end{equation*}
we have
\begin{equation}
H_{0}:\left\{ 
\begin{array}{l}
\beta _{i}=\beta \\ 
\epsilon _{i}\text{ is stationary}%
\end{array}%
\right. \text{ for }1\leq i\leq T_{m}.  \label{null}
\end{equation}%
Conversely, when the null does not hold, there could be at least two
interesting, non mutually exclusive alternatives. In the first case, there
could be a structural change whereby, after $i=m$, $\beta $ changes:%
\begin{equation}
H_{A,1}:\beta _{i}=\beta +\Delta _{\beta }I\left[ i>k^{\ast }\right] .
\label{break}
\end{equation}%
In (\ref{break}), $m\leq k^{\ast }<T$ is the potential breakdate. In
addition to this (or as an alternative), (\ref{model-1}) may switch to being
a non-cointegrating relationship at some point in time, viz.%
\begin{equation}
H_{A,2}:\epsilon _{i}=\epsilon _{i-1}+u_{i}^{\epsilon }\text{ for }k^{\ast
}+1\leq i\leq T.  \label{spurious}
\end{equation}%
In both cases, the case of no break is represented by having $k^{\ast }=T$.

As a general comment to our hypothesis testing framework, we point out that
the set-up in (\ref{null})-(\ref{spurious}) mirrors the analysis in %
\citet{wied} very closely. In particular, the null hypothesis is the
intersection of two (very different) requirements: \textit{(a)} the fact
that there is no time variation in the structural parameter $\beta $ in (\ref%
{model-1}) over the monitoring horizon, under the implicitly maintained
hypothesis that (\ref{model-1}) is always a cointegrating regression; and 
\textit{(b)} the fact that (\ref{model-1}) is indeed a cointegrating
relationship during the monitoring horizon. This could be the the set-up of
interest in various applications (see e.g. Section \ref{empirics});
furthermore, an \textquotedblleft omnibus\textquotedblright\ procedure which
is powerful versus a global alternative could be viewed as advantageous in
order to avoid having to test under a maintained hypothesis whose validity
may not always be assumed. On the other hand, the monitoring procedure
proposed in this paper (and in \citealp{wied}) can be argued to be
\textquotedblleft non-constructive\textquotedblright : after rejecting the
null and finding evidence of a change in the nature of (\ref{model-1}), it
is not clear which of the two alternatives the change can be ascribed to. In
the literature, there are tests for more focussed alternatives which could,
in principle, be extended into monitoring procedures. For example, under the
maintained assumption that $\epsilon _{i}$ is stationary over the whole
monitoring horizon, one could think of extending the test for breaks in
cointegrating regressions proposed by \citet{kejriwal}. Similarly, under the
maintained assumption that $\beta $ is constant for the whole interval $%
1\leq i\leq T$, a monitoring procedure could, in principle, be constructed
using the residuals $\widehat{\epsilon }_{i}$, e.g. by extending the test
for a change in persistence proposed by \citet{busetti}. Indeed, under the
same maintained hypotheses mentioned above, our procedure can also be
employed to test, separately, versus the two alternatives mentioned above.
In this respect, the monitoring scheme proposed in this paper could be
viewed as a preliminary step: upon finding evidence that a change occurred,
the researcher may decide to use a more specialised procedure to disentangle
the nature of the change in (\ref{model-1}).

\bigskip

In order to analyse the case of (\ref{spurious}), we need the following
assumption which characterizes the behaviour of $\epsilon _{i}$ under $%
H_{A,2}$.

\begin{assumption}
\label{ha-2}Under $H_{A,2}$, it holds that (i) 
\begin{equation}
\sup_{k^{\ast }+1\leq i\leq t}\left\vert \epsilon _{i}-W_{\epsilon }\left(
i\right) \right\vert =O_{a.s.}\left( t^{1/2-\delta ^{\prime \prime }}\right)
,  \label{sip-2}
\end{equation}%
for all $k^{\ast }+1\leq t\leq T$ and some $0<\delta ^{\prime \prime }<\frac{%
1}{2}$, where $W_{\epsilon }\left( i\right) $ is a Wiener process with
increments of positive variance equal to the long-run variance of $%
u_{i}^{\epsilon }$; (ii) 
\begin{equation*}
E\left\vert \sum_{i=k^{\ast }+1}^{t}\epsilon _{i}^{2}\right\vert \leq
c_{0}t^{2},
\end{equation*}%
for all $k^{\ast }+1\leq t\leq T$.
\end{assumption}

Assumption \ref{ha-2} supersedes parts \textit{(i)} and \textit{(iii)} of
Assumption \ref{as-2} in order to accommodate for the presence of a switch
to a non-cointegrating regression. According to the assumption, in essence,
after the breakdate $k^{\ast }$ the innovation $\epsilon _{i}$ becomes a
unit root process.

\subsubsection{The monitoring function\label{monitor-function}}

Our monitoring scheme is based on a non-recursive estimator of $\beta $:
estimation is carried out using the sample $1\leq i\leq m$ once and for all,
without updating the estimate as $i$ elapses. We focus only on this merely
for the sake of a concise discussion: this choice is not the only possible
one. \citet{lajos04}, \textit{inter alia}, propose a recursive monitoring
procedure (as well as a non-recursive one), where $\beta $ is estimated at
each $i$ using an expanding sample. It seems reasonable to conjecture that,
even in our context, the non-recursive scheme is probably likely to be less
affected by outliers, thus ensuring a better size control, whilst the
recursive procedure should be, by design, more sensitive to breaks.

Let%
\begin{equation}
\widehat{\beta }_{m}=\left( \sum_{i=1}^{m}x_{i}x_{i}^{\prime }\right)
^{-1}\sum_{i=1}^{m}x_{i}y_{i},  \label{beta}
\end{equation}%
where dependence on the sample size $m$ will be omitted whenever possible,
and define the residuals%
\begin{equation}
\widehat{\epsilon }_{i}=y_{i}-\widehat{\beta }_{m}^{\prime }x_{i}=\epsilon
_{i}+\left( \beta -\widehat{\beta }_{m}\right) ^{\prime }x_{i},
\label{residual}
\end{equation}%
for $m+1\leq i\leq T$ onwards. At each $k$, we define the cumulative process%
\begin{equation}
Q\left( m;k\right) =\left\vert \frac{1}{\widehat{\sigma }_{\epsilon }^{2}}%
\sum_{i=m+1}^{m+k}\widehat{\epsilon }_{i}^{2}\right\vert ,  \label{q-2}
\end{equation}%
for $1\leq k\leq T_{m}$.

\subsubsection{Estimation of $\protect\sigma _{\protect\epsilon }^{2}$\label%
{lr}}

In (\ref{q-2}), $\widehat{\sigma }_{\epsilon }^{2}$ is an estimator of $%
\sigma _{\epsilon }^{2}$. In our paper, we use a weighted-sum-of-covariance
estimator. In order to apply our theory, we need to show the almost sure
convergence of $\widehat{\sigma }_{\epsilon }^{2}$ to a positive limit;
thus, this section of our paper can be compared to \citet{berkesbartlett}.

Let $\rho _{l}^{\left( \epsilon \right) }$ denote the $l$-th order
autocovariance of $\epsilon _{i}$, i.e. $\rho _{l}^{\left( \epsilon \right)
}=E\left( \epsilon _{i}\epsilon _{i-l}\right) $. This can be estimated as%
\begin{equation}
\widehat{\rho }_{l}^{\left( \epsilon \right) }=\frac{1}{m}\sum_{i=l+1}^{m}%
\widehat{\epsilon }_{i}\widehat{\epsilon }_{i-l}.  \label{rho-e}
\end{equation}%
Based on (\ref{rho-e}), we define%
\begin{equation}
\widehat{\sigma }_{\epsilon }^{2}=\widehat{\rho }_{0}^{\left( \epsilon
\right) }+2\sum_{l=1}^{H}\left( 1-\frac{l}{H+1}\right) \widehat{\rho }%
_{l}^{\left( \epsilon \right) }.  \label{sig-e-hat}
\end{equation}%

Let $y_{i,l}^{\left( \epsilon \right) }=\epsilon _{i}\epsilon _{i-l}-\rho
_{l}^{\left( \epsilon \right) }$. We need the following regularity conditions

\begin{assumption}
\label{lrv}It holds that: (i) $\epsilon _{i}$ is covariance stationary with $%
E\left\vert \epsilon _{i}\right\vert ^{4}<\infty $ for all $i$; (ii) $%
\sum_{l=0}^{\infty }l\left\vert \rho _{l}^{\left( \epsilon \right)
}\right\vert <\infty $; (iii) $E\left\vert \sum_{i=l+1}^{m}y_{i,l}^{\left(
\epsilon \right) }\right\vert ^{2}\leq c_{0}m $.
\end{assumption}

It holds that

\begin{proposition}
\label{long-run}We assume that Assumptions \ref{as-2}-\ref{lrv} are
satisfied. As $\min \left( m,H\right) \rightarrow \infty $%
\begin{equation}
\widehat{\sigma }_{\epsilon }^{2}=\sigma _{\epsilon }^{2}+o_{a.s.}\left( 
\frac{H}{m^{1/2}}\left( \ln m\right) ^{3+\varepsilon }\left( \ln \ln
m\right) \left( \ln H\right) ^{2+\varepsilon }\right) +O\left( \frac{1}{H}%
\right) ,  \label{sig-e}
\end{equation}%
for every $\varepsilon >0$.
\end{proposition}

In Proposition \ref{long-run}, a crucial role is played by the bandwidth $H$%
. In order to ensure consistency, (\ref{sig-e}) requires that $H\rightarrow \infty$ and 
\begin{equation*}
\frac{H}{m^{1/2}}\left( \ln m\right)
^{3+\varepsilon }\left( \ln \ln m\right) \left( \ln H\right) ^{2+\varepsilon
}\rightarrow 0,
\end{equation*}%
as $m\rightarrow \infty $.

\subsection{The monitoring scheme\label{sequential}}

The main idea underpinning (\ref{q-2}) is that, by construction, $Q\left(
m;k\right) $ should pick up the presence of a break, which would introduce a
drift in its fluctuations. In order to check whether $Q\left( m;k\right) $
is growing \textquotedblleft naturally\textquotedblright , i.e. without
breaks, or not, we introduce the function%
\begin{equation}
g\left( m;k\right) =\left[ \left( m+k\right) +\left( \frac{m+k}{m}\right)
^{2}\right] ^{1+\gamma },  \label{g-1}
\end{equation}%
where the choice of $\gamma $ depends on the length of the monitoring
horizon. Heuristically, the function $g\left( m;k\right) $ should control
the growth rate of $Q\left( m;k\right) $: this is driven by a term
proportional to the cumulative sum of $\epsilon _{i}^{2}$ - which is
controlled by $m+k$ in (\ref{g-1}) - and one proportional to the cumulative
sum of $x_{i}^{2}$, multiplied by the (square of the) estimation error $%
\beta -\widehat{\beta }_{m}$ - which is controlled by the term $\left( \frac{%
m+k}{m}\right) ^{2}$ in (\ref{g-1}).

\begin{assumption}
\label{horizon}It holds that: (i) $T_{m}=c_{0}m^{\theta }$ for
some $\theta >1$ and $0 < c_{0} < \infty$; (ii) if $k^{\ast } < T$, $k^{\ast }=O\left( m^{\theta
^{\prime }}\right) $ with $0\leq \theta ^{\prime }<\theta $; (iii) $\lim
\inf_{m\rightarrow \infty }\frac{T_{m}}{m}>0$.
\end{assumption}

Assumption \ref{horizon} states that the monitoring horizon should go on for
a sufficiently long time (part \textit{(i)}), and obviously include the
breakdate if there is a break (part \textit{(ii)}). In particular, part 
\textit{(i)}, with its implications, is very similar to equation (1.12) in \citet{lajos07}, who also consider the case where monitoring goes on for an infinite time (unless a change is detected). 

In practice, $\theta $ is also a given parameter, which is calculated from
Assumption \ref{horizon}\textit{(i)}, once $m$ and $T_{m}$ have been set.
Hence, $\gamma $ is calculated according to the rule 
\begin{equation}
\gamma =\frac{1-\delta }{\theta -1},  \label{gamma}
\end{equation}%
where $\delta $ is chosen as $0<\delta <1$. In principle, any value of $%
\delta $ will ensure the validity of the theory below. In essence, $\gamma$ is chosen as a fraction of $\frac{1}{\theta-1}$; clearly, choosing $\delta$ close to $1$ yields a small $\gamma$, which in turn makes the divergence of $g\left( m;k\right) $ as $m \rightarrow \infty$ slower than in the case of a $\delta$ closer to zero. We discuss the practical impact of the choice of $\delta $ (and $\gamma $) on the ability of the
monitoring procedure to detect breaks in Section \ref{local-alt}. \newline
The function $g\left( m;k\right) $ has been chosen so as to distinguish the
growth rate that $Q\left( m;k\right) $ should have if there were no break,
from the rate at which it would diverge if there were a break.
Heuristically, in absence of breaks, $Q\left( m;k\right) $ should grow, but
slower than $g\left( m;k\right) $; on the other hand, if there is a break,
its presence in the residuals $\widehat{\epsilon }_{i}$ should make $Q\left(
m;k\right) $ grow at a faster pace, and faster than $g\left( m;k\right) $
itself. We point out that the term $\left( m+k\right) $ in (\ref{g-1}) is a
rather coarse estimate, and in principle it could be refined; however, this
term is anyway dominated by the second component of $g\left( m;k\right) $,
and (\ref{g-1}) yields very good results in simulations.

Define 
\begin{equation}
\psi _{m,k}=\frac{Q\left( m;k\right) }{g\left( m;k\right) }.  \label{psi}
\end{equation}%
Based on the above, we expect that $\psi _{m,k}$ drifts to zero as $m$ and $%
T_{m}$ diverge if there is no break, whereas it should explode if there is a
break; note that we only consider rates. Indeed, in order to separate such
rates even better, we use the transformation%
\begin{equation}
\widetilde{\psi }_{m,k}=\exp \left( \frac{1}{\psi _{m,k}}\right) -1.
\label{psi-1}
\end{equation}%
By construction, $\widetilde{\psi }_{m,k}$ has the opposite behaviour as $%
\psi _{m,k}$: it can be expected that $\widetilde{\psi }_{m,k}$ drifts to
zero in the presence of a break (that is, under the alternative);
conversely, it should diverge to positive infinity if there is no break
(that is, under the null). Indeed, in the Appendix, we prove that, as $%
m\rightarrow \infty $%
\begin{align*}
& P\left\{ \omega :\widetilde{\psi }_{m,k}=\infty \right\} =1\text{ under }%
H_{0}, \\
& P\left\{ \omega :\widetilde{\psi }_{m,k}=0\right\} = 1\text{ under }H_{A}.
\end{align*}

Given that the test statistic $\widetilde{\psi }_{m,k}$ does not converge to
a non-degenerate random variable under the null (or the alternative), we
propose to use a randomised version of $\widetilde{\psi }_{m,k}$. We present
this as an algorithm, whose output will be a sequence of \textit{i.i.d.}
random variables, with a known distribution (at least asymptotically) under $%
H_{0}$, and which diverge under $H_{A,1}$ and $H_{A,2}$.

\begin{description}
\item[\textbf{Step 1}] For each $k$, generate an \textit{i.i.d. }$N\left(
0,1\right) $\textit{\ }sequence $\left\{ \xi _{j}^{\left( k\right) },1\leq
j\leq R\right\} $.

\item[\textbf{Step 2}] Generate the Bernoulli sequence $\zeta _{j}^{\left(
k\right) }\left( u\right) =I\left( \left\vert \widetilde{\psi }%
_{m,k}\right\vert ^{1/2}\xi _{j}^{\left( k\right) }\leq u\right) $.

\item[\textbf{Step 3}] Compute%
\begin{equation}
\vartheta _{m,R}^{\left( k\right) }\left( u\right) =\frac{2}{R^{1/2}}%
\sum_{j=1}^{R}\left( \zeta _{j}^{\left( k\right) }\left( u\right) -\frac{1}{2%
}\right) .  \label{theta-minor}
\end{equation}

\item[\textbf{Step 4}] Define%
\begin{equation}
\Theta _{m,R}^{\left( k\right) }=\int_{-\infty }^{+\infty }\left\vert
\vartheta _{m,R}^{\left( k\right) }\left( u\right) \right\vert ^{2}dF\left(
u\right),  \label{theta-maior}
\end{equation}
where $F\left( u\right) $ is a distribution.
\end{description}

Some comments on the sequence $\left\{ \Theta _{m,R}^{\left( k\right)
},1\leq k\leq T_{m}\right\} $ are in order. Consider first the case of the
null of no break. The Bernoulli random variable $\zeta _{j}^{\left( k\right)
}\left( u\right) $ should - asymptotically - be equal to $1$ or $0$ with
probability $\frac{1}{2}$, and thus have mean $\frac{1}{2}$. In this case,
when constructing $\vartheta _{m,R}^{\left( k\right) }\left( u\right) $, a
Central Limit Theorem holds and therefore we expect $\Theta _{m,R}^{\left(
k\right) }$ to have a chi-square distribution. On the other hand, under the
alternative of a break, $\zeta _{j}^{\left( k\right) }\left( u\right) $
should be (heuristically) $0$ or $1$ with probability $0$ or $1$ (depending
on the sign of $u$) - thus its mean is not $\frac{1}{2}$, and a Law of Large
Numbers should hold. Note finally that, by construction, conditionally on
the sample, the sequence $\{\Theta _{m,R}^{\left( k\right) }\}_{k=1}^{T_{m}}$
is independent across $k$; also, by integrating out $u$ in Step 4, the
statistic $\Theta _{m,R}^{\left( k\right) }$ becomes invariant to the choice
of this specification.

The following regularity conditions are needed:

\begin{assumption}
\label{regularity}It holds that: $F\left( u\right)$ is a non-degenerate continuous distribution with (i) $\int_{-\infty }^{+\infty
}u^{2}dF\left( u\right) <\infty $; (ii) the sequences $\left\{ \xi
_{j}^{\left( k\right) },1\leq j\leq R\right\} $ are independent across $k$.
\end{assumption}

Let now $P^{\ast }$ represent the conditional probability with respect to $%
\left\{ u_{i},\epsilon _{i},1\leq i\leq T\right\} $; we use the notation
\textquotedblleft $\overset{D^{\ast }}{\rightarrow }$\textquotedblright\ and
\textquotedblleft $\overset{P^{\ast }}{\rightarrow }$\textquotedblright\ to
define, respectively, conditional convergence in distribution and in
probability according to $P^{\ast }$. It holds that

\begin{theorem}
\label{theta-1}We assume that Assumptions \ref{as-2}-\ref{regularity} hold.
As $\min \left( m,R\right) \rightarrow \infty $ with%
\begin{equation}
R\exp \left( -m^{\gamma }\right) \rightarrow 0,  \label{restriction}
\end{equation}%
under $H_{0}$\ it holds that, for each $1\leq k\leq T_{m}$ 
\begin{equation*}
\Theta _{m,R}^{\left( k\right) }\overset{D^{\ast }}{\rightarrow }\chi
_{1}^{2},
\end{equation*}%
for almost all realisations of $\left\{ u_{i},\epsilon _{i},1\leq i\leq
T\right\} $.
\end{theorem}

\begin{theorem}
\label{theta-2}We assume that Assumptions \ref{as-2}-\ref{regularity} hold.
As $\min \left( m,R\right) \rightarrow \infty $, under $H_{A,1}\cup H_{A,2}$%
\ it holds that, for each $k\geq \left\lfloor m^{\max \left\{ 1,\theta
^{\prime }\right\} \left( 1+\varepsilon \right) }\right\rfloor $ for all $%
\varepsilon >0$ 
\begin{equation*}
\frac{1}{R}\Theta _{m,R}^{\left( k\right) }\overset{P^{\ast }}{\rightarrow }%
1,
\end{equation*}%
for almost all realisations of $\left\{ u_{i},\epsilon _{i},1\leq i\leq
T\right\} $.
\end{theorem}

Theorems \ref{theta-1} and \ref{theta-2} are intermediate results. Theorem %
\ref{theta-1} stipulates that under the null $\Theta _{m,R}^{\left( k\right)
}$ has, asymptotically, a $\chi _{1}^{2}$ distribution; this result is of
independent interest, and we will make use of it to show that $\Theta
_{m,R}^{\left( k\right) }$ has finite moments of order $2+\varepsilon $ with 
$\varepsilon >0$. Note that the only thing that is required is the fact that 
$\Theta _{m,R}^{\left( k\right) }$ is an \textit{i.i.d.} sequence, with
finite moments of order $2+\varepsilon $: this is the building block on
which we can construct a detector whose properties can be studied
analytically. In this respect, any other transformation of $\vartheta
_{m,R}^{\left( k\right) }\left( u\right) $ (e.g., the absolute value, or a
power thereof) will also work, giving exactly the same results as in Theorem %
\ref{monitoring} below; the only advantage of defining $\Theta
_{m,R}^{\left( k\right) }$ as in Step 4 above is that its asymptotics has
already been studied (see e.g. \citealp{HT16}).

The theorems contain a restriction on the relative rate of divergence of the
pre-monitoring sample size $m$ and the artificial sample size $R$. Heuristically, note that our monitoring procedure is based on having a bounded sequence with finite moments under the null. As the proof of Theorem \ref{theta-1} shows, under the null the statistic $\Theta _{m,R}^{\left( k\right)
}$ has a non-centrality term which vanishes as long as (\ref{restriction}) is satisfied. Conversely, under the alternative it is required that $\Theta _{m,R}^{\left( k\right)
}$ should pass to infinity: Theorem \ref{theta-2} ensures that this occurs at a rate equal to $R$. Thus, Theorem \ref{theta-2} and (\ref{restriction}) provide a family of selection rules for $R$. Given that we only need convergence and divergence, the role played by $R$ can be expected to be rather marginal, which is also confirmed by our simulations (see
Section \ref{montecarlo}). However, we note that the choice $R=m$ satisfies (\ref%
{restriction}). Theorem \ref{theta-2}, conversely, states that, under the
alternative where there is a break at $k^{\ast }$, $\Theta _{m,R}^{\left(
k\right) }$ diverges to positive infinity after $k^{\ast }$. 

In light of these results, we build a monitoring function, based on the use
of the cumulative sums process. Define the detectors%
\begin{equation}
d\left( m;k\right) =\left\vert \sum_{i=m+1}^{m+k}\frac{\Theta _{m,R}^{\left(
i\right) }-1}{\sqrt{2}}\right\vert ,\text{ }1\leq k\leq T_{m}.
\label{detector}
\end{equation}%
As can be noted, $d\left( m;k\right) $ is the CUSUM\ process of $\left\{
\Theta _{m,R}^{\left( k\right) },1\leq k\leq T_{m}\right\} $, after
centering and standardizing. \\

Similarly to the literature on structural breaks (see e.g. \citealp{csorgo1997}), we now need to define a family of threshold functions such that if the CUSUM process exceeds the threshold, a change is detected. A standard choice (see \citealp{chu}) is
\begin{align}
& \nu \left( m;k\right) =c_{\alpha ,m}\nu ^{\ast }\left( m;k\right) ,
\label{nu-1-1} \\
& \nu ^{\ast }\left( m;k\right) =m^{1/2}\left( 1+\frac{k}{m}\right) .  \label{nu-2-1}
\end{align}%
Based on this choice, the FLCT yields that, for every $x$
\begin{align}
& P^{\ast }\left[ \max_{1\leq k\leq T_{m}}\frac{d\left( m;k\right) }{\nu
^{\ast }\left( m;k\right) }\leq x\right] \rightarrow P\left[ \sup_{0\leq
t\leq 1} \left\vert B\left( t\right) \right\vert \leq x%
\right] ,
\end{align}%
where $B$ is a standard Brownian motion. The limiting law of this expression involves a Brownian motion; intuitively, this being a heteroskedastic process, this procedure may not be the most powerful one; this is further corroborated by \citet{aue2004}, who show that the delay in detecting a changepoint increases as $\eta \rightarrow 0$. A possibility would be to re-scale the monitoring function as suggested in \citet{lajos04} and %
\citet{lajos07}, viz. using %
\begin{align}
& \nu \left( m;k\right) =c_{\alpha ,m}\nu ^{\ast }\left( m;k\right) ,
\label{nu-1} \\
& \nu ^{\ast }\left( m;k\right) =m^{1/2}\left( 1+\frac{k}{m}\right) \left( 
\frac{k}{m+k}\right) ^{\eta },  \label{nu-2}
\end{align}%
with $\eta \in \left[ 0,\frac{1}{2}\right] $, and $c_{\alpha ,m}$\ a
critical value. Intuitively, the difference with (\ref{nu-1-1}) is that the monitoring function is now smaller than before, which should ensure higher power. From a technical point of view,
however, choosing $\eta >0$ entails having to use a different asymptotics,
based on \textit{almost sure} as opposed to \textit{weak} convergence. The
fact that the building blocks of $d\left( m;k\right) $ are the $\Theta
_{m,R}^{\left( k\right) }$s - which are, conditional on the sample, \textit{%
i.i.d.} and with finite moments - entails that that we can use an array of
almost sure results (see the book by \citealp{csorgo1997}), which in turn
makes it possible to carry out the monitoring using $\eta >0$ in (\ref{nu-2}%
).\\
We point out that, despite the considerations above, the choice of threshold functions is by no
means unique, and, as \citet{chu} put it, \textquotedblleft often dictated by
mathematical convenience rather than optimality\textquotedblright ; in our
case, we have defined $\nu ^{\ast }\left( m;k\right) $ as per (\ref{nu-2})
give that the calculation of crossing probabilities (made according to (\ref%
{eta-1})) is tractable - see \citet{lajos04} and \citet{lajos07}. We then define the
stopping rule as%
\begin{equation}
\widehat{k}_{m}=\inf \left\{ 1\leq k\leq T_{m}\text{ s.t. }d\left(
m;k\right) \geq \nu \left( m;k\right) \right\} ,  \label{stopping}
\end{equation}%
setting $\widehat{k}_{m}=T_{m}$ when (\ref{stopping}) does not hold for $%
1\leq k\leq T_{m}$.

The critical value $c_{\alpha ,m}$ is defined, for a given level $\alpha $,
as%
\begin{align}
& P\left[ \sup_{0\leq t\leq 1}\frac{\left\vert B\left( t\right) \right\vert 
}{t^{\eta }}\leq c_{\alpha ,m}\right] =1-\alpha ,\text{ for }\eta <\frac{1}{2%
},  \label{eta-1} \\
& c_{\alpha ,m}=\frac{D_{m}-\ln \left( -\ln \left( 1-\alpha \right) \right) 
}{A_{m}},\text{ for }\eta =\frac{1}{2};  \label{eta-2}
\end{align}%
in (\ref{eta-1}), $\left\{ B\left( t\right) ,-\infty <t<\infty \right\} $ is
a standard Brownian motion, whereas in (\ref{eta-2}) we have defined%
\begin{equation}
A_{m}=\left( 2\ln \ln m\right) ^{1/2}\text{ and }D_{m}=2\ln \ln m+\frac{1}{2}%
\ln \ln \ln m-\frac{1}{2}\ln \pi ;  \label{am-dm}
\end{equation}%
critical values for (\ref{eta-1}) - which do not depend on $m$ - can be found in Table 1 in \citet{lajos04}.

We need the following assumption, which restricts (\ref{restriction}) and
strengthens Assumption \ref{regularity}.

\begin{assumption}
\label{restrict-2}It holds that: (i)%
\begin{equation*}
m^{1/2+\tau }R\exp \left( -m^{\gamma }\right) \rightarrow 0,
\end{equation*}%
as $\min \left( m,R\right) \rightarrow \infty $, for $\tau >0$; (ii) $%
\int_{-\infty }^{+\infty }u^{4+\tau }dF\left( u\right) <\infty $, for some $%
\tau >0$.
\end{assumption}
It holds that

\begin{theorem}
\label{monitoring}We assume that Assumptions \ref{as-2}-\ref{restrict-2} are
satisfied.

As $\min \left( m,R\right) \rightarrow \infty $ with (\ref{restriction}),
under $H_{0}$ it holds that%
\begin{align}
& P^{\ast }\left[ \max_{1\leq k\leq T_{m}}\frac{d\left( m;k\right) }{\nu
^{\ast }\left( m;k\right) }\leq x\right] \rightarrow P\left[ \sup_{0\leq
t\leq 1}\frac{\left\vert B\left( t\right) \right\vert }{t^{\eta }}\leq x%
\right] \text{ for }\eta <\frac{1}{2},  \label{null-asy-1} \\
& P^{\ast }\left[ \max_{1\leq k\leq T_{m}}\frac{d\left( m;k\right) }{\nu
^{\ast }\left( m;k\right) }\leq \frac{x+D_{m}}{A_{m}}\right] \rightarrow
\exp \left( -\exp \left( -x\right) \right) \text{ for }\eta =\frac{1}{2},
\label{null-asy-2}
\end{align}%
for $-\infty <x<\infty $ and almost all realisations of $\left\{
u_{i},\epsilon _{i},1\leq i\leq T\right\} $.

As $\min \left( m,R\right) \rightarrow \infty $, under $H_{A,1}\cup H_{A,2}$
it holds that%
\begin{equation}
\max_{1\leq k\leq T_{m}}\frac{d\left( m;k\right) }{\nu ^{\ast }\left(
m;k\right) }\overset{P^{\ast }}{\rightarrow }\infty ,\text{ for any }\eta
\in \left[ 0,\frac{1}{2}\right] ,  \label{power}
\end{equation}%
for almost all realisations of $\left\{ u_{i},\epsilon _{i},1\leq i\leq
T\right\} $.
\end{theorem}

Theorem \ref{monitoring} implies the following

\begin{corollary}
\label{corollary}Under the assumptions of Theorem \ref{monitoring} it holds
that: 
\begin{align}
& \lim_{\min (m,R)\rightarrow \infty }P^{\ast }\left( \widehat{k}%
_{m}<T_{m}\right) \leq \alpha ,\text{ \ \ under }H_{0},  \label{size} \\
& \lim_{\min (m,R)\rightarrow \infty }P^{\ast }\left( \widehat{k}%
_{m}<T_{m}\right) =1,\text{ \ \ under }H_{A,1}\cup H_{A,2},  \label{pow}
\end{align}%
for almost all realisations of $\left\{ u_{i},\epsilon _{i},1\leq i\leq
T\right\} $.
\end{corollary}

\section{Discussion and extensions\label{extensions}}

In this section, we investigate two aspects of the monitoring procedure
proposed above. Firstly, we examine the impact
of various test specifications on the power of our procedure (Section \ref%
{local-alt}); secondly, we consider the presence of deterministics in (\ref%
{model-1}) (Section \ref{deterministics}).

\subsection{Power versus shrinking alternatives and the impact of $g\left( m;k\right)$\label{local-alt}}

Our proposed methodology depends on several specifications in the
construction of the monitoring function, and in the algorithm to compute the
test statistic. In this section, we discuss the impact of such
specifications on the power of the monitoring procedure. In particular, in
this section we discuss the impact of $\gamma $ in the threshold function $%
g\left( m;k\right) $ on power versus shrinking alternatives. In Section \ref%
{montecarlo}, we also comment on the choices of $u$ and its distribution.

In order to discuss the impact of $\gamma $, we focus on a simple set-up
where there are no deterministics, viz. on model (\ref{model-1})%
\begin{equation*}
y_{i}=\beta ^{\prime }x_{i}+\epsilon _{i},
\end{equation*}%
and we consider the presence of power versus the local-to-null set-ups%
\begin{eqnarray}
H^{\ast}_{A,1} &:&\beta _{i}=\beta +\Delta _{\beta }\left( m\right) I\left[
i>k^{\ast }\right] ,  \label{ha-1-local} \\
H^{\ast}_{A,2} &:&\epsilon _{i}=\sigma _{\upsilon }\left( m\right) \upsilon
_{i}+u_{i}^{\epsilon }\text{ for }k^{\ast }+1\leq i\leq T.
\label{ha-2-local}
\end{eqnarray}%
In (\ref{ha-2-local}), we assume%
\begin{equation*}
\upsilon _{i}=\upsilon _{i-1}+u_{i}^{\upsilon },
\end{equation*}%
with $u_{i}^{\epsilon }$ independent of $u_{i}^{\upsilon }$. Specifically, in (\ref{ha-1-local}), we consider a
shrinking break where $\Delta _{\beta }\left( m\right) \rightarrow 0$,
whereas in (\ref{ha-2-local}), inspired by \citet{busetti}, we model the
local-to-null case as having $\sigma _{\upsilon }\left( m\right) \rightarrow
0$. Note that we consider, in both equations, the break as shrinking with $m$, since this can be viewed as the sample size on which estimation is based. 

Heuristically, as the proof of Theorem \ref{monitoring} shows, in order for
the monitoring procedure to detect changes, it is necessary that $\psi
_{m,k}\rightarrow \infty $ a.s.; thus, intuitively, $\Delta _{\beta }\left(
m\right) $ and $\sigma _{\upsilon }\left( m\right) $, as they drift to zero,
must be \textquotedblleft slow enough\textquotedblright\ to ensure that $%
\psi _{m,k}$ diverges. We formalise this in the following theorem

\begin{theorem}
\label{drift}We assume that Assumptions \ref{as-2}-\ref{restrict-2} are
satisfied. Then, under $H^{\ast}_{A,1}$, equation (\ref{power}) holds as $%
m\rightarrow \infty $ as long as%
\begin{eqnarray}
m^{\delta -\varepsilon }\Delta _{\beta }\left( m\right) &\rightarrow &\infty
,\text{ when }\theta \leq 2,  \label{delta-drift-1} \\
m^{\theta \frac{\theta -2+\delta }{2\left( \theta -1\right) }}\Delta _{\beta
}\left( m\right) &\rightarrow &\infty ,\text{ when }\theta >2,
\label{delta-drift-2}
\end{eqnarray}%
for some $\varepsilon >0$. Under $H^{\ast}_{A,2}$, equation (\ref{power}) holds as $%
m\rightarrow \infty $ as long as%
\begin{eqnarray}
m^{\delta -\varepsilon }\sigma _{\upsilon }\left( m\right) &\rightarrow
&\infty ,\text{ when }\theta \leq 2,  \label{sigma-drift-1} \\
m^{\theta \frac{\theta -2+\delta }{2\left( \theta -1\right) }}\sigma
_{\upsilon }\left( m\right) &\rightarrow &\infty ,\text{
when }\theta >2.  \label{sigma-drift-2}
\end{eqnarray}
\end{theorem}

Theorem \ref{drift}, together with (\ref{gamma}), illustrates what happens
to the power of the monitoring procedure depending on the value of $\gamma $. As can be expected in light of the definition of $g\left( m;k\right) $, choosing a \textquotedblleft small\textquotedblright\ $\gamma$ (which corresponds to choosing $\delta$ close to $1$) enhances the power of the procedure, which, conversely, declines as $\gamma$ increases. This can be understood by noting that the noncentrality of $Q\left( m;k\right) $ is divided by $g\left( m;k\right) $ too. When $\theta \leq 2$, the procedure could potentially (depending on $\delta$) be able to detect breaks as small as $O\left( \frac{1}{m^{1-\epsilon}}\right)$, with $\epsilon > 0$ arbitrarily small. When $\theta > 2$ - that is, when monitoring goes on for a very long time - the ability to detect a small break increases. \\
Note that an alternative could have been to express the break as shrinking with the whole (calibration plus monitoring) sample size, $T$, as done in \citet{wied}. 

\subsection{The monitoring procedure in the presence of deterministics \label{deterministics}}

In this section, we consider the following extension of (\ref{model-1})%
\begin{equation}
y_{i}=\mu _{0}+\mu _{1}i+\beta ^{\prime }x_{i}+\epsilon _{i}.
\label{model-2}
\end{equation}%
Equation (\ref{model-2}) contains, with respect to the previous model, a
constant and a deterministic trend; other extensions could of course be
possible. Our hypothesis testing framework is the same as in the previous
section, namely we test for%
\begin{equation*}
H_{0}:\left\{ 
\begin{array}{l}
\beta _{i}=\beta \\ 
\epsilon _{i}\text{ is stationary}%
\end{array}%
\right. \text{ for }1\leq i\leq T_{m},
\end{equation*}%
versus the two alternatives%
\begin{eqnarray*}
H_{A,1} &:&\beta _{i}=\beta +\Delta _{\beta }I\left[ i>k^{\ast }\right] , \\
H_{A,2} &:&\epsilon _{i}=\epsilon _{i-1}+u_{i}^{\epsilon }\text{ for }%
k^{\ast }+1\leq i\leq T.
\end{eqnarray*}%
Note that, for the sake of a concise discussion, we do not consider changes
in $\mu _{0}$ or $\mu _{1}$, although again this would be perfectly possible in
principle.

Our monitoring scheme can be adapted as follows. As is typical in this case,
we propose to demean and detrend both $y_{i}$ and $x_{i}$, by estimating 
\begin{eqnarray*}
y_{i} &=&a_{0}+a_{1}i+u_{i}^{y}, \\
x_{i} &=&b_{0}+b_{1}i+u_{i}^{x},
\end{eqnarray*}%
using OLS, and then computing%
\begin{equation}
\widehat{\beta }_{m}^{d}=\left( \sum_{i=1}^{m}\widehat{u}_{i}^{x}\widehat{u}%
_{i}^{x\prime }\right) ^{-1}\sum_{i=1}^{m}\widehat{u}_{i}^{x}\widehat{u}%
_{i}^{y},  \label{beta-hat-detrend}
\end{equation}%
where $\widehat{u}_{i}^{x}$ and $\widehat{u}_{i}^{y}$ are the OLS residuals from the regressions above. After defining%
\begin{equation*}
\widetilde{\epsilon }_{i}=y_{i}-\widehat{\beta }_{m}^{d\prime }x_{i},
\end{equation*}%
for $m+1\leq i\leq T$, we use the recursively detrended residuals%
\begin{equation}
\widehat{\epsilon }_{i}^{d}=\widetilde{\epsilon }_{i}-\left( \widehat{\mu }%
_{0,i}+\widehat{\mu }_{1,i}i\right) ,  \label{res-detrend}
\end{equation}%
where%
\begin{equation*}
\left( 
\begin{array}{c}
\widehat{\mu }_{0,i} \\ 
\widehat{\mu }_{1,i}%
\end{array}%
\right) =\left[ \sum_{j=1}^{i}\left( 
\begin{array}{cc}
1 & j \\ 
j & j^{2}%
\end{array}%
\right) \right] ^{-1}\sum_{j=1}^{i}\left( 
\begin{array}{c}
\widetilde{\epsilon }_{j} \\ 
j\widetilde{\epsilon }_{j}%
\end{array}%
\right) .
\end{equation*}%

Note that other detrending schemes could be proposed also; for example, in
the construction of $\widehat{\epsilon }_{i}^{d}$, one could use
non-recursive estimates $\widehat{\mu }_{0}$ and $\widehat{\mu }_{1}$
(computed once and for all using the sample $1\leq i\leq m$). 
\newline
We then define, as in (\ref{q-2}), the cumulative process 
\begin{equation}
Q^{d}\left( m;k\right) =\left\vert \frac{1}{\widetilde{\sigma }_{\epsilon
}^{2}}\sum_{i=m+1}^{m+k}\left( \widehat{\epsilon }_{i}^{d}\right)
^{2}\right\vert ,  \label{q-detrend}
\end{equation}%
where the long-run variance estimator $\widetilde{\sigma }_{\epsilon }^{2}$
is computed exactly as in (\ref{sig-e-hat}), using $\widehat{\epsilon }%
_{i}^{d}$. We need the following assumption, which complements Assumption %
\ref{as-2}\textit{(i)}.

\begin{assumption}
\label{as-detrend}It holds that: (i) $E\left\Vert \sum_{i=1}^{t}i\epsilon
_{i}\right\Vert ^{2}\leq c_{0}t^{3}$, for all $1\leq t\leq T$, and (ii) $%
E\left\Vert \sum_{i=1}^{t}ix_{i}\right\Vert ^{2}\leq c_{0}t^{5}$, for all $%
1\leq t\leq T$.
\end{assumption}

The next theorem shows that, when using $Q^{d}\left( m;k\right) $ instead of 
$Q\left( m;k\right) $ in constructing the monitoring procedure, the same
results hold.

\begin{theorem}
\label{detrending}We assume that Assumptions \ref{as-2}-\ref{as-detrend} are
satisfied. Then, when constructing $d\left( m;k\right) $ using $Q^{d}\left(
m;k\right) $, (\ref{null-asy-1})-(\ref{power}) hold.
\end{theorem}

\section{Numerical and empirical evidence\label{montecarlo}}

In this section, we illustrate the properties of our procedure through a
Monte Carlo exercise (Section \ref{simulations}), and through an application
to US housing market data (Section \ref{empirics}).

\subsection{Simulations\label{simulations}}

We consider the DGP in (\ref{model-1}), with the addition of a constant term, and with $p=1$, namely%
\begin{equation*}
y_{i}= \mu_{0} + \beta x_{i}+\epsilon _{i},\text{ with }x_{i}=\sum_{j=1}^{i}u_{j};
\end{equation*}%
to evaluate the finite sample performance of our proposed procedure. As discussed in section \ref{deterministics}, we can allow for a constant term in the DGP through recursive demeaning of the data. Incorporating a constant in this fashion allows us to directly compare our new procedure to the equivalent constant-only version of that proposed by \cite{wied}. Noting that our demeaned procedure is mean-invariant, we set $\mu_{0}=0$. We set $\beta =1$ for $1\leq i\leq m$, although unreported experiments show that, as can be expected, this value has no impact on the results. 

Innovations $\left\{ \epsilon _{i},u_{i}\right\} $ have been generated as
follows%
\begin{align}
u_{i}& =\rho ^{\left( x\right) }u_{i-1}+v_{i}^{u},  \label{ar-u} \\
\epsilon _{i}& =\left( \frac{1+\left( \rho ^{\left( x\epsilon \right)
}\right) ^{2}Var\left( v_{i}^{u}\right) }{1-\left( \rho ^{\left( \epsilon
\right) }\right) ^{2}}\right) ^{-1/2}\epsilon _{i}^{\ast },  \label{s-t-n} \\
\epsilon _{i}^{\ast }& =\rho ^{\left( \epsilon \right) }\epsilon
_{i-1}^{\ast }+v_{i}^{e}+\rho ^{\left( x\epsilon \right) }v_{i}^{u},
\label{ar-e}
\end{align}%
In (\ref{ar-u}), we allow for $AR\left( 1\right) $ dynamics in $u_{i}$,
setting $\rho ^{\left( x\right) }\in \left\{ 0,0.5\right\} $. We have also
experimented with other values, noting that results hardly change. In order
to control for the signal-to-noise ratio, we have generated the
idiosyncratic innovation $v_{i}^{u}$ as \textit{i.i.d.} $N\left( 0,\sigma
_{u}^{2}\right) $; by (\ref{s-t-n}). This entails that the signal-to-noise
ratio is exactly equal to $\sigma _{u}^{2}$, and we have used $\sigma
_{u}^{2}=2$ in our experiments. In unreported simulations, we considered $%
\sigma_{u}^{2}=1$, with qualitatively similar results. As far as (\ref{ar-e}%
) is concerned, we have generated $v_{i}^{e}$ as \textit{i.i.d.} $N\left(
0,1\right) $. Serial dependence in the error term $\epsilon _{i}$ is
explicitly allowed for through $\rho ^{\left( \epsilon \right) }$; note that
when $\rho ^{\left( \epsilon \right) }=1$, this corresponds to $H_{A,2}$,
that is, (\ref{model-1}) becomes a non-cointegrating regression. We report
results for $\rho ^{\left( \epsilon \right) }\in \left\{ 0,0.5,0.9\right\} $. In
(\ref{ar-e}), we also consider the possible presence of endogeneity through
the coefficient $\rho ^{\left( x\epsilon \right) }$, using $\rho ^{\left(
x\epsilon \right) }\in \left\{ 0,0.5\right\} $. The long-run variance of $%
\epsilon _{i}$\ is estimated as in (\ref{sig-e-hat}), setting $%
H=\left\lfloor m^{1/6}\right\rfloor $, where $\left\lfloor \cdot
\right\rfloor $ denotes the integer part.

As far as the other specifications of the experiment are concerned, we
report results for $T\in \left\{ 100,200,400\right\} $ and $m\in \left\{ 
\frac{T}{4},\frac{T}{2}\right\} $. When considering the presence of a break,
we have set the changepoint $k^{\ast }=m+\frac{T}{4}$. Experimenting with
other breakdates does not change results in any remarkable way. Under $%
H_{A,1}$, we have set $\beta _{i}=\beta +\Delta _{\beta }I\left[ i > k^{\ast
}\right] $, with $\Delta _{\beta }\in \left\{ 0.5,1\right\} $. In addition
to reporting empirical rejection frequencies under the alternative, we also
report the delay in changepoint detection, defined as%
\begin{equation}
delay=\frac{\widehat{k}_{m}-k^{\ast }}{k^{\ast }}.  \label{delay}
\end{equation}

We now turn to describing the implementation of the test and of the
randomisation algorithm. As far as the former is concerned, we have computed 
$g\left( m;k\right) $ setting, according to (\ref{g-1}), $\gamma =0.45$.
Results are similar, especially for large $m$, when using $\gamma =0.4$ and $%
\gamma =0.5$. 
In the randomisation algorithm, based on (\ref{restriction}), we set $R=m$;
altering this specification (which we have tried in some unreported
experiments) is virtually inconsequential on the empirical rejection
frequencies under both $H_{0}$ and $H_{A,1}\cup H_{A,2}$. Finally, we discuss the choice of $u$. Extracting $u$ from a
distribution - as we make explicit in Step 4 of our algorithm - has the
advantage that $u$ gets integrated out in the construction of $\Theta
_{m,R}^{\left( k\right) }$, making this invariant to the support of $u$
itself. In this respect, choosing $F\left( u\right) $ as the standard normal distribution is a possibility, which is very easy to implement. Indeed, in order to construct $\Theta _{m,R}^{\left( k\right) }$ practically, we can use a Gauss-Hermite quadrature to approximate the integral that defines it,
viz.%
\begin{equation}
\Theta _{m,R}^{\left( k\right) }=\frac{1}{\sqrt{\pi }}%
\sum_{s=1}^{n_{S}}w_{s}\left( \vartheta _{m,R}^{\left( k\right) }\left( \sqrt{2}%
z_{s}\right) \right) ^{2},  \label{upsilon-feasible}
\end{equation}%
where the $z_{s}$s, $1\leq s\leq n_{S}$, are the zeros of the Hermite
polynomial $H_{n_{S}}\left( z\right) $ and the weights $w_{s}$ are defined
as 
\begin{equation}
w_{s}=\frac{\sqrt{\pi }2^{n_{S}-1}\left( n_{S}-1\right) !}{n_{S}\left[
H_{n_{S}-1}\left( z_{s}\right) \right] ^{2}}.  \label{hermite-weights}
\end{equation}%
Thus, when constructing $\theta _{m,R}^{\left( k\right) }\left( u\right) $,
we construct $n_{S}$ of these statistics, each with $u=\sqrt{2}z_{s}$; the
values of the roots $z_{s}$, and of the corresponding weights $w_{s}$, are
tabulated e.g. in \citet{salzer}. in our case, we have used $n_{S}=2$, so
that $u=\pm 1$ with equal weight $\frac{1}{2}$; we note that in unreported
experiments we tried $n_{S}=4$ with the corresponding weights, but there
were no changes up to the $4$-th decimal in the empirical rejection
frequencies.  

We report results for $\eta=\{0,0.45,0.49,0.50\}$ for the
threshold function in (\ref{nu-2}). We offer a direct comparison of our procedure to that of \cite{wied}. We focus our attention on the IM-OLS version of their test in our simulations as the authors state a preference for IM-OLS, relative to the FM-OLS and D-OLS approaches that they also consider, on the basis of its finite sample performance. We denote this procedure $WW$--$IM$ in what follows. The nominal significance level $\alpha$, in the computation of critical values defined in (\ref{eta-1}) and (\ref{eta-2}), has been set as $\alpha =0.05$. Finally, all experiments have
been carried out using $1,000$ replications.

Empirical rejection frequencies under $H_{0}$ are reported in Table %
\ref{tab:Table1}. 
We point out that, in our context, the notion of size (control) differs from the standard
Neyman-Pearson testing paradigm. In the latter, empirical rejection
frequencies are expected to be close to their nominal level. In the context
of a sequential testing procedure like ours, as pointed out by %
\citet{lajos07} (see also the comments in Ch. 9 in \citealp{Sen}), the
primary purpose is to keep the false detection rate \textit{below} the
chosen level $\alpha$. Indeed, the proportion of false discoveries should go
to zero, since the monitoring can continue for an infinite amount of time.
In this respect, whilst this is the case for all four values of $\eta$
considered in our procedure, setting $\eta=0$ yields the best results, with empirical
rejection frequencies approaching zero across many settings of $m$ and $T$. Examining Panel A of the table, the case of no serial dependence in the error terms, it is clear that for $T=100$, $m=25$ all test procedures exhibit empirical rejection frequencies above their nominal significance levels, with the degree of distortion higher for our procedures than for the $WW$--$IM$ procedure. However, for $T=100$, $m=50$, whilst $WW$--$IM$ offer slightly inflated empirical rejection frequencies (between 0.052 and 0.080 depending on the value of $\rho ^{\left( x \epsilon \right) }$ and $\rho ^{\left( x  \right) }$), our test procedure offers empirical rejection frequencies lower than the nominal significance level. For $T=200$ and $T=400$ we observe a similar pattern of results to that under $T=100$, $m=50$; note that the $WW$--$IM$ procedure exhibits a higher degree of upwards distortion (with frequencies up to 0.080 observed).  

Allowing for \textit{AR(1)} dynamics in $u_{i}$ has a small upwards effect on the empirical rejection frequencies of all tests for $T=100$ and $m=25$. For all other combinations of $m$ and $T$, little effect is observed for our test procedures, with no effect at all in the case of $T=400$, whereas the $WW$--$IM$ test is observed to be a little more sensitive to these \textit{AR(1)} dynamics. Allowing for endogeneity results in modest increases in the empirical rejection frequencies for our procedures, for most settings of $m$ and $T$, whereas it has the opposite effect for the the $WW$--$IM$ test, resulting in modest decreases in size relative to the no endogeneity case. 

Considering Panel B, serial dependence in $\epsilon _{i}$ of $\rho ^{\left( \epsilon \right) }=0.5$ results in upwards size distortion for $WW$--$IM$ for all settings of $m$, $T$, $\rho ^{\left( x \epsilon \right) }$ and $\rho ^{\left( x \right) }$. In contrast, with the exception of $T=100$, $m=25$, where all procedures exhibit empirical rejection frequencies somewhat higher than the nominal significance level, the size of our procedures are robust to this degree of serial correlation, with only very small differences observed in the empirical rejection frequencies from the no serial dependence case (no larger than 0.006 for the settings considered here). This is a pleasing result, given that serial dependence is likely in practice. Finally, turning our attention to Panel C, the case of high serial dependence, $\rho ^{\left( \epsilon \right) }=0.9$, we find that the $WW$--$IM$ procedure exhibits substantial over-sizing for all combinations of $m$, $T$, $\rho ^{\left( x \epsilon \right) }$ and $\rho ^{\left( x \right) }$. As before, with the exception of $T=100$, $m=25$, our procedures display less size distortion than those of $WW$--$IM$.  When examining the performance of our procedures, we notice here the role that $m$ plays, with smaller empirical rejection frequencies observed for $m=\frac{T}{2}$ than for $m=\frac{T}{4}$, for a given $T$; and with empirical rejection frequencies decreasing as $T$ increases for a given setting of $m$. 

Empirical rejection frequencies and the associated detection delays under $%
H_{A,1}$ (i.e. under a change in $\beta$ in the cointegrating regression)
are reported in Tables \ref{tab:Table2a} and \ref{tab:Table2b} respectively.
Considering first the rejection frequencies in Table \ref{tab:Table2a}, in the case of no serial dependence in the errors (Panel A), it is clear that our monitoring procedures offer excellent power, with $\eta=\{0.45,0.49,0.5\}$ outperforming $WW$--$IM$ over most combinations of $\Delta _{\beta }$, $T$, $m$, $\rho ^{\left( x \right) }$, and $\rho ^{\left( x \epsilon \right) }$. There are only 8 instances out of the 48 different combinations of settings in Panel A where $WW$--$IM$ achieves higher power, all cases where $m=\frac{T}{2}$. This result is somewhat anticipated, given that our test allows for $T_{m} \rightarrow \infty$, assumes that the monitoring horizon should go on for a sufficiently long time, and is thus expected to perform better for small $m$ relative to $T_{m}$, whereas \citet{wied} choose $m$ to be large relative to $T$, as discussed in section \ref{sequential}. It is pleasing however, that even in the case of $m=\frac{T}{2}$, our procedures outperform $WW$--$IM$ in terms of power in the majority of cases. Despite its small empirical rejection frequencies under the null, our procedure with $\eta=0$ also performs very well in terms of power, with rejection frequencies under $H_{A,1}$ very similar to those of $\eta=\{0.45,0.49,0.5\}$ in most cases. In addition to the effects of $\rho ^{\left( x \right) }$, and $\rho ^{\left( x \epsilon \right) }$, our procedures appear to be robust to serial dependence in the errors, with high levels of
power maintained under different settings of $\rho ^{\left( \epsilon \right)
}$. Comparing the four values of $\eta$ that we consider here, no setting uniformly outperforms the others,
with little difference in rejection frequencies observed between these values. 

Turning our attention to the detection delays reported in Table \ref%
{tab:Table2b}, we observe that increasing $m$, increasing $T$, and
increasing $\Delta _{\beta }$ all contribute towards reducing the detection
delay, as we might expect. Contrary to the empirical rejection frequencies,
when considering detection delays a ranking does emerge amongst the
different values of $\eta$ for our procedure, with $\eta=0$ resulting in a longer detection
delay relative to the other settings. This detection delay can be seen as a
trade-off for the very small null rejection frequencies exhibited in Table \ref{tab:Table1}. Setting $%
\eta=\{0.45,0.49,0.5\}$ produces the shortest detection delay across
the various settings of $m$, $T$ and $\Delta _{\beta }$ considered here, with very little to distinguish between these settings. Our procedure is capable of detecting a break in the parameters of the
cointegrating regression shortly after the break occurs, as little as 2.7
observations on average after the break for the case of $T=400$, $m=200$ and 
$\Delta _{\beta }=1$, where $\rho ^{\left(\epsilon \right)}=0$, $\rho ^{\left(x \epsilon \right)}=0$, $\rho ^{\left(x\right)}=0.5$ and $\eta=0.49$. The $WW$--$IM$ procedure incurs a longer detection delay than our test procedure for every setting considered here.

Finally, empirical rejection frequencies and detection delays under $H_{A,2}$
(i.e. under a switch from a cointegrating to a non-cointegrating regression)
are given in Tables \ref{tab:Table3a} and \ref{tab:Table3b} respectively.
Considering first the rejection frequencies in Table \ref{tab:Table3a}, we note that our procedure is able to offer good levels of power against this alternative hypothesis for most settings. Relative to our results for $H_{A,1}$ in Table \ref{tab:Table2a}, increasing $m$ has a more severe effect on the empirical rejection frequencies, particularly for smaller values of $T$. Examining Panel A, our procedure outperforms $WW$--$IM$ in the majority of cases. Exceptions occur in some instances where $m=\frac{T}{2}$, $\rho ^{\left(x \epsilon \right)}=0.5$.

Comparing Panels A with Panels B and C, it is clear that serial correlation in the errors has the effect of reducing the empirical rejection frequencies for all tests, with a higher degree of serial correlation corresponding to a lower rejection frequency. Of course, this result is to be anticipated given the nature of the alternative hypothesis. In general, with serial correlation of $\rho ^{\left(\epsilon \right)}\{=0.5,0.9\}$, our procedure performs better for $m=\frac{T}{4}$ and $WW$--$IM$ performs better for $m=\frac{T}{2}$, although we note that the empirical rejection frequencies reported here are not size-adjusted, and given the degree of over-sizing exhibited by especially $WW$--$IM$ in Table \ref{tab:Table1}, it is hard to directly compare the tests' performance.

Considering the detection delays under $H_{A,2}$ in Table \ref{tab:Table3b}, we again observe that the delay decreases as $m$ and $T$ increase. We note that delay detections are generally longer under $H_{A,2}$ than for equivalent settings under $H_{A,1}$. As with the detection delays under $H_{A,1}$, when considering our procedure, setting $\eta=0$ provides the longest delay in detection, with $\eta=\{0.45,0.49,0.5\}$ providing the quickest detection of a break. $WW$--$IM$ exhibits a longer detection delay than our procedure with $\eta=\{0.45,0.49,0.5\}$ across all settings, except in one instance\footnote{$T=400$, $m=200$, $\rho ^{\left(\epsilon \right)}=0$, $\rho ^{\left(x \epsilon \right)}=0$ and $\rho ^{\left(x \right)}=0.5$ } where it is still outperformed by $\eta=\{0.45,0.5\}$.

When considering detection delays, it is possible that a test detects a break prematurely, which would lead to a negative delay for that replication according to (4.4), which in turn could result in a misleadingly low reported average delay in Tables \ref{tab:Table2b} and \ref{tab:Table3b}. To further examine the estimated break dates found by these procedures, and to verify whether premature detection is of concern here, in Figures \ref{histHA1} and \ref{histHA2} we consider histograms of the estimated break dates found by our procedure using $\eta=0$ and $\eta=0.45$, as well as the $WW$--$IM$ procedure. For simplicity, we consider the case of $\rho ^{\left(x \right)}=0$, $\rho ^{\left(\epsilon \right)}=0$ and $\rho ^{\left(x \epsilon \right)}=0$. Figure \ref{histHA1} displays estimated break dates under $H_{A,1}$, with \ref{HA1_200} considering $T=200$, $m=\frac{T}{4}$ and $\Delta_{\beta}=1$. It is clear that our procedure using either $\eta=0$ or $\eta=0.45$ provides more accurate break date estimation than $WW$--$IM$, with only a very small difference between these settings of $\eta$. A premature break date is found in only a handful of replications, in the case of $\eta=0.45$, suggesting that early detection is not a significant problem for our test.  In Figure \ref{histHA2} we set $T=400$, $m=\frac{T}{2}$ and $\Delta_{\beta}=0.5$, a more challenging circumstance for our procedure as it is designed for small $m$ relative to $T_{m}$. Nevertheless, our procedure displays more accuracy than that of $WW$-$IM$ here. 

Figure \ref{histHA2} displays estimated break dates under $H_{A,2}$, with \ref{HA2_200} and \ref{HA2_400} considering the same settings of $T$ and $m$ as in \ref{HA1_200} and \ref{HA1_400} respectively. Again, we are able to note the accuracy of our procedure relative to $WW$--$IM$, and the relatively small numbers of replications where a break is detected before the true break date, $k^{*}$.

Although $\eta=0$ provides the lowest null empirical rejection frequencies, we argue that a sequential monitoring test based on $\eta=0.45$ provides the best overall performance given that it maintains a null empirical rejection frequency below the nominal significance level across most settings of $m$ and $T$, as well as providing the shortest detection delays under both $%
H_{A,1}$ and $H_{A,2}$ (although we note that there is very little difference in performance between $\eta=\{0.45,0.49,0.5\}$).

\subsection{Empirical application\label{empirics}}

To demonstrate the practical relevance of the procedure developed in section %
\ref{theory}, and inspired by the empirical work of \citet{anundsen} and %
\citet{wied}, we investigate the possibility that the US housing market
experienced a structural break in cointegration. Based on the life-cycle
model of housing under the assumption of no arbitrage for the housing
market, \citet{anundsen} analyses two fundamentals-driven cointegrating
relationships. The first approach, known as the \textit{price-to-rent}
model, relies on the user cost of a property being equal to the cost of
renting a property of similar quality in equilibrium, and is given by: 
\begin{equation}  \label{pricerent}
ph_{t} = \gamma_{r}r_{t} + \gamma_{UC}UC_{t} + u_{t}
\end{equation}
where $ph_{t}$ is the logarithm of real housing prices at period $t$, $r_{t}$
is the logarithm of real rents, and $UC_{t}$ is the real direct user cost of
housing, computed as 
\begin{equation*}
UC_{t} = (1-\tau^{y}_{t})(i_{t}+\tau^{p}_{t}) - \pi_{t} + \delta_{t},
\end{equation*}
where $\tau^{y}_{t}$ is the marginal personal income tax rate (measured here
at twice the median income), $\tau^{p}_{t}$ is the marginal tax rate on
personal property, $i_{t}$ is the nominal interest rate, $\pi_{t}$ is
overall price inflation, and $\delta_{t}$ is the housing depreciation rate.

The second approach, known as the \textit{inverted demand} model, assumes
that imputed rent is a function of income and housing stock, and is given by
the below equation: 
\begin{equation}
ph_{t}=\widetilde{\gamma }_{y}y_{t}+\widetilde{\gamma }_{h}h_{t}+\widetilde{%
\gamma }_{UC}UC_{t} + \widetilde{u}_{t},  \label{invdemand}
\end{equation}%
where $y_{t}$ is the logarithm of real per capita disposable income and $%
h_{t}$ is the logarithm of the per capita housing stock.

Assuming that the variables in (\ref{pricerent}) and (\ref{invdemand}) are $%
I(1)$, economic theory predicts that $u_{t}$ and $\widetilde{u}_{t}$ are
both $I(0)$. That is, two cointegrating relationships exist between housing
prices and their fundamentals. A breakdown of these cointegrating
regressions therefore indicates that housing prices are no longer being
driven by these fundamentals. Following the definition of \cite{stiglitz}, 
\textit{inter alia}, that an asset bubble exists when its price no longer
appears to be justified by the value of its fundamental components, \cite%
{anundsen} interprets a breakdown in these cointegrating relationships as
evidence of a bubble in housing prices. Indeed, following \citet{anundsen},
we also allow for an intercept and linear trend term in each model, viz. 
\begin{equation}
ph_{t}=\theta _{1}+\theta _{2}t+\theta _{3}r_{t}+\theta _{4}UC_{t}+u_{t},
\label{ptr-2}
\end{equation}%
and 
\begin{equation}
ph_{t}=\widetilde{\theta }_{1}+\widetilde{\theta }_{2}t+\widetilde{\theta }%
_{3}y_{t}+\widetilde{\theta }_{4}h_{t}+\widetilde{\theta }_{4}UC_{t}+%
\widetilde{u}_{t}.  \label{id-2}
\end{equation}

\citet{anundsen} applies models (\ref{ptr-2})-(\ref{id-2}) to quarterly US
housing market data over the sample period 1976:Q1 - 2010:Q4. Specifically,
he estimates vector autoregression models and undertakes Johansen
cointegration testing for both the price-to-rent and inverted demand
equations using expanding sub-samples of the data, starting with an initial
sub-sample from 1976:Q1 - 1995:Q4 and then subsequently adding four new
observations until the full sample is used. This analysis finds evidence of
a cointegrating relationship in the housing market up until 2002 in the
price-to-rent model; evidence in favour of cointegration disappears when
2002:Q4 is included in the sample, but there is evidence of a return to a
cointegrating relationship towards the end of the sample. As far as the
inverted demand model is concerned, a similar pattern is found, with
cointegration breaking down in 2001:Q4. These results imply the emergence of
bubble behaviour in the housing market beginning in 2001-2002; however, %
\citet{wied} highlight that the analysis suffers from the problem of
multiple testing, leading to uncontrolled size. Considering the same
dataset, they apply their real time monitoring procedure to models (\ref%
{ptr-2})-(\ref{id-2}), with $WW$--$IM$ detecting a breakdown in cointegration at 2006:Q4 for the price-to-rent model and 2004:Q2 for the inverted demand model (with the FM-OLS version of the procedure finding a slightly earlier break of 2003:Q2 for this model). This delay in detection, relative to the results of \citet{anundsen}, can be viewed as a trade-off for asymptotic validity, and therefore controlled size.

We apply our sequential monitoring procedure to the dataset discussed above, containing information on US house prices from 1976:Q1 - 2010:Q4.\footnote{Detailed information on the dataset sources and construction are contained
in \cite{anundsen}. The dataset has been downloaded from the \textit{Journal
of Applied Econometrics} data archive.} In line with the previous two
studies, our calibration sample runs from 1976:Q1 -
1995:Q4, such that $m=80$; effectively, this means that \textquotedblleft future data\textquotedblright (and our monitoring) starts in 1996:Q1 (this being the earliest possible break date that we can detect). We set $R=m$ and $n_{s}=2$ as before. We allow for a constant and linear trend in both models through the recursive demeaning and detrending method discussed in section \ref{deterministics}. In view of our
simulation results in the previous section, we have used $\eta=0.45$.%
\footnote{%
Although, in unreported results, we note that setting $\eta=\{0,0.49,0.50\}$
provides the same break date estimate as $\eta=0.45$ for both models.} Similarly, based on the Monte Carlo evidence, we set $\gamma=0.4$, which should ensure size control, while decreasing the detection delay.

Figure \ref{empexgraphs} displays the residuals obtained from our estimation of the price to rent and inverted demand models, using recursive demeaning and detrending. From visual inspection, it is
clear that the residuals of both models undergo a period of mean-reverting
behaviour in the earlier part of the sample, whereas more persistent
behaviour in the residuals is observed from the early 2000s, lending support
to the hypothesis that a structural break in cointegration occurs during the
sample period. Considering first the price-to-rent model in Figure \ref{ptr}, our sequential monitoring procedure finds evidence of a break in
cointegration in 2005:Q3, 5 quarters earlier than the $WW$--$IM$ test, although
still somewhat later than the detection date in the initial experiment of 
\cite{anundsen}. Examining the inverted demand model, in Figure \ref{id},
evidence is found of a break in cointegration in 2004:Q1, somewhat earlier than for the price-to-rent model, in line with the results of \cite{wied} and \cite{anundsen}. Thus, our results support the claim of a breakdown in fundamentals-driven cointegrating relationships in the US housing markets during the housing bubble of the 2000s. 

\section{Conclusions\label{conclusions}}

In this paper we have investigated the issue of monitoring a cointegrating
regression. Having stability as the null hypothesis, we develop a procedure
to detect changes in the regression coefficients and/or from
cointegration to non-cointegration. Our procedure is based
on using the cumulative sums of squared residuals; at each point in the
monitoring horizon, we randomise the cumulative sum process, thereby
obtaining an \textit{i.i.d.} sequence with finite moments of arbitrarily
high order. We then use the results in \citet{lajos04} and \citet{lajos07}
to construct a family of procedures which may be viewed as a complement to
the results in \citet{wied}.

We point out that, as well as deriving the aforementioned statistics, in
this paper we have proposed a general methodology to construct monitoring
schemes in the context of a cointegrating regression. The approach we
propose can be readily generalised to use other statistics (e.g., upon
calculating the relevant rates, even the KPSS type statistic employed in %
\citet{wied} could be randomised and used in our algorithm), or to other
hypothesis testing frameworks. As a leading example, \citet{wagner} consider
the very interesting case where (\ref{model-1}) is, to begin with, a
non-cointegrating regression with $\epsilon _{i}\sim I\left( 1\right) $, and
the purpose of monitoring is to verify whether (\ref{model-1}) becomes a
cointegrating regression, with $\epsilon _{i}\sim I\left( 0\right) $.
Although we leave this interesting research question for future study, we
point out that a monitoring scheme for this case could be readily developed.
Indeed, one could use exactly the same approach as we do, using 
\begin{equation*}
\widehat{\psi }_{m,k}=\exp \left( \psi _{m,k}\right) -1,
\end{equation*}%
instead of (\ref{psi-1}). This and others issues are under investigation by
the authors.

{\small {\setlength{\bibsep}{.2cm} 
\bibliographystyle{chicago}
\bibliography{biblio}
}}

\clearpage
\appendix
\setcounter{section}{0} \setcounter{subsection}{-1} \setcounter{equation}{0} %
\setcounter{lemma}{0} \renewcommand{\thelemma}{A.\arabic{lemma}} %
\renewcommand{\theequation}{A.\arabic{equation}}

\section{Technical Lemmas\label{lemmas}}

Similarly to the main paper, we present results and proofs for the
univariate case, i.e. for $p=1$.

\begin{lemma}
\label{stou}Consider a multi-index random variable $U_{i_{1},...,i_{h}}$,
with $1\leq i_{1}\leq S_{1}$, $1\leq i_{2}\leq S_{2}$, etc... Assume that%
\begin{equation}
\sum_{S_{1}}\cdot \cdot \sum_{S_{h}}\frac{1}{S_{1}\cdot ...\cdot S_{h}}%
P\left( \max_{1\leq i_{1}\leq S_{1},...,1\leq i_{h}\leq S_{h}}\left\vert
U_{i_{1},...,i_{h}}\right\vert >\epsilon L_{S_{1},...,S_{h}}\right) <\infty ,
\label{bc-1}
\end{equation}%
for some $\epsilon >0$ and a sequence $L_{S_{1},...,S_{h}}$ defined as%
\begin{equation*}
L_{S_{1},...,S_{h}}=S_{1}^{d_{1}}\cdot ...\cdot S_{h}^{d_{h}}l_{1}\left(
S_{1}\right) \cdot ...l_{h}\left( S_{h}\right) ,
\end{equation*}%
where $d_{1}$, $d_{2}$, etc. are non-negative numbers and $l_{1}\left( \cdot
\right) $, $l_{2}\left( \cdot \right) $, etc. are slowly varying functions
in the sense of Karamata. Then it holds that%
\begin{equation}
\lim \sup_{\left( S_{1},...,S_{h}\right) \rightarrow \infty }\frac{%
\left\vert U_{S_{1},...,S_{h}}\right\vert }{L_{S_{1},...,S_{h}}}=0\text{ 
\textit{a.s.}}  \label{bc-2}
\end{equation}
\begin{proof}
The lemma is shown in \citet{bt3} - see, in particular, Lemma B1 therein.
\end{proof}
\end{lemma}

\begin{lemma}
\label{lemma1}Under Assumption \ref{as-2}, it holds that there exist a
random variable $m_{0}$ and a constant $0<c_{0}<\infty $ such that, for all $%
m\geq m_{0}$%
\begin{equation*}
\sum_{i=1}^{m}x_{i}^{2}\geq c_{0}\frac{m^{2}}{\ln \ln m}.
\end{equation*}

\begin{proof}
The lemma is an immediate consequence of Assumption \ref{as-2}; indeed we have%
\begin{align*}
& \sum_{i=1}^{m}x_{i}^{2} =\sum_{i=1}^{m}W_{x}^{2}\left( i\right)
-2\sum_{i=1}^{m}W_{x}\left( i\right) \left( W_{x}\left( i\right)
-x_{i}\right) +\sum_{i=1}^{m}\left( W_{x}\left( i\right) -x_{i}\right) ^{2}
\\
& =I+II+III.
\end{align*}%
Consider first $II$; we have%
\begin{align*}
& \frac{\ln \ln m}{m^{2}}II \leq 2\frac{\ln \ln m}{m^{2}}\sup_{1\leq i\leq
m}\left\vert W_{x}\left( i\right) -x_{i}\right\vert \sum_{i=1}^{m}\left\vert
W_{x}\left( i\right) \right\vert \\
&\leq c_{0}\frac{\ln \ln m}{m^{2}}m^{1/2-\delta ^{\prime
}}\sum_{i=1}^{m}i^{1/2}\left( \ln \ln i\right) ^{1/2}=o_{a.s.}\left(
1\right) .
\end{align*}%
Similarly%
\begin{equation*}
\frac{\ln \ln m}{m^{2}}III\leq c_{0}\frac{\ln \ln m}{m^{2}}m\sup_{1\leq
i\leq m}\left\vert W_{x}\left( i\right) -x_{i}\right\vert
^{2}=o_{a.s.}\left( 1\right) .
\end{equation*}%
Finally, by the Law of the Iterated Logarithm (LIL henceforth) for functionals of Brownian
motions (see Example 2 in\ \citealp{donsker1977}) we have%
\begin{equation*}
\frac{\ln \ln m}{m^{2}}I\geq c_{0},\text{ a.s.;}
\end{equation*}%
putting all together, the desired result follows. 
\end{proof}
\end{lemma}

\begin{lemma}
\label{lemma2}Under Assumption \ref{as-2}, it holds that 
\begin{equation*}
\sum_{i=1}^{m}x_{i}\epsilon _{i} \leq o_{a.s.}\left( m\left( \ln m\right)
^{1+\varepsilon }\right) ,  \label{x-e-2}
\end{equation*}%
for every $\varepsilon >0$.

\begin{proof}
Given that, by Assumption \ref{as-2}\textit{(ii)}, $E\left( \sum_{i=1}^{m}x_{i}\epsilon
_{i}\right) ^{2}\leq c_{0}m^{2}$, using the results in \citet{serfling1970}
it follows that that 
\begin{equation*}
E\max_{1\leq i\leq m}\left\vert \sum_{i=1}^{m}x_{i}\epsilon _{i}\right\vert
^{2}\leq c_{0}m^{2}\left( \ln m\right) ^{2};
\end{equation*}%
The lemma now follows from Lemma \ref{stou} and the Markov inequality. 
\end{proof}
\end{lemma}

\begin{lemma}
\label{lemma3}Under Assumption \ref{as-2}, it holds that 
\begin{equation*}
\widehat{\beta }_{m}-\beta =o_{a.s.}\left( \frac{\left( \ln m\right)
^{1+\varepsilon }\left( \ln \ln m\right) }{m}\right) .
\end{equation*}
\begin{proof}
The lemma is an immediate consequence of Lemmas \ref{lemma1} and \ref{lemma2}.
\end{proof}
\end{lemma}

\begin{lemma}
\label{lemma5}Under Assumptions \ref{as-2}-\ref{lrv}, it holds that, under $%
H_{0}$ 
\begin{equation*}
Q\left( m;k\right) =o_{a.s.}\left( r_{m,k}\right) ,
\end{equation*}%
where%
\begin{align}
r_{m,k} =&\left( m+k\right) \left( \ln \left( m+k\right) \right)
^{2+\varepsilon }+\frac{m+k}{m}\left( \ln m\right) ^{1+\varepsilon }\left(
\ln \ln m\right) \left( \ln \ln \left( m+k\right) \right) ^{1+\varepsilon }
\label{r-mk} \\
&+ \left( \frac{m+k}{m}\right) ^{2}\left( \ln \ln m\right) ^{2}\left( \ln
\ln \left( m+k\right) \right) \left( \ln m\right) ^{2+\varepsilon }.  \notag
\end{align}
for every $\varepsilon >0$.

\begin{proof}
Consider 
\begin{equation*}
\widehat{\sigma }_{\epsilon }^{2}Q\left( m;k\right) =\left\vert
\sum_{i=m+1}^{m+k}\widehat{\epsilon }_{i}^{2}\right\vert ,
\end{equation*}%
and note that, by Proposition \ref{long-run} and Assumption \ref{as-2}%
\textit{(i)}(b), $\widehat{\sigma }_{\epsilon }^{2}>0$ a.s. Thus, the order
of magnitude of $Q\left( m;k\right) $ can be studied by estimating 
\begin{equation*}
\sum_{i=m+1}^{m+k}\widehat{\epsilon }_{i}^{2}=\sum_{i=m+1}^{m+k}\epsilon
_{i}^{2}+2\left( \beta -\widehat{\beta }_{m}\right)
\sum_{i=m+1}^{m+k}x_{i}\epsilon _{i}+\left( \widehat{\beta }_{m}-\beta
\right) ^{2}\sum_{i=m+1}^{m+k}x_{i}^{2}=I+II+III.
\end{equation*}%
Using Assumption \ref{as-2}\textit{(i)} we obtain the (non-sharp) bound%
\begin{equation*}
E\sum_{i=m+1}^{m+k}\epsilon _{i}^{2}\leq E\sum_{i=1}^{m+k}\epsilon
_{i}^{2}\leq c_{0}\left( m+k\right) ,
\end{equation*}%
so that by Lemma \ref{stou}%
\begin{equation}
I=o_{a.s.}\left( \left( m+k\right) \left( \ln \left( m+k\right) \right)
^{2+\varepsilon }\right) .  \label{esq}
\end{equation}%
\ We now turn to $II$. Consider%
\begin{equation*}
\left\vert \sum_{i=m+1}^{m+k}x_{i}\epsilon _{i}\right\vert \leq \left\vert
\sum_{i=1}^{m+k}x_{i}\epsilon _{i}\right\vert +\left\vert
\sum_{i=1}^{m+1}x_{i}\epsilon _{i}\right\vert =II_{a}+II_{b}.
\end{equation*}%
By Assumption \ref{as-2}\textit{(ii)}%
\begin{equation*}
E\left\vert \sum_{i=1}^{m+k}x_{i}\epsilon _{i}\right\vert ^{2}\leq
c_{0}\left( m+k\right) ^{2};
\end{equation*}%
then, by Theorem A in \citet{serfling1970}%
\begin{equation*}
E\max_{1\leq j\leq k}\left\vert \sum_{i=1}^{m+j}x_{i}\epsilon
_{i}\right\vert ^{2}\leq c_{0}\left( m+k\right) ^{2}\left( \ln \left(
m+k\right) \right) ^{2};
\end{equation*}%
using Lemma \ref{stou} and the Markov inequality, we finally obtain $%
II_{a}=o_{a.s.}\left( \left( m+k\right) \left( \ln \left( m+k\right) \right)
^{1+\varepsilon }\right) $, and, similarly, $II_{b}=o_{a.s.}\left( m\left(
\ln \left( m\right) \right) ^{1+\varepsilon }\right) $. Thus, the following
(non sharp) estimate%
\begin{equation}
\sum_{i=m+1}^{m+k}x_{i}\epsilon _{i}=o_{a.s.}\left( \left( m+k\right) \left(
\ln \left( m+k\right) \right) ^{1+\varepsilon }\right) ,  \label{cp}
\end{equation}%
holds for every $\varepsilon >0$. Using Lemma \ref{lemma3} we obtain%
\begin{equation*}
II=o_{a.s.}\left( \frac{m+k}{m}\left( \ln m\right) ^{1+\varepsilon }\left(
\ln \left( m+k\right) \right) ^{1+\varepsilon }\left( \ln \ln m\right)
\right) .
\end{equation*}%
Finally, as far as $III$ is concerned, Assumption \ref{as-2}\textit{(ii)}
and the LIL (see \citealp{donsker1977}) entail%
\begin{equation}
\sum_{i=m+1}^{m+k}x_{i}^{2}=O_{a.s.}\left( \left( m+k\right) ^{2}\ln \ln
\left( m+k\right) \right) ;  \label{bm}
\end{equation}%
combining this with Lemma \ref{lemma3}, we have%
\begin{equation*}
III=o_{a.s.}\left( \left( \frac{m+k}{m}\right) ^{2}\left( \ln m\right)
^{2+\varepsilon }\left( \ln \ln m\right) ^{2}\ln \ln \left( m+k\right)
\right) .
\end{equation*}

The desired result now follows from putting everything together.%
\end{proof}
\end{lemma}

\begin{lemma}
\label{lemma6}Under Assumptions \ref{as-2}-\ref{lrv}, it holds that as $%
m\rightarrow \infty $%
\begin{equation*}
\frac{Q\left( m;k\right) }{g\left( m;k\right) }\rightarrow \infty \text{
a.s.,}
\end{equation*}%
under $H_{A,1}\cup H_{A,2}$, for $k\geq \left\lfloor \left( m^{\max \left\{
1,\theta ^{\prime }\right\} }\right) ^{1+\varepsilon }\right\rfloor $, for
every $\varepsilon >0$.

\begin{proof}
We prove the result separately under $H_{A,1}$ and $H_{A,2}$, starting from
the former. Recall that, in this case, $\beta _{i}=\beta +\Delta_{\beta} I\left[
i>m+k^{\ast }\right] $. We have 
\begin{align}
& \sum_{i=m+1}^{m+k}\widehat{\epsilon }_{i}^{2}=\sum_{i=m+1}^{m+k}\left(
\epsilon _{i}-\left( \widehat{\beta }_{m}-\beta \right) x_{i}\right)
^{2}+\Delta _{\beta }^{2}\sum_{i=m+k^{\ast }+1}^{m+k}x_{i}^{2} \\
& +2\Delta _{\beta }\left( \beta -\widehat{\beta }_{m}\right)
\sum_{i=m+1}^{m+k}x_{i}\left( \epsilon _{i}+\left( \beta -\widehat{\beta }%
_{m}\right) x_{i}\right)   \notag \\
& =I+II+III.  \label{res-sq-2}
\end{align}%
By Lemma \ref{lemma5}, $I=o_{a.s.}\left( r_{m,k}\right) $, with $r_{m,k}$\
defined in (\ref{r-mk}). Turning to $II$, note that%
\begin{equation*}
II=\Delta _{\beta }^{2}\sum_{i=m+k^{\ast }+1}^{m+k}x_{i}^{2}=\Delta _{\beta
}^{2}\sum_{i=1}^{m+k}x_{i}^{2}-\Delta _{\beta }^{2}\sum_{i=1}^{m+k^{\ast
}}x_{i}^{2}=II_{a}+II_{b}.
\end{equation*}%
By Assumption \ref{as-2} and the LIL (\citealp{donsker1977}), it holds that
there is a random variable $m_{0}$ such that for $m\geq m_{0}$%
\begin{align}
II_{a} & \geq c_{0}\frac{\left( m+k\right) ^{2}}{\ln \ln \left( m+k\right) },
\label{lil-2a} \\
II_{b} & \leq c_{0}\left( m+k^{\ast }\right) ^{2}\ln \ln \left( m+k^{\ast
}\right) ,  \notag
\end{align}%
so that%
\begin{align*}
\frac{II_{a}}{II_{b}} \geq &\frac{\left( m+k\right) ^{2}}{\left( m+k^{\ast
}\right) ^{2}\ln \ln \left( m+k\right) \ln \ln \left( m+k^{\ast }\right) } \\
\geq &\frac{\left( m+k\right) ^{2}}{\left( m+k^{\ast }\right) ^{2}\left(
\ln \ln \left( m+k\right) \right) ^{2}} \\
\geq &c_{0}\frac{\left( m^{2\max \left\{ 1,\theta ^{\prime }\right\}
}\right) ^{\varepsilon }}{\left( \ln \ln \left( m\right) \right) ^{2}}%
\rightarrow \infty ,
\end{align*}%
and therefore the term that dominates is $II_{a}$. It is immediate to see
that by (\ref{lil-2a}) and the definition of $\gamma $%
\begin{align*}
\frac{II_{a}}{\left( m+k\right) ^{1+\gamma }\left( \ln \left( m+k\right)
\right) ^{\left( 2+\varepsilon \right) \left( 1+\gamma \right) }} 
\geq c_{0}\frac{\left( m+k\right) ^{2}}{\left( m+k\right) ^{\left(
1+\gamma \right) }\left( \ln \left( m+k\right) \right) ^{\left(
2+\varepsilon \right) \left( 1+\gamma \right) }\ln \ln \left( m+k\right) } 
\geq c_{0}\left( m+k\right) ^{1-\gamma -\varepsilon ^{\prime }},
\end{align*}%
and%
\begin{align*}
&\frac{II_{a}}{\left( \left( \ln \ln m\right) ^{2}\left( \ln \ln \left(
m+k\right) \right) \left( \ln m\right) ^{2+\varepsilon }\right) ^{1+\gamma }}%
\left( \frac{m}{m+k}\right) ^{2\left( 1+\gamma \right) } \\
\geq & c_{0}\frac{\left( m+k\right) ^{2}}{\left( \left( \ln \ln m\right)
^{2}\left( \ln \ln \left( m+k\right) \right) \left( \ln m\right)
^{2+\varepsilon }\right) ^{1+\gamma }\ln \ln \left( m+k\right) }\left( \frac{%
m}{m+k}\right) ^{2\left( 1+\gamma \right) }\\
\geq & c_{0}\frac{m^{2\left(
1+\gamma \right) -\varepsilon ^{\prime }}}{\left( m+k\right) ^{2\gamma }},
\end{align*}%
which entails that, as $m\rightarrow \infty $%
\begin{equation*}
\frac{II}{g\left( m;k\right) }\rightarrow \infty \text{ a.s.}
\end{equation*}%
Finally, consider $III$; we have%
\begin{equation*}
\frac{III}{2\Delta _{\beta }}=\left( \beta -\widehat{\beta }_{m}\right)
\sum_{i=m+1}^{m+k}x_{i}\epsilon _{i}+\left( \beta -\widehat{\beta }%
_{m}\right) ^{2}\sum_{i=m+1}^{m+k}x_{i}^{2}=III_{a}+III_{b}.
\end{equation*}%
Using Lemma \ref{lemma3}, (\ref{cp}) and noting that%
\begin{equation*}
\sum_{i=m+1}^{m+k}x_{i}^{2}\leq \sum_{i=1}^{m+k}x_{i}^{2}=O_{a.s.}\left(
\left( m+k\right) ^{2}\ln \ln \left( m+k\right) \right) ,
\end{equation*}%
it holds that%
\begin{equation*}
III=o_{a.s.}\left( \frac{m+k}{m}\left( \ln m\right) ^{1+\varepsilon }\left(
\ln \ln m\right) \left( \ln \left( m+k\right) \right) ^{1+\varepsilon
}+\left( \frac{m+k}{m}\right) ^{2}\ln \ln \left( m+k\right) \left( \ln
m\right) ^{2+\varepsilon }\left( \ln \ln m\right) ^{2}\right) ,
\end{equation*}%
which immediately entails%
\begin{equation*}
\frac{III}{g\left( m;k\right) }\rightarrow 0\text{ a.s.}
\end{equation*}%
Putting everything together, the desired result obtains.

Under $H_{A,2}$, recall (\ref{spurious}), and write 
\begin{equation*}
\sum_{i=m+1}^{m+k}\widehat{\epsilon }_{i}^{2}=\sum_{i=m+1}^{m+k}\epsilon
_{i}^{2}+2\left( \beta -\widehat{\beta }_{m}\right)
\sum_{i=m+1}^{m+k}x_{i}\epsilon _{i}+\left( \widehat{\beta }_{m}-\beta
\right) ^{2}\sum_{i=m+1}^{m+k}x_{i}^{2}=I+II+III.
\end{equation*}%
We know from the passages above that 
\begin{equation*}
III=o_{a.s.}\left( \frac{\left( m+k\right) ^{2}}{m}\left( \ln m\right)
^{1+\varepsilon }\left( \ln \ln m\right) \ln \ln \left( m+k\right) \right) .
\end{equation*}%
Turning to $II$, by Assumption \ref{ha-2}\textit{(ii)} and similar passages
as in the previous proofs, we have 
\begin{equation*}
E\left\vert \sum_{i=m+1}^{m+k}x_{i}\epsilon _{i}\right\vert \leq E\left\vert
\sum_{i=m+1}^{m+k^{\ast }}x_{i}\epsilon _{i}\right\vert +E\left\vert
\sum_{i=m+k^{\ast }+1}^{m+k}x_{i}\epsilon _{i}\right\vert \leq c_{0}\left(
m+k^{\ast }\right) +c_{1}k^{2},
\end{equation*}%
which, by the same arguments as in the above, yields the bound $%
II=o_{a.s.}\left( \frac{k^{2}}{m}\left( \ln k\right) ^{2+\varepsilon }\left(
\ln m\right) ^{1+\varepsilon }\left( \ln \ln m\right) \right) $. Finally, as
far as $I$ is concerned, note that, using the same arguments in the proof of
Lemma \ref{lemma1}, it can be shown that, by Assumption \ref{ha-2}\textit{(i)%
} and the LIL%
\begin{equation*}
\sum_{i=m+1}^{m+k}\epsilon _{i}^{2}=\sum_{i=m+k^{\ast }+1}^{m+k}\epsilon
_{i}^{2}-\sum_{i=m+1}^{m+k^{\ast }}\epsilon _{i}^{2}\geq c_{0}\frac{k^{2}}{%
\ln \ln k},
\end{equation*}%
for some $0<c_{0}<\infty $ and sufficiently large $k$. Thus, term $I$ is the
one that dominates. The desired result follows by the same passages as
above. 
\end{proof}
\end{lemma}

\begin{lemma}
\label{lemma7}We assume that Assumptions \ref{as-2}-\ref{restrict-2} hold.
Under $H_{0}$, $\left\{ \Theta _{m,R}^{\left( i\right) },1\leq i\leq
T_{m}\right\} $ is an i.i.d. sequence conditionally on the sample, with%
\begin{align}
&\max_{1\leq k\leq T_{m}}\sqrt{\frac{m}{k\left( m+k\right) }}\left\vert
\sum_{i=m+1}^{m+k}\left( E^{\ast }\Theta _{m,R}^{\left( i\right) }-1\right)
\right\vert =O\left( m^{-\varepsilon }\right) ,  \label{7-1} \\
& \max_{1\leq k\leq T_{m}}\sqrt{\frac{m}{k\left( m+k\right) }}\left\vert
\sum_{i=m+1}^{m+k}\left( V^{\ast }\Theta _{m,R}^{\left( i\right) }-2\right)
\right\vert =O\left( m^{-\varepsilon }\right) ,  \label{7-2} \\
& E^{\ast }\left\vert \Theta _{m,R}^{\left( i\right) }\right\vert
^{2+\varepsilon ^{\prime }} <\infty ,  \label{7-3}
\end{align}%
where $\varepsilon ,\varepsilon ^{\prime }>0$.

\begin{proof}
The sequence $\Theta _{m,R}^{\left( i\right) }$ is independent across $i$
(conditionally on the sample) by construction. We begin by showing (\ref{7-1}%
). Using the fact that $\xi _{j}^{\left( i\right) }$ is \textit{i.i.d.}
across $j$, it holds that%
\begin{equation*}
4E^{\ast }\Theta _{m,R}^{\left( i\right) }=4E^{\ast }\int_{-\infty
}^{+\infty }\left\vert \left( \zeta _{1}^{\left( i\right) }-\frac{1}{2}%
\right) \right\vert ^{2}dF\left( u\right) .
\end{equation*}%
By the same passages as in the proof of Theorem \ref{theta-1}, it follows
that%
\begin{equation}
\left\vert 4E^{\ast }\Theta _{m,R}^{\left( i\right) }-1\right\vert \leq
c_{0}\left( \left\vert \widetilde{\psi }_{m,i}\right\vert ^{-1}+R\left\vert 
\widetilde{\psi }_{m,i}\right\vert ^{-2}\right) .  \label{error-7-1}
\end{equation}%
Equation (\ref{psi-3}) entails that there exist a constant $0<c_{0}<\infty $
and a random variable $m_{0}$ such that, for $m\geq m_{0}$ so that $%
\left\vert \widetilde{\psi }_{m,i}\right\vert ^{-1}\leq c_{0}\exp \left(
-m^{-\gamma }\right) $%
\begin{align*}
& \max_{1\leq k\leq T_{m}}\sqrt{\frac{m}{k\left( m+k\right) }}\left\vert
\sum_{i=m+1}^{m+k}E^{\ast }\Theta _{m,R}^{\left( i\right) }-\frac{1}{4}%
\right\vert  \\
& \leq c_{0}\exp \left( -m^{\gamma }\right) \left( 1+R\exp \left( -m^{\gamma
}\right) \right) \max_{1\leq k\leq T_{m}}\frac{m^{1/2}k^{1/2}}{\left(
m+k\right) ^{1/2}} \\
& \leq c_{0}m^{1/2}\exp \left( -m^{\gamma }\right) \left( 1+R\exp \left(
-m^{\gamma }\right) \right) ,
\end{align*}%
whence (\ref{7-1}) follows from Assumption \ref{restrict-2}.

Turning to (\ref{7-2}), we define $\zeta _{j}^{\left( i\right) }\left(
0\right) =I\left( \left\vert \widetilde{\psi }_{m,i}\right\vert ^{1/2}\xi
_{j}^{\left( i\right) }\leq 0\right) $. Elementary calculations yield%
\begin{equation}
E^{\ast }\left\vert \int_{-\infty }^{+\infty }\left\vert
R^{-1/2}\sum_{j=1}^{R}\left( \zeta _{j}^{\left( i\right) }\left( 0\right) -%
\frac{1}{2}\right) \right\vert ^{2}dF\left( u\right) \right\vert ^{2}=\frac{3%
}{16}.  \label{second-moment}
\end{equation}%
We begin by showing%
\begin{eqnarray}
&&E^{\ast }\left\vert \int_{-\infty }^{+\infty }\left\vert
R^{-1/2}\sum_{j=1}^{R}\left( \zeta _{j}^{\left( i\right) }\left( u\right) -%
\frac{1}{2}\right) \right\vert ^{2}dF\left( u\right) \right\vert
^{2}-E^{\ast }\left\vert \int_{-\infty }^{+\infty }\left\vert
R^{-1/2}\sum_{j=1}^{R}\left( \zeta _{j}^{\left( i\right) }\left( 0\right) -%
\frac{1}{2}\right) \right\vert ^{2}dF\left( u\right) \right\vert ^{2}
\notag\\
&\leq &c_{0}\left( R^{-1/4}\left\vert \widetilde{\psi }_{m,i}\right\vert
^{-1/4}+\left\vert \widetilde{\psi }_{m,i}\right\vert ^{-1/2}\right) . 
\label{second-moment-1} 
\end{eqnarray}%
Let%
\begin{align*}
X& =\int_{-\infty }^{+\infty }\left\vert R^{-1/2}\sum_{j=1}^{R}\left( \zeta
_{j}^{\left( i\right) }\left( u\right) -\frac{1}{2}\right) \right\vert
^{2}dF\left( u\right) +\int_{-\infty }^{+\infty }\left\vert
R^{-1/2}\sum_{j=1}^{R}\left( \zeta _{j}^{\left( i\right) }\left( 0\right) -%
\frac{1}{2}\right) \right\vert ^{2}dF\left( u\right)  \\
Y& =\int_{-\infty }^{+\infty }\left\vert R^{-1/2}\sum_{j=1}^{R}\left( \zeta
_{j}^{\left( i\right) }\left( u\right) -\frac{1}{2}\right) \right\vert
^{2}dF\left( u\right) -\int_{-\infty }^{+\infty }\left\vert
R^{-1/2}\sum_{j=1}^{R}\left( \zeta _{j}^{\left( i\right) }\left( 0\right) -%
\frac{1}{2}\right) \right\vert ^{2}dF\left( u\right)  \\
& =\left( \int_{-\infty }^{+\infty }R^{-1/2}\sum_{j=1}^{R}\left( \zeta
_{j}^{\left( i\right) }\left( u\right) -\zeta _{j}^{\left( i\right) }\left(
0\right) \right) dF\left( u\right) \right) \times \int_{-\infty }^{+\infty
}R^{-1/2}\sum_{j=1}^{R}\left( \left( \zeta _{j}^{\left( i\right) }\left(
u\right) -\frac{1}{2}\right) +\left( \zeta _{j}^{\left( i\right) }\left(
0\right) -\frac{1}{2}\right) \right) dF\left( u\right)  \\
& =Y_{1}Y_{2}.
\end{align*}%
Using the Cauchy-Schwartz inequality, $E^{\ast }\left( XY\right) \leq \left(
E^{\ast }\left( X^{2}\right) \right) ^{1/2}\left( E^{\ast }\left(
Y_{1}^{4}\right) \right) ^{1/4}\left( E^{\ast }\left( Y_{2}^{4}\right)
\right) ^{1/4}$. Now, 
\begin{align*}
E^{\ast }\left( X^{2}\right) & \leq c_{0}E^{\ast }\left( \int_{-\infty
}^{+\infty }\left\vert R^{-1/2}\sum_{j=1}^{R}\left( \zeta _{j}^{\left(
i\right) }\left( u\right) -\frac{1}{2}\right) \right\vert ^{2}dF\left(
u\right) \right) ^{2}+c_{0}E^{\ast }\left( \int_{-\infty }^{+\infty
}\left\vert R^{-1/2}\sum_{j=1}^{R}\left( \zeta _{j}^{\left( i\right) }\left(
0\right) -\frac{1}{2}\right) \right\vert ^{2}dF\left( u\right) \right) ^{2}
\\
& \leq c_{0}E^{\ast }\int_{-\infty }^{+\infty }\left\vert
R^{-1/2}\sum_{j=1}^{R}\left( \zeta _{j}^{\left( i\right) }\left( u\right) -%
\frac{1}{2}\right) \right\vert ^{4}dF\left( u\right) +c_{0}E^{\ast
}\left\vert R^{-1/2}\sum_{j=1}^{R}\left( \zeta _{j}^{\left( i\right) }\left(
0\right) -\frac{1}{2}\right) \right\vert ^{4}.
\end{align*}%
Hence, standard arguments entail $E^{\ast }\left( X^{2}\right) <\infty $. By
similar passages, it can be shown that $E^{\ast }\left( Y_{2}^{4}\right)
<\infty $. Also%
\begin{align}
E^{\ast }\left( Y_{1}^{4}\right) & =E^{\ast }\left( \int_{-\infty }^{+\infty
}R^{-1/2}\sum_{j=1}^{R}\left( \zeta _{j}^{\left( i\right) }\left( u\right)
-\zeta _{j}^{\left( i\right) }\left( 0\right) \right) dF\left( u\right)
\right) ^{4}  \notag \\
& \leq c_{0}R^{-2}\sum_{j=1}^{R}\int_{-\infty }^{+\infty }E^{\ast
}\left\vert \zeta _{j}^{\left( i\right) }\left( u\right) -\zeta _{j}^{\left(
i\right) }\left( 0\right) \right\vert ^{4}dF\left( u\right)
+c_{0}R^{-2}\left( \sum_{j=1}^{R}\int_{-\infty }^{+\infty }E^{\ast
}\left\vert \left( \zeta _{j}^{\left( i\right) }\left( u\right) -\zeta
_{j}^{\left( i\right) }\left( 0\right) \right) \right\vert ^{2}dF\left(
u\right) \right) ^{2}  \label{rosenthal}
\end{align}%
having used convexity and Rosenthal's inequality. We have%
\begin{align*}
& R^{-2}\sum_{j=1}^{R}\int_{-\infty }^{+\infty }E^{\ast }\left\vert \zeta
_{j}^{\left( i\right) }\left( u\right) -\zeta _{j}^{\left( i\right) }\left(
0\right) \right\vert ^{4}dF\left( u\right)  \\
& \leq c_{0}R^{-1}\int_{-\infty }^{+\infty }E^{\ast }\left\vert \zeta
_{1}^{\left( i\right) }\left( u\right) -\zeta _{1}^{\left( i\right) }\left(
0\right) \right\vert dF\left( u\right)  \\
& \leq c_{0}R^{-1}\left\vert \widetilde{\psi }_{m,i}\right\vert
^{-1/2}\int_{-\infty }^{+\infty }\left\vert u\right\vert dF\left( u\right)
\leq c_{0}R^{-1}\left\vert \widetilde{\psi }_{m,i}\right\vert ^{-1/2},
\end{align*}%
and%
\begin{align*}
& \sum_{j=1}^{R}\int_{-\infty }^{+\infty }E^{\ast }\left\vert \left( \zeta
_{j}^{\left( i\right) }\left( u\right) -\zeta _{j}^{\left( i\right) }\left(
0\right) \right) \right\vert ^{2}dF\left( u\right)  \\
& = c_{0}R\int_{-\infty }^{+\infty }E^{\ast }\left\vert \zeta
_{1}^{\left( i\right) }\left( u\right) -\zeta _{1}^{\left( i\right) }\left(
0\right) \right\vert^{2} dF\left( u\right)  \\
& \leq c_{0}R\left\vert \widetilde{\psi }_{m,i}\right\vert
^{-1/2}\int_{-\infty }^{+\infty }\left\vert u\right\vert^{2} dF\left( u\right)
\leq c_{0}R\left\vert \widetilde{\psi }_{m,i}\right\vert ^{-1/2}.
\end{align*}%
Thus, using (\ref{rosenthal})%
\begin{equation*}
E^{\ast }\left( Y_{1}^{4}\right) \leq c_{0}\left( R^{-1}\left\vert 
\widetilde{\psi }_{m,i}\right\vert ^{-1/2}+\left\vert \widetilde{\psi }%
_{m,i}\right\vert ^{-1}\right) .
\end{equation*}%
Thus, combining the results above with (\ref{second-moment})%
\begin{equation*}
E^{\ast }\left\vert \int_{-\infty }^{+\infty }\left\vert
R^{-1/2}\sum_{j=1}^{R}\left( \zeta _{j}^{\left( i\right) }\left( u\right) -%
\frac{1}{2}\right) \right\vert ^{2}dF\left( u\right) \right\vert ^{2}=\frac{3%
}{16}+O_{P^{\ast }}\left( R^{-1/4}\left\vert \widetilde{\psi }%
_{m,i}\right\vert ^{-1/8}\right) +O_{P^{\ast }}\left( \left\vert \widetilde{%
\psi }_{m,i}\right\vert ^{-1/4}\right) .
\end{equation*}%
Putting all together, and using (\ref{error-7-1}), we have%
\begin{equation*}
V^{\ast }\Theta _{m,R}^{\left( i\right) }=16\left( \frac{3}{16}-\frac{1}{16}%
\right) +O_{P^{\ast }}\left( R^{-1/4}\left\vert \widetilde{\psi }%
_{m,i}\right\vert ^{-1/8}\right) +O_{P^{\ast }}\left( \left\vert \widetilde{%
\psi }_{m,i}\right\vert ^{-1/4}\right) ,
\end{equation*}%
whence the desired result follows. 

Finally, consider (\ref{7-3}). We need to show that 
\begin{equation*}
E^{\ast }\left\vert \int_{-\infty }^{+\infty }\left\vert
R^{-1/2}\sum_{j=1}^{R}\left( \zeta _{j}^{\left( i\right) }\left( u\right) -%
\frac{1}{2}\right) \right\vert ^{2}dF\left( u\right) \right\vert
^{2+\varepsilon ^{\prime }}<\infty .
\end{equation*}%
Note first that, by convexity%
\begin{equation*}
E^{\ast }\left\vert \int_{-\infty }^{+\infty }\left\vert
R^{-1/2}\sum_{j=1}^{R}\left( \zeta _{j}^{\left( i\right) }\left( u\right) -%
\frac{1}{2}\right) \right\vert ^{2}dF\left( u\right) \right\vert
^{2+\varepsilon ^{\prime }}\leq E^{\ast }\int_{-\infty }^{+\infty
}\left\vert R^{-1/2}\sum_{j=1}^{R}\left( \zeta _{j}^{\left( i\right) }\left(
u\right) -\frac{1}{2}\right) \right\vert ^{2\left( 2+\varepsilon ^{\prime
}\right) }dF\left( u\right) ;
\end{equation*}%
further%
\begin{align*}
& E^{\ast }\int_{-\infty }^{+\infty }\left\vert R^{-1/2}\sum_{j=1}^{R}\left(
\zeta _{j}^{\left( i\right) }\left( u\right) -\frac{1}{2}\right) \right\vert
^{2\left( 2+\varepsilon ^{\prime }\right) }dF\left( u\right)  \\
& \leq c_{0}E^{\ast }\int_{-\infty }^{+\infty }\left\vert
R^{-1/2}\sum_{j=1}^{R}\left( \zeta _{j}^{\left( i\right) }\left( u\right)
-G\left( u\left\vert \widetilde{\psi }_{m,i}\right\vert ^{-1/2}\right) \right)
\right\vert ^{2\left( 2+\varepsilon ^{\prime }\right) }dF\left( u\right)  \\
& +c_{0}\int_{-\infty }^{+\infty }\left\vert R^{1/2}\left( G\left(
u\left\vert \widetilde{\psi }_{m,i}\right\vert ^{-1/2}\right) -\frac{1}{2}%
\right) \right\vert ^{2\left( 2+\varepsilon ^{\prime }\right) }dF\left(
u\right) ,
\end{align*}%
by the C$_{r}$-inequality. Consider the first term. The sequence $\left\{
\zeta _{j}^{\left( i\right) },1\leq j\leq R\right\} $ is independent
(conditional on the sample); thus, by Burkholder's inequality and convexity,
we get%
\begin{align*}
& \int_{-\infty }^{+\infty }E^{\ast }\left\vert R^{-1/2}\sum_{j=1}^{R}\left(
\zeta _{j}^{\left( i\right) }\left( u\right) -G\left( u\left\vert \widetilde{%
\psi }_{m,i}\right\vert ^{-1/2}\right) \right) \right\vert ^{2\left(
2+\varepsilon ^{\prime }\right) }dF\left( u\right)  \\
& \leq \int_{-\infty }^{+\infty }E^{\ast }\left\vert
R^{-1}\sum_{j=1}^{R}\left( \zeta _{j}^{\left( i\right) }\left( u\right)
-G\left( u\left\vert \widetilde{\psi }_{m,i}\right\vert ^{-1/2}\right)
\right) ^{2}\right\vert ^{\left( 2+\varepsilon ^{\prime }\right) }dF\left(
u\right)  \\
& \leq R^{-1}\sum_{j=1}^{R}\int_{-\infty }^{+\infty }E^{\ast }\left\vert
\left( \zeta _{j}^{\left( i\right) }\left( u\right) -G\left( u\left\vert 
\widetilde{\psi }_{m,i}\right\vert ^{-1/2}\right) \right) ^{2}\right\vert
^{\left( 2+\varepsilon ^{\prime }\right) }dF\left( u\right) \leq c_{0}.
\end{align*}%
Also 
\begin{equation*}
R^{2+\varepsilon ^{\prime }}\int_{-\infty }^{+\infty }\left\vert \left(
G\left( u\left\vert \widetilde{\psi }_{m,i}\right\vert ^{-1/2}\right) -\frac{%
1}{2}\right) \right\vert ^{2\left( 2+\varepsilon ^{\prime }\right) }dF\left(
u\right) \leq c_{0}\left( R\left\vert \widetilde{\psi }_{m,i}\right\vert
^{-1}\right) ^{2+\varepsilon ^{\prime }}\int_{-\infty }^{+\infty }\left\vert
u\right\vert ^{2\left( 2+\varepsilon ^{\prime }\right) }dF\left( u\right) ,
\end{equation*}%
which vanishes on account of (\ref{restriction}) and Assumption \ref%
{restrict-2}\textit{(ii)}. Putting all together, (\ref{7-3}) obtains. 
\end{proof}
\end{lemma}

\begin{lemma}
\label{lemma8}We assume that Assumptions \ref{as-2} and \ref{as-detrend} are
satisfied. Then it holds that%
\begin{equation*}
\widehat{\beta }_{m}^{d}-\beta =o_{a.s.}\left( \frac{\left( \ln m\right)
^{3+\varepsilon }}{m}\right) ,
\end{equation*}%
for all $\varepsilon >0$, where $\widehat{\beta }_{m}^{d}$ is defined in (%
\ref{beta-hat-detrend}).

\begin{proof}
The proof follows similar passages to the proofs of Lemmas \ref{lemma2} and %
\ref{lemma3}. Note that, by the Frisch-Waugh-Lovell theorem, we can write%
\begin{equation}
\widehat{\beta }_{m}^{d}-\beta =\left[ \sum_{i=1}^{m}\left( \widehat{u}%
_{i}^{x}\right) ^{2}\right] ^{-1}\sum_{i=1}^{m}\widehat{u}_{i}^{x}\epsilon
_{i}.  \label{est-error-detrend}
\end{equation}%
We begin by showing that there exist a finite constant $c_{0}>0$ and a random
variable $m_{0}$ such that, for $m\geq m_{0}$ and all $\varepsilon >0$%
\begin{equation}
\frac{\ln \left( m\left( \ln m\right) ^{2+\varepsilon }\right) }{m^{2}}%
\sum_{i=1}^{m}\left( \widehat{u}_{i}^{x}\right) ^{2}\geq c_{0}.
\label{deno-detrend}
\end{equation}%
As an immediate consequence of Assumption \ref{as-2}, tedious but standard
calculations yield%
\begin{equation}
\frac{\int_{0}^{1}\widehat{W}_{x}^{2}\left( r\right) dr}{\frac{1}{m^{2}}%
\sum_{i=1}^{m}\left( \widehat{u}_{i}^{x}\right) ^{2}}=o_{a.s.}\left(
1\right) ,  \label{sip-detrend}
\end{equation}%
where $\widehat{W}_{x}\left( r\right) $, $r\in \left[ 0,1\right] $, is a
detrended Brownian motion defined as%
\begin{equation*}
\widehat{W}_{x}\left( r\right) =\bar{W}_{x}\left( r\right) -12\left( r-\frac{%
1}{2}\right) \int_{0}^{1}\left( s-\frac{1}{2}\right) \bar{W}_{x}\left(
s\right) ds,
\end{equation*}%
where $\bar{W}_{x}\left( r\right) =W_{x}\left( r\right)
-\int_{0}^{1}W_{x}\left( r\right) dr$, and $W_{x}\left( r\right) $ is
defined in Assumption \ref{as-2}. Then, by Proposition 3.3 in \citet{wembo2012}, it
holds that, for any sequence $f_{m}\rightarrow 0$%
\begin{equation*}
P\left( \int_{0}^{1}\widehat{W}_{x}^{2}\left( r\right) dr\leq f_{m}\right) \leq
c_{0}\frac{1}{f_{m}}\exp \left( -\frac{1}{8f_{m}}\right) .
\end{equation*}%
Upon using $f_{m}^{-1}=8\ln \left( m\left( \ln m\right)
^{2+\varepsilon }\right) $, it is easy to see that%
\begin{equation*}
\sum_{m=1}^{\infty }P\left( \int_{0}^{1}\widehat{W}_{x}^{2}\left( r\right) dr \leq c_{0}f_{m}^{-1}\right) <\infty ,
\end{equation*}%
which, combined with (\ref{sip-detrend}), yields (\ref{deno-detrend}) by the
Borel-Cantelli lemma. 

We now show that%
\begin{equation}
\sum_{i=1}^{m}\widehat{u}_{i}^{x}\epsilon _{i}=o_{a.s.}\left( m\left( \ln
m\right) ^{2+\varepsilon }\right) ,  \label{num-detrend}
\end{equation}%
for all $\varepsilon >0$. Note that%
\begin{equation*}
\widehat{u}_{i}^{x}=u_{i}^{x}+\left( b_{0}-\widehat{b}_{0}\right) +\left(
b_{1}-\widehat{b}_{1}\right) i,
\end{equation*}%
where%
\begin{equation*}
\left( 
\begin{array}{c}
\widehat{b}_{0}-b_{0} \\ 
\widehat{b}_{1}-b_{1}%
\end{array}%
\right) =\left( 
\begin{array}{cc}
m & \frac{m\left( m+1\right) }{2} \\ 
\frac{m\left( m+1\right) }{2} & \frac{m\left( m+1\right) \left( 2m+1\right) }{%
6}%
\end{array}%
\right) ^{-1}\left( 
\begin{array}{c}
\sum_{i=1}^{m}u_{i}^{x} \\ 
\sum_{i=1}^{m}iu_{i}^{x}%
\end{array}%
\right) .
\end{equation*}%
It can be shown by standard arguments that $Var\left( \sum_{i=1}^{m}u_{i}^{x}\right) =O\left( m^{3}\right) 
$ and $Var\left( \sum_{i=1}^{m}iu_{i}^{x}\right) =O\left( m^{5}\right) $. Thus, using Lemma \ref{lemma1} yields%
\begin{eqnarray*}
\widehat{b}_{0}-b_{0} &=&o_{a.s.}\left( m^{-1/2}\left( \ln m\right)
^{1+\varepsilon }\right) , \\
\widehat{b}_{1}-b_{1} &=&o_{a.s.}\left( m^{-1/2}\left( \ln m\right)
^{1+\varepsilon }\right) ,
\end{eqnarray*}%
for all $\varepsilon >0$. Hence%
\begin{equation*}
\sum_{i=1}^{m}\widehat{u}_{i}^{x}\epsilon
_{i}=\sum_{i=1}^{m}u_{i}^{x}\epsilon _{i}+\left( b_{0}-\widehat{b}%
_{0}\right) \sum_{i=1}^{m}\epsilon _{i}+\left( b_{1}-\widehat{b}_{1}\right)
\sum_{i=1}^{m}i\epsilon _{i}.
\end{equation*}%
Assumption \ref{as-2} yields $Var\left( \sum_{i=1}^{m}\epsilon _{i}\right)
=O\left( m\right) $ and $Var\left( \sum_{i=1}^{m}i\epsilon _{i}\right)
=O\left( m^{3}\right) $. Using Lemmas \ref{lemma1} and \ref{lemma3} and
putting everything together, (\ref{num-detrend}) follows. Recalling (\ref%
{est-error-detrend}), this immediately yields the desired result.
\end{proof}
\end{lemma}

\begin{proof}[Proof of Proposition \protect\ref{long-run}]
It holds that%
\begin{align}
& \widehat{\sigma }_{\epsilon }^{2}-\sigma _{\epsilon }^{2}=\widehat{\rho }%
_{0}^{\left( \epsilon \right) }-\rho _{0}^{\left( \epsilon \right)
}+2\sum_{l=1}^{H}\left( 1-\frac{l}{H+1}\right) \left( \widehat{\rho }%
_{l}^{\left( \epsilon \right) }-\rho _{l}^{\left( \epsilon \right) }\right) -%
\frac{2}{H+1}\sum_{l=1}^{H}l\rho _{l}^{\left( \epsilon \right)
}-2\sum_{l=H+1}^{\infty }\rho _{l}^{\left( \epsilon \right) }
\label{sig-error} \\
& =I+II+III+IV.  \notag
\end{align}%
By standard arguments, it follows that Assumption \ref{lrv}\textit{(ii)}
entails that $III=O\left( H^{-1}\right) $ and $IV=o\left( H^{-1}\right) $.
We now consider $I+II$, defining $\kappa \left( l\right) $ such that $\kappa
\left( 0\right) =1$ and $\kappa \left( l\right) =2\left( 1-\frac{l}{H+1}%
\right) $ for $l\geq 1$. Note first that%
\begin{align}
& \widehat{\rho }_{l}^{\left( \epsilon \right) }=\frac{1}{m}%
\sum_{i=l+1}^{m}\epsilon _{i}\epsilon _{i-l}+\frac{1}{m}\left( \beta -%
\widehat{\beta }_{m}\right) \sum_{i=l+1}^{m}x_{i}\epsilon _{i-l}
\label{rho-eps} \\
& +\frac{1}{m}\left( \beta -\widehat{\beta }_{m}\right)
\sum_{i=l+1}^{m}x_{i-l}\epsilon _{i}+\frac{1}{m}\left( \beta -\widehat{\beta 
}_{m}\right) ^{2}\sum_{i=l+1}^{m}x_{i}x_{i-l},  \notag
\end{align}%
and let $\widetilde{\rho }_{l}^{\left( \epsilon \right)
}=m^{-1}\sum_{i=l+i}^{m}\epsilon _{i}\epsilon _{i-l}$. We begin by studying%
\begin{eqnarray*}
&&E\left\vert \sum_{l=0}^{H}\kappa \left( l\right) \left( \widetilde{\rho }%
_{l}^{\left( \epsilon \right) }-\rho _{l}^{\left( \epsilon \right) }\right)
\right\vert ^{2} \\
&=&\sum_{l=0}^{H}\sum_{h=0}^{H}\kappa \left( l\right) \kappa \left( h\right)
E\left( \frac{1}{m}\sum_{i=l+1}^{m}y_{i,l}^{\left( \epsilon \right) }\right)
\left( \frac{1}{m}\sum_{i=h+1}^{m}y_{h,l}^{\left( \epsilon \right) }\right)
+\sum_{l=0}^{H}\sum_{h=0}^{H}\kappa \left( l\right) \kappa \left( h\right) 
\frac{lh}{m^{2}}\left( \rho _{l}^{\left( \epsilon \right) }\rho _{h}^{\left(
\epsilon \right) }\right) \\
&=&\sum_{l=0}^{H}\sum_{h=0}^{H}\kappa \left( l\right) \kappa \left( h\right)
E\left( \frac{1}{m}\sum_{i=l+1}^{m}y_{i,l}^{\left( \epsilon \right) }\right)
\left( \frac{1}{m}\sum_{i=h+1}^{m}y_{h,l}^{\left( \epsilon \right) }\right)
+\left( \sum_{l=0}^{H}\kappa \left( l\right) \frac{l}{m}\rho _{l}^{\left(
\epsilon \right) }\right) ^{2}.
\end{eqnarray*}%
Noting that $\left\vert \kappa \left( l\right) \right\vert \leq 2$, the
second term is bounded by Assumption \ref{lrv}\textit{(ii)}; as far as the
first term is concerned, using the Cauchy-Schwartz inequality, this is
bounded by%
\begin{equation*}
\frac{4}{m^{2}}\sum_{l=0}^{H}\sum_{h=0}^{H}\left( E\left\vert
\sum_{i=l+1}^{m}y_{i,l}^{\left( \epsilon \right) }\right\vert ^{2}\right)
^{1/2}\left( E\left\vert \sum_{i=h+1}^{m}y_{h,l}^{\left( \epsilon \right)
}\right\vert ^{2}\right) ^{1/2}\leq c_{0}\frac{H^{2}}{m},
\end{equation*}%
by Assumption \ref{lrv}\textit{(iii)}. By the maximal inequality for
rectangular sums (see \citealp{moricz1983}), it follows that%
\begin{equation*}
E\max_{1\leq m^{\prime }\leq m,1\leq h\leq H}\left\vert \sum_{l=0}^{H}\kappa
\left( l\right) \frac{1}{m}\sum_{i=l+1}^{m^{\prime }}y_{i,l}^{\left(
\epsilon \right) }\right\vert ^{2}\leq c_{0}\frac{H^{2}}{m}\ln m\ln H,
\end{equation*}%
which in turn, by\ Lemma \ref{stou}, entails%
\begin{equation*}
\sum_{l=0}^{H}\kappa \left( l\right) \left( \widetilde{\rho }_{l}^{\left(
\epsilon \right) }-\rho _{l}^{\left( \epsilon \right) }\right)
=o_{a.s.}\left( \frac{H}{m^{1/2}}\left( \ln m\right) ^{1+\varepsilon }\left(
\ln H\right) ^{1+\varepsilon }\right) ,
\end{equation*}%
for every $\varepsilon >0$. Recalling (\ref{rho-eps}), note that%
\begin{equation*}
E\left\vert \sum_{l=0}^{H}\kappa \left( l\right)
\sum_{i=l+i}^{m}x_{i}x_{i-l}\right\vert \leq
2\sum_{l=0}^{H}\sum_{i=l+i}^{m}\left( Ex_{i}^{2}\right) ^{1/2}\left(
Ex_{i-l}^{2}\right) ^{1/2}\leq c_{0}m^{2}H,
\end{equation*}%
having used Assumption \ref{as-2}\textit{(iii)}. Thus, by the same logic as
above and Lemma \ref{lemma3}%
\begin{equation*}
\frac{1}{m}\left( \beta -\widehat{\beta }_{m}\right)
^{2}\sum_{l=0}^{H}\kappa \left( l\right)
\sum_{i=l+i}^{m}x_{i}x_{i-l}=o_{a.s.}\left( \frac{H}{m}\left( \ln m\right)
^{4+\varepsilon }\left( \ln \ln m\right) ^{2}\left( \ln H\right)
^{2+\varepsilon }\right) .
\end{equation*}%
Also we can derive the (crude) estimate%
\begin{equation*}
E\left\vert \sum_{l=0}^{H}\kappa \left( l\right)
\sum_{i=l+i}^{m}x_{i}\epsilon _{i-l}\right\vert \leq
\sum_{l=0}^{H}\sum_{i=l+i}^{m}\left( Ex_{i}^{2}\right) ^{1/2}\left(
E\epsilon _{i-l}^{2}\right) ^{1/2}\leq c_{0}m^{3/2}H,
\end{equation*}%
by virtue of Assumptions \ref{as-2}\textit{(i)} and \ref{as-2}\textit{(iv)}.
Thus, by the same logic as above%
\begin{equation*}
\frac{1}{m}\left( \beta -\widehat{\beta }_{m}\right) \sum_{l=0}^{H}\kappa
\left( l\right) \sum_{i=l+i}^{m}x_{i}\epsilon _{i-l}=o_{a.s.}\left( \frac{H}{%
m^{1/2}}\left( \ln m\right) ^{3+\varepsilon }\left( \ln \ln m\right) \left(
\ln H\right) ^{2+\varepsilon }\right) .
\end{equation*}%
Thus, in (\ref{sig-error}), we have 
\begin{equation*}
I+II=o_{a.s.}\left( \frac{H}{m^{1/2}}\left( \ln m\right) ^{3+\varepsilon
}\left( \ln \ln m\right) \left( \ln H\right) ^{2+\varepsilon }\right) .
\end{equation*}%
Putting all together, the desired result follows. 
\end{proof}

\setcounter{subsection}{-1} \setcounter{equation}{0} %
\renewcommand{\theequation}{B.\arabic{equation}}

\section{Proofs\label{proofs}}

Results and proofs are presented for the
univariate case, i.e. for $p=1$, for simplicity and without loss of
generality. Also, henceforth we define $E^{\ast }$ and $V^{\ast }$ as the
expected value and the variance according to $P^{\ast }$.

\begin{proof}[Proof of Theorem \protect\ref{theta-1}]
We begin by noting that, by Lemma \ref{lemma5} and by the definition of $%
g\left( m;k\right) $, it holds that 
\begin{equation*}
\frac{Q\left( m;k\right) }{g\left( m;k\right) }=o_{a.s.}\left( m^{-\gamma
}\right) ,
\end{equation*}%
for $1\leq k\leq T_{m}$. In turn, this entails that there exists a random
variable $m_{0}$ such that, for $m\geq m_{0}$%
\begin{equation}
\widetilde{\psi }_{m,k}\geq c_{0}\exp \left( m^{\gamma }\right) ,
\label{psi-3}
\end{equation}%
which also entails that we can assume%
\begin{equation}
\lim_{m\rightarrow \infty }\widetilde{\psi }_{m,k}=\infty \text{.}
\label{psi-2}
\end{equation}%
We show the theorem using $u>0$ without loss of generality. Let $G\left(
\cdot \right) $\ denote the normal distribution. We have%
\begin{align*}
& R^{-1/2}\sum_{j=1}^{R}\left( \zeta _{j}^{\left( k\right) }-\frac{1}{2}%
\right)  \\
& =R^{-1/2}\sum_{j=1}^{R}\left( I\left( \left\vert \widetilde{\psi }%
_{m,k}\right\vert ^{1/2}\xi _{j}^{\left( k\right) }\leq 0\right) -\frac{1}{2}%
\right) +R^{-1/2}\sum_{j=1}^{R}\left( G\left( u\left\vert \widetilde{\psi }%
_{m,k}\right\vert ^{-1/2}\right) -\frac{1}{2}\right)  \\
& +R^{-1/2}\sum_{j=1}^{R}\left( I\left( 0<\left\vert \widetilde{\psi }%
_{m,k}\right\vert ^{1/2}\xi _{j}^{\left( k\right) }\leq u\right) -\left(
G\left( u\left\vert \widetilde{\psi }_{m,k}\right\vert ^{-1/2}\right) -\frac{%
1}{2}\right) \right)  \\
& =I+II+III.
\end{align*}%
We start with $III$; note that $E^{\ast }I\left( 0<\left\vert \widetilde{%
\psi }_{m,k}\right\vert ^{1/2}\xi _{j}^{\left( k\right) }\leq u\right) $ $=$ 
$\left( G\left( u\left\vert \widetilde{\psi }_{m,k}\right\vert
^{-1/2}\right) -\frac{1}{2}\right) $, and%
\begin{equation*}
V^{\ast }I\left( 0<\left\vert \widetilde{\psi }_{m,k}\right\vert ^{1/2}\xi
_{j}^{\left( k\right) }\leq u\right) =\left( G\left( u\left\vert \widetilde{%
\psi }_{m,k}\right\vert ^{-1/2}\right) -\frac{1}{2}\right) \left[ 1-\left(
G\left( u\left\vert \widetilde{\psi }_{m,k}\right\vert ^{-1/2}\right) +\frac{%
1}{2}\right) \right] .
\end{equation*}%
Thus we have%
\begin{align*}
& E^{\ast }\int_{-\infty }^{+\infty }\left\vert R^{-1/2}\sum_{j=1}^{R}\left(
I\left( 0<\left\vert \widetilde{\psi }_{m,k}\right\vert ^{1/2}\xi
_{j}^{\left( k\right) }\leq u\right) -\left( G\left( u\left\vert \widetilde{%
\psi }_{m,k}\right\vert ^{-1/2}\right) -\frac{1}{2}\right) \right)
\right\vert ^{2}dF\left( u\right)  \\
& =\int_{-\infty }^{+\infty }E^{\ast }\left\vert \left( I\left( 0<\left\vert 
\widetilde{\psi }_{m,k}\right\vert ^{1/2}\xi _{1}^{\left( k\right) }\leq
u\right) -\left( G\left( u\left\vert \widetilde{\psi }_{m,k}\right\vert
^{-1/2}\right) -\frac{1}{2}\right) \right) \right\vert ^{2}dF\left( u\right) 
\\
& =\int_{-\infty }^{+\infty }V^{\ast }I\left( 0<\left\vert \widetilde{\psi }%
_{m,k}\right\vert ^{1/2}\xi _{1}^{\left( k\right) }\leq u\right) dF\left(
u\right)  \\
& \leq \int_{-\infty }^{+\infty }E^{\ast }I\left( 0<\left\vert \widetilde{%
\psi }_{m,k}\right\vert ^{1/2}\xi _{1}^{\left( k\right) }\leq u\right)
dF\left( u\right)  \\
& \leq \frac{1}{\sqrt{2\pi }}\left\vert \widetilde{\psi }_{m,k}\right\vert
^{-1/2}\int_{-\infty }^{+\infty }\left\vert u\right\vert dF\left( u\right) .
\end{align*}%
By (\ref{psi-2}) and Assumption \ref{regularity}\textit{(i)}, it holds that $%
III=o_{P^{\ast }}\left( 1\right) $. We now turn to $II$, by studying%
\begin{align*}
& \int_{-\infty }^{+\infty }\left\vert R^{-1/2}\sum_{j=1}^{R}\left( G\left(
u\left\vert \widetilde{\psi }_{m,k}\right\vert ^{-1/2}\right) -\frac{1}{2}%
\right) \right\vert ^{2}dF\left( u\right)  \\
& =R\int_{-\infty }^{+\infty }\left\vert \left( G\left( u\left\vert 
\widetilde{\psi }_{m,k}\right\vert ^{-1/2}\right) -\frac{1}{2}\right)
\right\vert ^{2}dF\left( u\right)  \\
& \leq R\frac{1}{2\pi }\left\vert \widetilde{\psi }_{m,k}\right\vert
^{-1}\int_{-\infty }^{+\infty }u^{2}dF\left( u\right) ,
\end{align*}%
which, by (\ref{psi-3}), Assumption \ref{regularity}\textit{(i)} and (\ref%
{restriction}), entails that $II=o\left( 1\right) $. Putting all together,
Markov inequality yields%
\begin{equation*}
\int_{-\infty }^{+\infty }\left\vert 2R^{-1/2}\sum_{j=1}^{R}\left( \zeta
_{j}^{\left( k\right) }-\frac{1}{2}\right) \right\vert ^{2}dF\left( u\right)
=\int_{-\infty }^{+\infty }\left\vert 2R^{-1/2}\sum_{j=1}^{R}\left( I\left(
\left\vert \widetilde{\psi }_{m,k}\right\vert ^{1/2}\xi _{j}^{\left(
k\right) }\leq 0\right) -\frac{1}{2}\right) \right\vert ^{2}dF\left(
u\right) +o_{P^{\ast }}\left( 1\right) ;
\end{equation*}%
the desired result now follows from the CLT for Bernoulli random variables. 
\end{proof}

\begin{proof}[Proof of Theorem \protect\ref{theta-2}]
Recall that, as $m\rightarrow \infty $, Lemma \ref{lemma6} entails that, for
every $k\geq \left\lfloor m^{\max \left\{ 1,\theta ^{\prime }\right\} \left(
1+\varepsilon \right) }\right\rfloor $, $\varepsilon >0$%
\begin{equation*}
P\left\{ \omega :\psi _{m,k}=\infty \right\} =1.
\end{equation*}%
Therefore we can assume that%
\begin{equation}
\lim_{m\rightarrow \infty }\widetilde{\psi }_{m,k}=0.  \label{psi-4}
\end{equation}%
As in the proof of the previous theorem, we consider the case of $u>0$ only.
We have%
\begin{align*}
& E^{\ast }\left\vert R^{-1/2}\sum_{j=1}^{R}\left( I\left( \left\vert 
\widetilde{\psi }_{m,k}\right\vert ^{1/2}\xi _{j}^{\left( k\right) }\leq
u\right) -\frac{1}{2}\right) \right\vert ^{2} \\
& =E^{\ast }\left\vert R^{-1/2}\sum_{j=1}^{R}\left( I\left( \left\vert 
\widetilde{\psi }_{m,k}\right\vert ^{1/2}\xi _{j}^{\left( k\right) }\leq
u\right) -G\left( u\left\vert \widetilde{\psi }_{m,k}\right\vert
^{-1/2}\right) \right) \right\vert ^{2} \\
& +R\left( G\left( u\left\vert \widetilde{\psi }_{m,k}\right\vert
^{-1/2}\right) -\frac{1}{2}\right) ^{2}=I+II.
\end{align*}%
Equation (\ref{psi-4}) yields immediately that $R^{-1}II\rightarrow \frac{1}{%
4}$. Similarly, note that%
\begin{eqnarray*}
&&E^{\ast }\left\vert \left( I\left( \left\vert \widetilde{\psi }%
_{m,k}\right\vert ^{1/2}\xi _{j}^{\left( k\right) }\leq u\right) -G\left(
u\left\vert \widetilde{\psi }_{m,k}\right\vert ^{-1/2}\right) \right)
\right\vert ^{2} \\
&=&R^{-1}\sum_{j=1}^{R}E^{\ast }\left( I\left( \left\vert \widetilde{\psi }%
_{m,k}\right\vert ^{1/2}\xi _{j}^{\left( k\right) }\leq u\right) -G\left(
u\left\vert \widetilde{\psi }_{m,k}\right\vert ^{-1/2}\right) \right) ^{2} \\
&=&V^{\ast }\left( I\left( \left\vert \widetilde{\psi }_{m,k}\right\vert
^{1/2}\xi _{1}^{\left( k\right) }\leq u\right) -G\left( u\left\vert 
\widetilde{\psi }_{m,k}\right\vert ^{-1/2}\right) \right) <\infty ,
\end{eqnarray*}%
whence we conclude that $I=O_{P^{\ast }}\left( 1\right) $, so that ultimately%
\begin{equation*}
\frac{1}{R}\left\vert R^{-1/2}\sum_{j=1}^{R}\left( I\left( \left\vert 
\widetilde{\psi }_{m,k}\right\vert ^{1/2}\xi _{j}^{\left( k\right) }\leq
u\right) -\frac{1}{2}\right) \right\vert ^{2}\overset{P^{\ast }}{\rightarrow 
}\frac{1}{4},
\end{equation*}%
and therefore, under $H_{A,1}\cup H_{A,2}$%
\begin{equation*}
\frac{1}{R}\int_{-\infty }^{+\infty }\left\vert R^{-1/2}\sum_{j=1}^{R}\left(
\zeta _{j}^{\left( k\right) }-\frac{1}{2}\right) \right\vert ^{2}dF\left(
u\right) =\frac{1}{4}+o_{P^{\ast }}\left( 1\right) .
\end{equation*}%
\end{proof}

\begin{proof}
[Proof of Theorem \ref{monitoring}] The proof of (\ref{null-asy-1}) and (\ref%
{null-asy-2}) follows immediately from Lemma \ref{lemma7}, once noting that%
\begin{equation*}
\max_{1\leq k\leq T_{m}}\left\vert \sum_{i=m+1}^{m+k}\frac{\Theta
_{m,R}^{\left( i\right) }-1}{\sqrt{2}}\right\vert =\max_{1\leq k\leq
T_{m}}\left\vert \sum_{i=m+1}^{m+k}Z_{i}\right\vert +O_{P^{\ast }}\left(
m^{-\varepsilon }\right) ,
\end{equation*}%
for $\varepsilon >0$, where $Z_{i}$ is \textit{i.i.d.} with $E^{\ast
}Z_{h,i}=0$, $V^{\ast }Z_{h,i}=1$ and $E^{\ast }\left\vert
Z_{h,i}\right\vert ^{2+\varepsilon }<\infty $. Detailed passages, based on %
\citet{lajos04} and \citet{lajos07}, can be found in \citet{bt1}.
\end{proof}

\begin{proof}
[Proof of Corollary \ref{corollary}]
Corollary \ref{corollary} follows immediately from Theorem \ref{monitoring};
passages can be found in \citet{bt1}.
\end{proof}

\begin{proof}[Proof of Theorem \protect\ref{drift}]
We present the proof only for (\ref{delta-drift-1}) and (\ref{delta-drift-2}); (\ref{sigma-drift-1}) and (\ref{sigma-drift-2}) can
be shown using exactly the same arguments, and we therefore omit it to save
space. From the proof of (\ref{power}), it is clear that, in order for the
monitoring procedure to detect changes, it must hold that%
\begin{equation*}
\frac{Q\left( m;k\right) }{g\left( m;k\right) }\rightarrow \infty \text{
a.s.,}
\end{equation*}%
as $m\rightarrow \infty $. Also, based on the passages in the proof of Lemma %
\ref{lemma6}, under $H^{\ast}_{A,1}$ the term that dominates in the expansion of $%
Q\left( m;k\right) $ is defined in equation (A.9) as%
\begin{equation*}
c_{0}\Delta _{\beta }^{2}\left( m\right) \frac{\left( m+k\right) ^{2}}{\ln
\ln \left( m+k\right) }.
\end{equation*}%
This entails that, in order to have power, based on (\ref{g-1}), a
sufficient condition is%
\begin{equation}
\Delta _{\beta }^{2}\left( m\right) \frac{\left( m+k\right) ^{2-\varepsilon }%
}{\left( m+k\right) ^{1+\gamma }}\left\vert \max \left\{ 1,\frac{m+k}{m^{2}}%
\right\} \right\vert ^{-\left( 1+\gamma \right) }\rightarrow \infty ,
\label{drift-limit}
\end{equation}%
for arbitrarily small $\varepsilon >0$. On account of the fact that $T_{m}=O\left(
m^{\theta }\right) $, as $k$ approaches the end of the monitoring period it
holds that%
\begin{equation*}
\max \left\{ 1,\frac{m+k}{m^{2}}\right\} =\left\{ 
\begin{array}{l}
1 \\ 
\frac{m+k}{m^{2}}%
\end{array}%
\right. \text{ according as }%
\begin{array}{c}
\theta \leq 2 \\ 
\theta >2%
\end{array}%
;
\end{equation*}%
thus, it is convenient to consider the two cases $\theta \leq 2$ and $\theta
>2$ separately. In the latter case, (\ref{drift-limit}) becomes%
\begin{equation*}
\Delta _{\beta }^{2}\left( m\right) \left( m+k\right) ^{1-\gamma
-\varepsilon }\rightarrow \infty ;
\end{equation*}%
recalling (\ref{gamma}), this can also be rewritten as%
\begin{equation*}
\Delta _{\beta }^{2}\left( m\right) \left( m+k\right) ^{\theta \frac{\theta
-2+\delta }{\theta -1}-\varepsilon }\rightarrow \infty ,
\end{equation*}%
whence (\ref{delta-drift-1}) follows immediately. When $\theta >2$, it must
hold that%
\begin{equation*}
\Delta _{\beta }^{2}\left( m\right) \frac{\left( m+k\right) ^{-\varepsilon }%
}{\left( m+k\right) ^{2\gamma }}m^{2\left( 1+\gamma \right) }\rightarrow
\infty ,
\end{equation*}%
which after some manipulations can be written as%
\begin{equation*}
m^{2-\varepsilon }\Delta _{\beta }^{2}\left( m\right) \frac{m^{2\gamma }}{%
\left( m+k\right) ^{2\gamma }}\rightarrow \infty .
\end{equation*}%
However%
\begin{equation*}
m^{2-\varepsilon }\Delta _{\beta }^{2}\left( m\right) \frac{m^{2\gamma }}{%
\left( m+k\right) ^{2\gamma }}\geq m^{2-\varepsilon }\Delta _{\beta
}^{2}\left( m\right) \frac{m^{2\gamma }}{m^{2\gamma \theta }}%
=m^{2-\varepsilon }\Delta _{\beta }^{2}\left( m\right) m^{2\left( \delta
-1\right) },
\end{equation*}%
which yields (\ref{delta-drift-2}) immediately. 
\end{proof}

\begin{proof}[Proof of Theorem \protect\ref{detrending}]
The proof of the validity of (\ref{power}) follows from using Lemma \ref%
{lemma7}, and exactly the same arguments as in the proof of Lemma \ref%
{lemma6}. The proof of the validity of (\ref{null-asy-1}) and (\ref%
{null-asy-2}) require only to show that, under $H_{0}$%
\begin{equation}
\frac{Q^{d}\left( m;k\right) }{g\left( m;k\right) }\rightarrow 0,
\label{zero-detrend}
\end{equation}%
as $m\rightarrow \infty $ - everything else follows from the same
calculations as in the previous results. Consider%
\begin{eqnarray}
\widehat{\epsilon }_{i}^{d} &=&\widetilde{\epsilon }_{i}-\left( \widehat{\mu 
}_{0,i}+\widehat{\mu }_{1,i}i\right)   \notag \\
&=&\epsilon _{i}+\left( \beta -\widehat{\beta }_{m}^{d}\right) x_{i}+\left(
\mu _{0,i}-\widehat{\mu }_{0,i}\right) +\left( \mu _{1,i}-\widehat{\mu }%
_{1,i}\right) i,  \label{error-2}
\end{eqnarray}%
and note that%
\begin{equation*}
\left( 
\begin{array}{c}
\widehat{\mu }_{0,i}-\mu _{0,i} \\ 
\widehat{\mu }_{1,i}-\mu _{1,i}%
\end{array}%
\right) =\left( 
\begin{array}{cc}
i & \frac{i\left( i+1\right) }{2} \\ 
\left( i+1\right)  & \frac{i\left( i+1\right) \left( 2i+1\right) }{6}%
\end{array}%
\right) ^{-1}\left( 
\begin{array}{c}
\sum_{j=1}^{i}\widetilde{\epsilon }_{j} \\ 
\sum_{j=1}^{i}j\widetilde{\epsilon }_{j}%
\end{array}%
\right) .
\end{equation*}%
Using Lemma \ref{lemma7} and Assumption \ref{as-detrend} it is easy to see
that 
\begin{eqnarray*}
\sum_{j=1}^{i}\widetilde{\epsilon }_{j} &=&o_{a.s.}\left( i^{1/2}\left( \ln
i\right) ^{1+\varepsilon }\left( 1+\frac{\left( \ln m\right) ^{3+\varepsilon
}}{m}i\right) \right) , \\
\sum_{j=1}^{i}j\widetilde{\epsilon }_{j} &=&o_{a.s.}\left( i^{3/2}\left( \ln
i\right) ^{1+\varepsilon }\left( 1+\frac{\left( \ln m\right) ^{3+\varepsilon
}}{m}i\right) \right) ,
\end{eqnarray*}%
whence it follows that%
\begin{eqnarray}
\widehat{\mu }_{0,i}-\mu _{0,i} &=&o_{a.s.}\left( i^{-1/2}\left( \ln
i\right) ^{1+\varepsilon }\left( 1+\frac{\left( \ln m\right) ^{3+\varepsilon
}}{m}i\right) \right) ,  \label{m-0} \\
\widehat{\mu }_{1,i}-\mu _{1,i} &=&o_{a.s.}\left( i^{-3/2}\left( \ln
i\right) ^{1+\varepsilon }\left( 1+\frac{\left( \ln m\right) ^{3+\varepsilon
}}{m}i\right) \right) .  \label{m-1}
\end{eqnarray}%
By (\ref{error-2})%
\begin{equation*}
\frac{1}{4}\sum_{i=m+1}^{m+k}\left( \widehat{\epsilon }_{i}^{d}\right)
^{2}\leq \sum_{i=m+1}^{m+k}\epsilon _{i}^{2}+\left( \widehat{\beta }%
_{m}^{d}-\beta \right)
^{2}\sum_{i=m+1}^{m+k}x_{i}^{2}+\sum_{i=m+1}^{m+k}\left( \widehat{\mu }%
_{0,i}-\mu _{0,i}\right) ^{2}+\sum_{i=m+1}^{m+k}i^{2}\left( \widehat{\mu }%
_{1,i}-\mu _{1,i}\right) ^{2}.
\end{equation*}%
Using Assumption \ref{as-2}\textit{(i)}, Lemma \ref{lemma7} and the LIL\ for
functionals of Brownian motion, (\ref{m-0}) and (\ref{m-1}), it follows that%
\begin{equation*}
\sum_{i=m+1}^{m+k}\left( \widehat{\epsilon }_{i}^{d}\right)
^{2}=o_{a.s.}\left( r_{m,k}^{\prime }\right) ,
\end{equation*}%
with%
\begin{equation*}
r_{m,k}^{\prime }=\left( m+k\right) \left( \ln \left( m+k\right) \right)
^{2+\varepsilon }+\frac{\left( \ln m\right) ^{6+\varepsilon }\left( \ln
k\right) ^{2+\varepsilon }k^{2}}{m^{2}}+\left( \ln k\right) ^{2+\varepsilon
}\left( 1+k\frac{\left( \ln m\right) ^{3+\varepsilon }}{m}\right) ^{2}.
\end{equation*}%
Hence, (\ref{zero-detrend}) follows from (\ref{g-1}). 
\end{proof}

\newpage
\topmargin1.0cm
\textwidth14.25cm
\textheight20.0cm
\oddsidemargin1cm
\evensidemargin1cm
\section{Tables and Figures \label{tables}}
\FloatBarrier
\begin{table}[H]
\centering
\caption{Empirical rejection frequencies under $H_{0}$}
\label{tab:Table1}
{\tiny 
\begin{tabular}{llllllllll}
&  &  &  &  &  &  &  &  &  \\[-6pt] \hline
&  &  &  &  &  &  &  &  &  \\ 
&  & \multicolumn{2}{c}{$T=100$} &  & \multicolumn{2}{c}{$T=200$} &  & 
\multicolumn{2}{c}{$T=400$} \\[2pt] 
&  & $m=25$ & $m=50$ &  & $m=50$ & $m=100$ &  & $m=100$ & $m=200$ \\%
[2pt] \cline{3-4}\cline{6-7}\cline{9-10}
&  &  &  &  &  &  &  &  &  \\[-6pt] 
\multicolumn{2}{l}{Panel A: no serial dependence: $\rho ^{(e)}=0$} &  &  & 
&  &  &  &  &  \\[2pt] 
\multicolumn{2}{l}{\textit{No endogeneity: $\rho ^{(xe)}=0$}} &  &  &  &  & 
&  &  &  \\[2pt] 
$\rho ^{(x)}=0$ & $WW$--$IM$ & 0.057 & 0.056 && 0.064 & 0.040 && 0.049 & 0.057 \\ 
& $\eta =0$ & 0.135 & 0.004 && 0.024 & 0.000 && 0.005 & 0.000 \\ 
& $\eta =0.45$ & 0.154 & 0.046 && 0.051 & 0.049 && 0.047 & 0.040 \\ 
& $\eta =0.49$ & 0.152 & 0.045 && 0.049 & 0.056 && 0.050 & 0.051 \\ 
& $\eta =0.5$ & 0.149 & 0.043 && 0.048 & 0.053 && 0.049 & 0.045 \\[2pt]
$\rho ^{(x)}=0.5$ & $WW$--$IM$ & 0.092 & 0.080 && 0.080 & 0.054 && 0.054 & 0.062 \\ 
& $\eta =0$ & 0.173 & 0.005 &  & 0.026 & 0.000 &  & 0.005 & 0.000 \\ 
& $\eta =0.45$ & 0.188 & 0.048 &  & 0.056 & 0.049 &  & 0.047 & 0.040 \\ 
& $\eta =0.49$ & 0.176 & 0.047 &  & 0.051 & 0.056 &  & 0.050 & 0.051 \\ 
& $\eta =0.5$ & 0.176 & 0.045 &  & 0.049 & 0.053 &  & 0.049 & 0.045 \\[2pt] 
\multicolumn{2}{l}{\textit{Endogeneity: $\rho ^{(xe)}=0.5$}} &  &  &  &  & 
&  &  &  \\ 
$\rho ^{(x)}=0$ & $WW$--$IM$ & 0.054 & 0.052 && 0.058 & 0.036 && 0.045 & 0.058 \\     
& $\eta =0$ & 0.140 & 0.006 &  & 0.029 & 0.000 &  & 0.012 & 0.000 \\ 
& $\eta =0.45$ & 0.157 & 0.048 &  & 0.056 & 0.049 &  & 0.053 & 0.040 \\ 
& $\eta =0.49$ & 0.152 & 0.047 &  & 0.054 & 0.056 &  & 0.056 & 0.051 \\ 
& $\eta =0.5$ & 0.151 & 0.044 &  & 0.054 & 0.053 &  & 0.055 & 0.045 \\[2pt] 
$\rho ^{(x)}=0.5$ & $WW$--$IM$ & 0.077 & 0.067 && 0.070 & 0.042 && 0.048 & 0.057 \\    
& $\eta =0$ & 0.145 & 0.003 &  & 0.023 & 0.000 &  & 0.008 & 0.000 \\ 
& $\eta =0.45$ & 0.164 & 0.045 &  & 0.051 & 0.049 &  & 0.049 & 0.040 \\ 
& $\eta =0.49$ & 0.153 & 0.044 &  & 0.047 & 0.056 &  & 0.051 & 0.051 \\ 
& $\eta =0.5$ & 0.152 & 0.042 &  & 0.047 & 0.053 &  & 0.050 & 0.045 \\[6pt] 
\multicolumn{2}{l}{Panel B: serial dependence: $\rho ^{(e)}=0.5$} &  &  &  & 
&  &  &  &  \\[2pt] 
\multicolumn{2}{l}{\textit{No endogeneity: $\rho ^{(xe)}=0$}} &  &  &  &  & 
&  &  &  \\[2pt] 
$\rho ^{(x)}=0$ & $WW$--$IM$ & 0.118 & 0.121 && 0.108 & 0.082 && 0.077 & 0.080 \\   
& $\eta =0$ & 0.099 & 0.005 &  & 0.020 & 0.000 &  & 0.005 & 0.000 \\ 
& $\eta =0.45$ & 0.119 & 0.047 &  & 0.049 & 0.049 &  & 0.047 & 0.040 \\ 
& $\eta =0.49$ & 0.114 & 0.046 &  & 0.047 & 0.056 &  & 0.050 & 0.051 \\ 
& $\eta =0.5$ & 0.114 & 0.044 &  & 0.046 & 0.053 &  & 0.049 & 0.045 \\[2pt] 
$\rho ^{(x)}=0.5$ & $WW$--$IM$ & 0.139 & 0.125 && 0.113 & 0.080 && 0.071 & 0.075 \\     
& $\eta =0$ & 0.139 & 0.006 &  & 0.025 & 0.000 &  & 0.006 & 0.000 \\ 
& $\eta =0.45$ & 0.151 & 0.049 &  & 0.054 & 0.049 &  & 0.048 & 0.040 \\ 
& $\eta =0.49$ & 0.146 & 0.047 &  & 0.053 & 0.056 &  & 0.051 & 0.051 \\ 
& $\eta =0.5$ & 0.145 & 0.045 &  & 0.052 & 0.053 &  & 0.049 & 0.045 \\[2pt] 
\multicolumn{2}{l}{\textit{Endogeneity: $\rho ^{(xe)}=0.5$}} &  &  &  &  & 
&  &  &  \\ 
$\rho ^{(x)}=0$  & $WW$--$IM$ & 0.108 & 0.129 && 0.115 & 0.090 && 0.091 & 0.088 \\    
& $\eta =0$ & 0.118 & 0.004 &  & 0.023 & 0.000 && 0.010 & 0.000 \\ 
& $\eta =0.45$ & 0.137 & 0.046 &  & 0.052 & 0.049 && 0.050 & 0.040 \\ 
& $\eta =0.49$ & 0.132 & 0.045 &  & 0.048 & 0.056 && 0.052 & 0.051 \\ 
& $\eta =0.5$ & 0.130 & 0.043 &  & 0.047 & 0.053 && 0.051 & 0.045 \\[2pt] 
$\rho ^{(x)}=0.5$ & $WW$--$IM$ & 0.126 & 0.113 && 0.107 & 0.075 && 0.069 & 0.071 \\  
& $\eta =0$ & 0.134 & 0.003 &  & 0.021 & 0.000 &  & 0.006 & 0.000 \\ 
& $\eta =0.45$ & 0.152 & 0.045 &  & 0.048 & 0.049 &  & 0.048 & 0.040 \\ 
& $\eta =0.49$ & 0.147 & 0.044 &  & 0.046 & 0.056 &  & 0.051 & 0.051 \\ 
& $\eta =0.5$ & 0.146 & 0.042 &  & 0.046 & 0.053 &  & 0.049 & 0.045 \\[6pt]
\multicolumn{2}{l}{Panel C: serial dependence: $\rho ^{(e)}=0.9$} &  &  &  & 
&  &  &  &  \\[2pt]  
\multicolumn{2}{l}{\textit{No endogeneity: $\rho ^{(xe)}=0$}} &  &  &  &  & 
&  &  &  \\[2pt] 
$\rho ^{(x)}=0$ & $WW$--$IM$ & 0.280 & 0.269 && 0.222 & 0.205 && 0.173 & 0.148 \\   
& $\eta =0$ & 0.228 & 0.029 && 0.081 & 0.002 && 0.026 & 0.001  \\ 
& $\eta =0.45$ & 0.246 & 0.069 && 0.111 & 0.051 && 0.067 & 0.041  \\ 
& $\eta =0.49$ & 0.238 & 0.068 && 0.106 & 0.058 && 0.070 & 0.052 \\ 
& $\eta =0.5$ & 0.236 & 0.065 && 0.104 & 0.055 && 0.068 & 0.046 \\[2pt] 
$\rho ^{(x)}=0.5$ & $WW$--$IM$ & 0.309 & 0.298 && 0.233 & 0.219 && 0.181 & 0.153  \\     
& $\eta =0$ & 0.293 & 0.038 && 0.111 & 0.006 && 0.032 & 0.001 \\ 
& $\eta =0.45$ & 0.308 & 0.080 && 0.138 &  0.054 && 0.070 & 0.041 \\ 
& $\eta =0.49$ & 0.302 & 0.077 && 0.134 & 0.060 && 0.073 & 0.052 \\ 
& $\eta =0.5$ & 0.298 & 0.074 && 0.132 & 0.057 && 0.072 & 0.046 \\[2pt]
\multicolumn{2}{l}{\textit{Endogeneity: $\rho ^{(xe)}=0.5$}} &  &  &  &  & 
&  &  &  \\ 
$\rho ^{(x)}=0$  & $WW$--$IM$ & 0.284 & 0.297 && 0.233 & 0.227 && 0.190 & 0.163  \\    
& $\eta =0$ & 0.235 & 0.031 && 0.080 & 0.003 && 0.025 & 0.001 \\ 
& $\eta =0.45$ & 0.248 & 0.071 && 0.105 & 0.052 && 0.066 & 0.041 \\ 
& $\eta =0.49$ & 0.233 & 0.070 && 0.104 & 0.059 && 0.068 & 0.052\\ 
& $\eta =0.5$ & 0.229 & 0.066 && 0.103 & 0.056 && 0.066 & 0.046 \\[2pt]

$\rho ^{(x)}=0.5$ & $WW$--$IM$ & 0.288 & 0.293 && 0.232 & 0.230 && 0.190 & 0.160 \\  
& $\eta =0$ &  0.254 & 0.030 && 0.096 & 0.005 && 0.031 & 0.001 \\ 
& $\eta =0.45$ & 0.267 & 0.069 && 0.120 & 0.054 && 0.070 & 0.041 \\ 
& $\eta =0.49$ & 0.258 & 0.068 && 0.116 & 0.060 && 0.073 & 0.052 \\ 
& $\eta =0.5$ & 0.257 & 0.066 && 0.114 & 0.057 && 0.072 & 0.046\\[2pt] \hline
\end{tabular}
} 
\end{table}

\newpage

\topmargin1.0cm \textwidth14.25cm \textheight22.0cm \oddsidemargin-1.2cm %
\evensidemargin-1.2cm 
\begin{landscape}
\begin{table}[h!]
\centering
\caption{Empirical rejection frequencies under $H_{A,1}$.}
{\tiny 
\begin{tabular}{llccccclccccclccccc} \\[-6pt] \hline \\
&  & \multicolumn{5}{c}{$T=100$} &  & \multicolumn{5}{c%
}{$T=200$} &  & 
\multicolumn{5}{c}{$T=400$} \\[2pt]
&  & \multicolumn{2}{c%
}{$\Delta _{\beta }=0.5$} &  & \multicolumn{2}{c}{$%
\Delta _{\beta }=1$} & 
\multicolumn{1}{c}{} & \multicolumn{2}{c}{$\Delta
_{\beta }=0.5$} &  & 
\multicolumn{2}{c}{$\Delta _{\beta }=1$} & 
\multicolumn{1}{c}{} & 
\multicolumn{2}{c}{$\Delta _{\beta }=0.5$} &  & 
\multicolumn{2}{c}{$\Delta
_{\beta }=1$} \\[2pt]
&  & $m=25$ & $m=50$ &  & $m=25$ & $m=50$ & \multicolumn{1%
}{c}{} & $m=50$ & $%
m=100$ &  & $m=50$ & $m=100$ & \multicolumn{1}{c}{} & $%
m=100$ & $m=200$ &  & 
$m=100$ & $m=200$ \\[2pt]
\cline{3-4}\cline{6-7}\cline{%
9-10}\cline{12-13}\cline{15-16}\cline{18-19} \\[-6pt]
\multicolumn{2}{l}{Panel A: no
serial dependence: $\rho ^{(e)}=0$} &  &  & 
&  &  & \multicolumn{1}{c}{} &
&  &  &  &  & \multicolumn{1}{c}{} &  &  & 
&  &  \\[2pt]
\multicolumn{2}{l%
}{\textit{No endogeneity: $\rho ^{(xe)}=0$}} &  &  &  &  & 
& \multicolumn{1%
}{c}{} &  &  &  &  &  & \multicolumn{1}{c}{} &  &  &  &  & 
\\[2pt]
$\rho
^{(x)}=0$ & $WW$--$IM$ & 0.549 & 0.663 && 0.787 & 0.838 & \multicolumn{1}{c}{} & 0.758 & 0.820 && 0.914 & 0.928 & \multicolumn{1}{c}{} & 0.915 & 0.917 && 0.982 & 0.980 \\
& $\eta =0$ & $0.855$ & $0.699$ &  & $0.990$ & $0.918$ & 
\multicolumn{1}{c}{} & $0.944$ & $0.799$ &  & $0.999$ & $0.964$ & 
\multicolumn{1}{c}{} & $0.985$ & $0.872$ &  & $1.000$ & $0.989$ \\ 
& $\eta=0.45$ & $0.848$ & $0.722$ &  & $0.990$ & $0.926$ & 
\multicolumn{1}{c}{} & $%
0.942$ & $0.804$ &  & $0.999$ & $0.971$ & 
\multicolumn{1}{c}{} & $0.984$ & 
$0.882$ &  & $1.000$ & $0.990$ \\ 
& $\eta =0.49$ & $0.843$ & $0.711$ &  & $%
0.987$ & $0.921$ & 
\multicolumn{1}{c}{} & $0.939$ & $0.806$ &  & $0.999$ & 
$0.966$ & 
\multicolumn{1}{c}{} & $0.981$ & $0.879$ &  & $1.000$ & $0.989$
\\ 
& $\eta =0.5$ & $0.842$ & $0.706$ &  & $0.987$ & $0.918$ & 
\multicolumn{1}{c}{} & $0.939$ & $0.805$ &  & $0.999$ & $0.965$ & 
\multicolumn{1}{c}{} & $0.981$ & $0.877$ &  & $1.000$ & $0.989$ \\[2pt]
$\rho ^{(x)}=0.5$ & $WW$--$IM$ & 0.781 & 0.833 && 0.925 & 0.923 & \multicolumn{1}{c}{} & 0.907 & 0.915 && 0.977 & 0.976 & \multicolumn{1}{c}{} & 0.979 & 0.972 && 0.999 & 0.996 \\
& $\eta =0$ & 0.967 & 0.864 && 0.999 & 0.973 & 
\multicolumn{1}{c}{} & 0.996 & 0.929 && 1.000 & 0.997  & 
\multicolumn{1}{c}{} & 1.000 & 0.982 && 1.000 & 1.000 \\ 
& $\eta=0.45$ & 0.965 & 0.873 && 0.999 & 0.974 & 
\multicolumn{1}{c}{} & 0.996 & 0.931 && 1.000 & 0.997 & 
\multicolumn{1}{c}{} & 1.000 & 0.982 && 1.000 & 1.000\\ 
& $\eta =0.49$ & 0.962 & 0.868 && 0.999 & 0.973 & 
\multicolumn{1}{c}{} & 0.995 & 0.931 && 1.000 & 0.997 & 
\multicolumn{1}{c}{} & 1.000 & 0.982 && 1.000 & 1.000\\ 
& $\eta =0.5$ & 0.962 & 0.868 && 0.999 & 0.973 & 
\multicolumn{1}{c}{} & 0.995 & 0.931 && 1.000 & 0.997  & 
\multicolumn{1}{c}{} & 1.000 & 0.982 && 1.000 & 1.000 \\ 
\multicolumn{2}{l}{\textit{Endogeneity: $\rho ^{(xe)}=0.5$}} &  &  &  &  &

& \multicolumn{1}{c}{} &  &  &  &  &  & \multicolumn{1}{c}{} &  &  &  &  &

\\ 
$\rho ^{(x)}=0$ & $WW$--$IM$ & 0.609 & 0.724 && 0.825 & 0.865 & \multicolumn{1}{c}{} & 0.807 & 0.851 && 0.934 & 0.947 & \multicolumn{1}{c}{} & 0.948 & 0.941 && 0.990 & 0.988 \\
& $\eta =0$ & 0.819 & 0.652 && 0.988 & 0.910 & 
\multicolumn{1}{c}{} & 0.925 & 0.784 && 0.999 & 0.971 & 
\multicolumn{1}{c}{} & 0.988 & 0.867 && 1.000 & 0.990 \\ 
& $\eta=0.45$ & 0.812 & 0.679 && 0.987 & 0.917  & 
\multicolumn{1}{c}{} & 0.923 & 0.794 && 0.998 & 0.975 & 
\multicolumn{1}{c}{} & 0.988 & 0.878 && 1.000 & 0.990\\ 
& $\eta =0.49$ & 0.808 & 0.668 && 0.983 & 0.912 & 
\multicolumn{1}{c}{} & 0.922 & 0.793 && 0.998 & 0.972 & 
\multicolumn{1}{c}{} & 0.988 & 0.874 && 1.000 & 0.990\\ 
& $\eta =0.5$ & 0.808 & 0.663 && 0.983 & 0.911 & 
\multicolumn{1}{c}{} & 0.921 & 0.789 && 0.998 & 0.972  & 
\multicolumn{1}{c}{} & 0.985 & 0.872 && 1.000 & 0.990 \\[2pt]
$%
\rho ^{(x)}=0.5$ & $WW$--$IM$ & 0.808 & 0.853 && 0.942 & 0.937 & \multicolumn{1}{c}{} & 0.930 & 0.937 && 0.982 & 0.982 & \multicolumn{1}{c}{} & 0.985 & 0.979 && 0.999 & 0.996\\
& $\eta =0$ & 0.946 & 0.821 && 0.999 & 0.965 & 
\multicolumn{1}{c}{} & 0.989 & 0.908 && 1.000 & 0.991  & 
\multicolumn{1}{c}{} & 0.999 & 0.957 && 1.000 & 1.000 \\ 
& $\eta=0.45$ & 0.944 & 0.831 && 0.999 & 0.966 & 
\multicolumn{1}{c}{} & 0.989 & 0.911 && 1.000 & 0.991  & 
\multicolumn{1}{c}{} & 0.999 & 0.960 && 1.000 & 1.000  \\ 
& $\eta =0.49$ & 0.939 & 0.827 && 0.999 & 0.965 & 
\multicolumn{1}{c}{} & 0.989 & 0.911 && 1.000 & 0.991& 
\multicolumn{1}{c}{} & 0.999 & 0.958 && 1.000 & 1.000 \\ 
& $\eta =0.5$ & 0.938 & 0.824 && 0.999 & 0.965 & 
\multicolumn{1}{c}{} & 0.988 & 0.908 && 1.000 & 0.991 & 
\multicolumn{1}{c}{} & 0.999 & 0.957 && 1.000 & 1.000 \\[6pt] 
\multicolumn{2}{l}{Panel B: serial dependence: $\rho ^{(e)}=0.5$} &  &  &  &

&  & \multicolumn{1}{c}{} &  &  &  &  &  & \multicolumn{1}{c}{} &  &  &  &

&  \\[2pt]
\multicolumn{2}{l}{\textit{No endogeneity: $\rho ^{(xe)}=0$}}
&  &  &  &  & 
& \multicolumn{1}{c}{} &  &  &  &  &  & \multicolumn{1}{c}{}
&  &  &  &  & 
\\[2pt]
$\rho ^{(x)}=0$ & $WW$--$IM$ & 0.526 & 0.598 && 0.762 & 0.810 & \multicolumn{1}{c}{} & 0.700 & 0.748 && 0.874 & 0.881 & \multicolumn{1}{c}{} & 0.868 & 0.868 && 0.967 & 0.949  \\
& $\eta =0$ & 0.753 & 0.547 && 0.959 & 0.831 & 
\multicolumn{1}{c}{} & 0.840 & 0.653 && 0.987 & 0.906 & 
\multicolumn{1}{c}{} & 0.927 & 0.753 && 1.000 & 0.953\\ 
& $\eta =0.45$ & 0.749 & 0.572 && 0.955 & 0.842  & 
\multicolumn{1}{c}{} & 0.840 & 0.671 && 0.987 & 0.910  & 
\multicolumn{1}{c}{} &  0.923 & 0.765 && 1.000 & 0.956\\ 
& $\eta=0.49$ & 0.741 & 0.564 && 0.952 & 0.837 & 
\multicolumn{1}{c}{} & 0.836 & 0.671 && 0.986 & 0.909 & 
\multicolumn{1}{c}{} & 0.924 & 0.763 && 1.000 & 0.956 \\ 
& $\eta =0.5$ & 0.739 & 0.560 && 0.952 & 0.834 & 
\multicolumn{1}{c}{} & 0.833 & 0.671 && 0.985 & 0.908 & 
\multicolumn{1}{c}{} & 0.922 & 0.759 && 1.000 & 0.953\\[2pt]
$\rho ^{(x)}=0.5$ & $WW$--$IM$ & 0.735 & 0.779 && 0.894 & 0.892 & \multicolumn{1}{c}{} & 0.847 & 0.860 && 0.948 & 0.947 & \multicolumn{1}{c}{} & 0.954 & 0.935 && 0.988 & 0.985 \\
& $\eta =0$ & 0.916 & 0.773 && 0.999 & 0.946 & 
\multicolumn{1}{c}{} & 0.972 & 0.858 && 1.000 & 0.982 &
\multicolumn{1}{c}{} & 0.995 & 0.927 && 1.000 & 0.998 \\ 
& $%
\eta =0.45$ & 0.915 & 0.788 && 0.998 & 0.951 & 
\multicolumn{1}{c}{} & 0.971 & 0.863 && 1.000 & 0.984 & 
\multicolumn{1}{c}{} & 0.994 & 0.931 && 1.000 & 0.998 \\ 
& $\eta =0.49$ & 0.911 & 0.780 && 0.997 & 0.947  & 
\multicolumn{1}{c}{} & 0.969 & 0.865 && 1.000 & 0.983  & 
\multicolumn{1}{c}{} & 0.994 & 0.929 && 1.000 & 0.998 \\ 
& $\eta =0.5$ & 0.910 & 0.778 && 0.997 & 0.947  & 
\multicolumn{1}{c}{} & 0.969 & 0.863 && 1.000 & 0.983  & 
\multicolumn{1}{c}{} & 0.994 & 0.928 && 1.000 & 0.998 \\
\multicolumn{2}{l}{\textit{Endogeneity: $\rho ^{(xe)}=0.5$}} &  &  &  &  &

& \multicolumn{1}{c}{} &  &  &  &  &  & \multicolumn{1}{c}{} &  &  &  &  &
\\ 
$\rho ^{(x)}=0$ & $WW$--$IM$ & 0.599 & 0.676 && 0.818 & 0.854 & \multicolumn{1}{c}{} & 0.766 & 0.813 && 0.910 & 0.918 & \multicolumn{1}{c}{} & 0.902 & 0.902 && 0.980 & 0.969  \\
& $\eta =0$ & 0.739 & 0.566 && 0.971 & 0.859 & 
\multicolumn{1}{c}{} & 0.858 & 0.694 && 0.993 & 0.927 & 
\multicolumn{1}{c}{} & 0.953 & 0.782 && 1.000 & 0.967 \\ 
& $\eta=0.45$ & 0.735 & 0.591 && 0.969 & 0.868  & 
\multicolumn{1}{c}{} & 0.859 & 0.712 && 0.993 & 0.931 & 
\multicolumn{1}{c}{} & 0.954 & 0.794 && 1.000 & 0.977 \\ 
& $\eta =0.49$ & 0.727 & 0.584 && 0.966 & 0.864  & 
\multicolumn{1}{c}{} & 0.853 & 0.710 && 0.992 & 0.930 & 
\multicolumn{1}{c}{} & 0.953 & 0.792 && 1.000 & 0.976 \\ 
& $\eta =0.5$ & 0.724 & 0.579 && 0.965 & 0.862  & 
\multicolumn{1}{c}{} & 0.850 & 0.708 && 0.992 & 0.928 & 
\multicolumn{1}{c}{} & 0.952 & 0.789 && 1.000 & 0.976 \\[2pt]
$\rho ^{(x)}=0.5$ & $WW$--$IM$ & 0.784 & 0.825 && 0.920 & 0.909 & \multicolumn{1}{c}{} & 0.880 & 0.885 && 0.962 & 0.962 & \multicolumn{1}{c}{} & 0.969 & 0.950 && 0.996 & 0.991 \\
& $\eta =0$ & 0.921 & 0.780 && 0.999 & 0.951 & 
\multicolumn{1}{c}{} & 0.978 & 0.875 && 1.000 & 0.987 & 
\multicolumn{1}{c}{} & 0.998 & 0.933 && 1.000 & 0.998  \\ 
& $\eta=0.45$ & 0.918 & 0.790 && 0.999 & 0.953  & 
\multicolumn{1}{c}{} & 0.978 & 0.879 && 1.000 & 0.987 & 
\multicolumn{1}{c}{} & 0.998 & 0.937 && 1.000 & 0.998 \\ 
& $\eta =0.49$ & 0.912 & 0.787 && 0.999 & 0.952 & 
\multicolumn{1}{c}{} & 0.974 & 0.877 && 1.000 & 0.987 & 
\multicolumn{1}{c}{} & 0.998 & 0.936 && 1.000 & 0.998\\ 
& $\eta =0.5$ & 0.910 & 0.784 && 0.999 & 0.952  & 
\multicolumn{1}{c}{} & 0.973 & 0.876 && 1.000 & 0.987 & 
\multicolumn{1}{c}{} & 0.998 & 0.935 && 1.000 & 0.998 \\[6pt]
\multicolumn{2}{l}{Panel C: serial dependence: $\rho ^{(e)}=0.9$} &  &  &  &

&  & \multicolumn{1}{c}{} &  &  &  &  &  & \multicolumn{1}{c}{} &  &  &  &

&  \\[2pt]
\multicolumn{2}{l}{\textit{No endogeneity: $\rho ^{(xe)}=0$}}
&  &  &  &  & 
& \multicolumn{1}{c}{} &  &  &  &  &  & \multicolumn{1}{c}{}
&  &  &  &  & 
\\[2pt]
$\rho ^{(x)}=0$ & $WW$--$IM$ & 0.652 & 0.617 && 0.840 & 0.784 & \multicolumn{1}{c}{} & 0.685 & 0.660 && 0.855 & 0.829 & \multicolumn{1}{c}{} & 0.766 & 0.764 && 0.921 & 0.889\\
& $\eta =0$ & 0.845 & 0.595 && 0.980 & 0.850 & 
\multicolumn{1}{c}{} & 0.847 & 0.646 && 0.980 & 0.884 & 
\multicolumn{1}{c}{} & 0.896 & 0.705 && 0.996 & 0.928\\ 
& $\eta =0.45$ & 0.843 & 0.614 && 0.977 & 0.859 & 
\multicolumn{1}{c}{} & 0.846 & 0.665 && 0.978 & 0.887 & 
\multicolumn{1}{c}{} & 0.900 & 0.718 && 0.996 & 0.934 \\ 
& $\eta=0.49$ & 0.838 & 0.610 && 0.976 & 0.854  & 
\multicolumn{1}{c}{} &  0.843 & 0.667 && 0.978 & 0.889 & 
\multicolumn{1}{c}{} &  0.896 & 0.716 && 0.995 & 0.931 \\ 
& $\eta =0.5$ & 0.837 & 0.606 && 0.976 & 0.852 & 
\multicolumn{1}{c}{} & 0.839 & 0.664 && 0.977 & 0.888 & 
\multicolumn{1}{c}{} & 0.895 & 0.713 && 0.995 & 0.929 \\[2pt]
$\rho ^{(x)}=0.5$ & $WW$--$IM$ & 0.824 & 0.770 && 0.942 & 0.892 & \multicolumn{1}{c}{} & 0.834 & 0.801 && 0.935 & 0.911 & \multicolumn{1}{c}{} & 0.900 & 0.872 && 0.975 & 0.947  \\
& $\eta =0$ & 0.957 & 0.803 && 0.998 & 0.949 & 
\multicolumn{1}{c}{} & 0.966 & 0.842 && 0.998 & 0.975  &
\multicolumn{1}{c}{} & 0.989 & 0.892 && 1.000 & 0.990 \\ 
& $%
\eta =0.45$ & 0.954 & 0.816 && 0.998 & 0.949 & 
\multicolumn{1}{c}{} & 0.966 & 0.848 && 0.998 & 0.976 & 
\multicolumn{1}{c}{} & 0.989 & 0.897 && 1.000 & 0.990 \\ 
& $\eta =0.49$ & 0.952 & 0.811 && 0.998 & 0.949 & 
\multicolumn{1}{c}{} & 0.964 & 0.848 && 0.998 & 0.976 & 
\multicolumn{1}{c}{} & 0.988 & 0.896 && 1.000 & 0.990 \\ 
& $\eta =0.5$ & 0.952 & 0.807 && 0.998 & 0.949 & 
\multicolumn{1}{c}{} & 0.964 & 0.846 && 0.998 & 0.974 & 
\multicolumn{1}{c}{} & 0.988 & 0.894 && 1.000 & 0.990 \\ 
\multicolumn{2}{l}{\textit{Endogeneity: $\rho ^{(xe)}=0.5$}} &  &  &  &  &

& \multicolumn{1}{c}{} &  &  &  &  &  & \multicolumn{1}{c}{} &  &  &  &  &
\\ 
$\rho ^{(x)}=0$ & $WW$--$IM$ & 0.714 & 0.682 && 0.882 & 0.836 & \multicolumn{1}{c}{} & 0.741 & 0.734 && 0.886 & 0.869 & \multicolumn{1}{c}{} & 0.824 & 0.814 && 0.947 & 0.919 \\
& $\eta =0$ & 0.877 & 0.655 && 0.988 & 0.895 & 
\multicolumn{1}{c}{} & 0.904 & 0.709 && 0.990 & 0.919 & 
\multicolumn{1}{c}{} & 0.950 & 0.776 && 0.999 & 0.966 \\ 
& $\eta=0.45$ & 0.870 & 0.676 && 0.987 & 0.900 & 
\multicolumn{1}{c}{} & 0.902 & 0.725 && 0.989 & 0.923 & 
\multicolumn{1}{c}{} &  0.951 & 0.787 && 0.999 & 0.968\\ 
& $\eta =0.49$ & 0.865 & 0.670 && 0.987 & 0.897 & 
\multicolumn{1}{c}{} & 0.896 & 0.727 && 0.989 & 0.922 & 
\multicolumn{1}{c}{} & 0.951 & 0.784 && 0.999 & 0.967\\ 
& $\eta =0.5$ & 0.865 & 0.665 && 0.987 & 0.895 & 
\multicolumn{1}{c}{} & 0.894 & 0.726 && 0.988 & 0.920  & 
\multicolumn{1}{c}{} & 0.945 & 0.781 && 0.999 & 0.966 \\[2pt]
$\rho ^{(x)}=0.5$ & $WW$--$IM$ & 0.861 & 0.810 && 0.954 & 0.909 & \multicolumn{1}{c}{} & 0.872 & 0.845 && 0.957 & 0.933 & \multicolumn{1}{c}{} & 0.932 & 0.894 && 0.987 & 0.968\\
& $\eta =0$ & 0.975 & 0.834 && 0.999 & 0.964 & 
\multicolumn{1}{c}{} & 0.979 & 0.879 && 1.000 & 0.987  & 
\multicolumn{1}{c}{} & 0.997 & 0.927 && 1.000 & 0.997 \\ 
& $\eta=0.45$ & 0.972 & 0.844 && 0.999 & 0.965 & 
\multicolumn{1}{c}{} & 0.980 & 0.883 && 1.000 & 0.987 & 
\multicolumn{1}{c}{} & 0.996 & 0.930 && 1.000 & 0.997 \\ 
& $\eta =0.49$ & 0.969 & 0.839 && 0.999 & 0.965  & 
\multicolumn{1}{c}{} & 0.980 & 0.883 && 1.000 & 0.987  & 
\multicolumn{1}{c}{} & 0.996 & 0.929 && 1.000 & 0.997\\ 
& $\eta =0.5$ & 0.968 & 0.837 && 0.999 & 0.964 & 
\multicolumn{1}{c}{} & 0.980 & 0.881 && 1.000 & 0.987 & 
\multicolumn{1}{c}{} & 0.996 & 0.928 && 1.000 & 0.997 \\[2pt]\hline
\end{tabular}
}
\label{tab:Table2a}
\end{table}
\end{landscape}

\topmargin1.0cm 
\textwidth14.25cm 
\textheight22.0cm 
\oddsidemargin-1.2cm
\evensidemargin-1.2cm 
\begin{landscape}
\begin{table}[h!]
\centering
\caption{Detection delay under $H_{A,1}$.}
{\tiny 
\begin{tabular}{llccccclccccclccccc} \\[-6pt] \hline \\

&  & \multicolumn{5}{c}{$T=100$} &  & \multicolumn{5}{c}{$T=200$} &  & 
\multicolumn{5}{c}{$T=400$} \\[2pt]
&  & \multicolumn{2}{c}{$\Delta _{\beta
}=0.5$} &  & \multicolumn{2}{c}{$%
\Delta _{\beta }=1$} & \multicolumn{1}{c}{%
} & \multicolumn{2}{c}{$\Delta
_{\beta }=0.5$} &  & \multicolumn{2}{c}{$%
\Delta _{\beta }=1$} & 
\multicolumn{1}{c}{} & \multicolumn{2}{c}{$\Delta
_{\beta }=0.5$} &  & 
\multicolumn{2}{c}{$\Delta _{\beta }=1$} \\[2pt] 
&  & $%
m=25$ & $m=50$ &  & $m=25$ & $m=50$ & \multicolumn{1}{c}{} & $m=50$ & $%
m=100$ &  & $m=50$ & $m=100$ & \multicolumn{1}{c}{} & $m=100$ & $m=200$ &
& 
$m=100$ & $m=200$ \\[2pt]
\cline{3-4}\cline{6-7}\cline{9-10}\cline{12-13}%
\cline{15-16}\cline{18-19} \\[-6pt] 
\multicolumn{2}{l}{Panel A: no serial
dependence: $\rho ^{(e)}=0$} &  &  & 
&  &  & \multicolumn{1}{c}{} &  &  & 
&  &  & \multicolumn{1}{c}{} &  &  & 
&  &  \\[2pt]
\multicolumn{2}{l}{%
\textit{No endogeneity: $\rho ^{(xe)}=0$}} &  &  &  &  & 
& \multicolumn{1}{%
c}{} &  &  &  &  &  & \multicolumn{1}{c}{} &  &  &  &  & 
\\[2pt]
$\rho
^{(x)}=0$ & $WW$--$IM$ & 0.524 & 0.184 && 0.405 & 0.137 & \multicolumn{1}{c}{} & 0.406 & 0.136 && 0.300 & 0.102 & \multicolumn{1}{c}{} & 0.301 & 0.102 && 0.208 & 0.074 \\
& $\eta =0$ & 0.244 & 0.127 && 0.154 & 0.088 & 
\multicolumn{1}{c}{} & 0.213 & 0.098 && 0.107 & 0.065 & 
\multicolumn{1}{c}{} &  0.157 & 0.083 && 0.066 & 0.049 \\ 
& $\eta=0.45$ & 0.216 & 0.099 && 0.137 & 0.067 & 
\multicolumn{1}{c}{} & 0.190 & 0.070 && 0.089 & 0.046
 & 
\multicolumn{1}{c}{} & 0.128 & 0.066 && 0.043 & 0.034 \\ 
& $\eta =0.49$ & 0.222 & 0.100 && 0.141& 0.068 & 
\multicolumn{1}{c}{} & 0.193 & 0.067 && 0.092 & 0.042 & 
\multicolumn{1}{c}{} & 0.127 & 0.061 && 0.042 & 0.032
\\ 
& $\eta =0.5$ & 0.224 & 0.101 && 0.143 & 0.069 & 
\multicolumn{1}{c}{} & 0.197 & 0.069 && 0.094 & 0.043 & 
\multicolumn{1}{c}{} & 0.129 & 0.063 && 0.043 & 0.032 \\[2pt]
$%
\rho ^{(x)}=0.5$ & $WW$--$IM$ & 0.392 & 0.136 && 0.292 & 0.099 & \multicolumn{1}{c}{} & 0.307 & 0.103 && 0.212 & 0.074 & \multicolumn{1}{c}{} & 0.214 & 0.077 && 0.138 & 0.052 \\
& $\eta =0$ & 0.172 & 0.092 && 0.090 & 0.059 & 
\multicolumn{1}{c}{} & 0.129 & 0.069 && 0.061 & 0.040 & 
\multicolumn{1}{c}{} & 0.082 & 0.058 && 0.035 & 0.027 \\ 
& $\eta=0.45$ & 0.152 & 0.069 && 0.072 & 0.040 & 
\multicolumn{1}{c}{} & 0.110 & 0.048 && 0.045 & 0.023 & 
\multicolumn{1}{c}{} & 0.057 & 0.042 && 0.013 & 0.013\\ 
& $\eta =0.49$ & 0.157 & 0.070 && 0.077 & 0.041  & 
\multicolumn{1}{c}{} & 0.113 & 0.044 && 0.046 & 0.020 & 
\multicolumn{1}{c}{} & 0.057 & 0.039 && 0.012 & 0.009\\ 
& $\eta =0.5$ & 0.160 & 0.072 && 0.078 & 0.043 & 
\multicolumn{1}{c}{} & 0.115 & 0.046 && 0.048 & 0.021 & 
\multicolumn{1}{c}{} & 0.058 & 0.041 && 0.013 & 0.011 \\ 
\multicolumn{2}{l}{\textit{Endogeneity: $\rho ^{(xe)}=0.5$}} &  &  &  &  &

& \multicolumn{1}{c}{} &  &  &  &  &  & \multicolumn{1}{c}{} &  &  &  &  &\\ 
$\rho ^{(x)}=0$ & $WW$--$IM$ & 0.501 & 0.173 && 0.376 & 0.126 & \multicolumn{1}{c}{} & 0.382 & 0.125 && 0.270 & 0.092 && 0.279 & 0.094 && 0.183 & 0.066\\
& $\eta =0$ & 0.259 & 0.130 && 0.163 & 0.090 & 
\multicolumn{1}{c}{} & 0.222 & 0.100 && 0.108 & 0.068 & 
\multicolumn{1}{c}{} & 0.166 & 0.084 && 0.067 & 0.049 \\ 
& $\eta=0.45$ & 0.229 & 0.101 && 0.144 & 0.068 & 
\multicolumn{1}{c}{} & 0.198 & 0.073 && 0.089 & 0.048  & 
\multicolumn{1}{c}{} & 0.138 & 0.067 && 0.043 & 0.034 \\ 
& $\eta =0.49$ & 0.237 & 0.101 && 0.148 & 0.070 & 
\multicolumn{1}{c}{} & 0.205 & 0.068 && 0.093 & 0.044 & 
\multicolumn{1}{c}{} & 0.140 & 0.061 && 0.043 & 0.032
\\ 
& $\eta =0.5$ & 0.241 & 0.102 && 0.150 & 0.071  & 
\multicolumn{1}{c}{} & 0.207 & 0.069 && 0.094 & 0.046 & 
\multicolumn{1}{c}{} & 0.140 & 0.063 && 0.044 & 0.032 \\[2pt]
$%
\rho ^{(x)}=0.5$ & $WW$--$IM$ & 0.378 & 0.130 && 0.282 & 0.094 & \multicolumn{1}{c}{} & 0.289 & 0.097 && 0.197 & 0.068 & \multicolumn{1}{c}{} & 0.196 & 0.070 && 0.126 & 0.046 \\
& $\eta =0$ & 0.202 & 0.103 && 0.105 & 0.066 &   
\multicolumn{1}{c}{} & 0.151 & 0.078 && 0.070 & 0.046 &
\multicolumn{1}{c}{} & 0.098 & 0.063 && 0.042 & 0.032 \\ 
& $\eta=0.45$ & 0.182 & 0.078 && 0.088 & 0.046  &  
\multicolumn{1}{c}{} & 0.131 & 0.056 && 0.053 & 0.028 & 
\multicolumn{1}{c}{} & 0.073 & 0.047 && 0.019 & 0.018 \\ 
& $\eta =0.49$ & 0.186 & 0.080 && 0.093 & 0.048 & 
\multicolumn{1}{c}{} & 0.136 & 0.052 && 0.055 & 0.025 & 
\multicolumn{1}{c}{} & 0.073 & 0.042 && 0.018 & 0.014
\\ 
& $\eta =0.5$ & 0.187 & 0.081 && 0.094 & 0.049 & 
\multicolumn{1}{c}{} & 0.137 & 0.053 && 0.057 & 0.026  & 
\multicolumn{1}{c}{} & 0.075 & 0.044 && 0.019 & 0.016  \\[6pt] 
\multicolumn{2}{l}{Panel B: serial dependence: $\rho ^{(e)}=0.5$} &  &  &  &

&  & \multicolumn{1}{c}{} &  &  &  &  &  & \multicolumn{1}{c}{} &  &  &  &

&  \\[2pt]
\multicolumn{2}{l}{\textit{No endogeneity: $\rho ^{(xe)}=0$}}
&  &  &  &  & 
& \multicolumn{1}{c}{} &  &  &  &  &  & \multicolumn{1}{c}{}
&  &  &  &  & 
\\[2pt]
$\rho ^{(x)}=0$ & $WW$--$IM$ &  0.498 & 0.180 && 0.402 & 0.144 & \multicolumn{1}{c}{} & 0.438 & 0.148 && 0.328 & 0.112 & \multicolumn{1}{c}{} & 0.348 & 0.119 && 0.247 & 0.087 \\
& $\eta =0$ & 0.284 & 0.140 && 0.196 & 0.100 & 
\multicolumn{1}{c}{} & 0.264 & 0.109 && 0.150 & 0.080 & 
\multicolumn{1}{c}{} & 0.219 & 0.096 && 0.101 & 0.064\\ 
& $\eta =0.45$ & 0.253 & 0.105 && 0.176 & 0.077 & 
\multicolumn{1}{c}{} & 0.237 & 0.077 && 0.131 & 0.057 & 
\multicolumn{1}{c}{} & 0.182 & 0.075 && 0.077 & 0.048\\ 
& $\eta=0.49$ & 0.258 & 0.105 && 0.182 & 0.078 & 
\multicolumn{1}{c}{} & 0.241 & 0.072 && 0.135 & 0.053 & 
\multicolumn{1}{c}{} & 0.186 & 0.068 && 0.076 & 0.044 \\ 
& $\eta =0.5$ & 0.260 & 0.107 && 0.184 & 0.079  & 
\multicolumn{1}{c}{} & 0.243 & 0.074 && 0.136 & 0.055  & 
\multicolumn{1}{c}{} & 0.188 & 0.070 && 0.078 & 0.046 \\[2pt]

$\rho ^{(x)}=0.5$ & $WW$--$IM$ & 0.406 & 0.152 && 0.303 & 0.111 & \multicolumn{1}{c}{} & 0.351 & 0.119 && 0.245 & 0.087 & \multicolumn{1}{c}{} & 0.270 & 0.093 && 0.175 & 0.066 \\
& $\eta =0$ & 0.206 & 0.106 && 0.119 & 0.073 & 
\multicolumn{1}{c}{} & 0.175 & 0.085 && 0.085 & 0.054 &
\multicolumn{1}{c}{} & 0.122 & 0.072 && 0.052 & 0.039 \\ 
& $%
\eta =0.45$ & 0.186 & 0.081 && 0.102 & 0.054 & 
\multicolumn{1}{c}{} & 0.154 & 0.061 && 0.068 & 0.036 & 
\multicolumn{1}{c}{} & 0.095 & 0.055 && 0.029 & 0.025 \\ 
& $\eta =0.49$ & 0.191 & 0.082 && 0.106 & 0.054  & 
\multicolumn{1}{c}{} & 0.157 & 0.057 && 0.070 & 0.032  & 
\multicolumn{1}{c}{} & 0.096 & 0.051 && 0.028 & 0.021 \\ 
& $\eta =0.5$ & 0.192 & 0.083 && 0.107 & 0.056  & 
\multicolumn{1}{c}{} & 0.160 & 0.059 && 0.071 & 0.034  & 
\multicolumn{1}{c}{} & 0.097 & 0.053 && 0.029 & 0.023 \\ 
\multicolumn{2}{l}{\textit{Endogeneity: $\rho ^{(xe)}=0.5$}} &  &  &  &  &

& \multicolumn{1}{c}{} &  &  &  &  &  & \multicolumn{1}{c}{} &  &  &  &  &

\\ 
$\rho ^{(x)}=0$ & $WW$--$IM$ & 0.474 & 0.165 && 0.374 & 0.127 & \multicolumn{1}{c}{} & 0.391 & 0.134 && 0.289 & 0.099 & \multicolumn{1}{c}{} & 0.300 & 0.106 && 0.211 & 0.077  \\
& $\eta =0$ & 0.278 & 0.136 && 0.194 & 0.098 & 
\multicolumn{1}{c}{} & 0.250 & 0.107 && 0.136 & 0.075 & 
\multicolumn{1}{c}{} & 0.206 & 0.090 && 0.088 & 0.061 \\ 
& $\eta=0.45$ & 0.247 & 0.101 && 0.176 & 0.075 & 
\multicolumn{1}{c}{} & 0.226 & 0.077 && 0.117 & 0.054 & 
\multicolumn{1}{c}{} & 0.175 & 0.069 && 0.063 & 0.045 \\ 
& $\eta =0.49$ & 0.254 & 0.102 && 0.182 & 0.076 & 
\multicolumn{1}{c}{} & 0.229 & 0.072 && 0.121 & 0.050 & 
\multicolumn{1}{c}{} &0.176 & 0.064 && 0.063 & 0.040 \\ 
& $\eta =0.5$ & 0.254 & 0.104 && 0.183 & 0.077 & 
\multicolumn{1}{c}{} & 0.230 & 0.073 && 0.234 & 0.051 & 
\multicolumn{1}{c}{} & 0.178 & 0.066 && 0.064 & 0.043 \\[2pt]
$%
\rho ^{(x)}=0.5$ & $WW$--$IM$ & 0.390 & 0.145 && 0.284 & 0.102 & \multicolumn{1}{c}{} & 0.322 & 0.110 && 0.223 & 0.080 & \multicolumn{1}{c}{} & 0.241 & 0.084 && 0.158 & 0.060 \\
& $\eta =0$ & 0.217 & 0.107 && 0.116 & 0.071 & 
\multicolumn{1}{c}{} & 0.175 & 0.086 && 0.081 & 0.053  & 
\multicolumn{1}{c}{} & 0.119 & 0.071 && 0.049 & 0.038\\ 
& $\eta=0.45$ & 0.196 & 0.081 && 0.100 & 0.051  & 
\multicolumn{1}{c}{} & 0.155 & 0.062 && 0.064 & 0.034 & 
\multicolumn{1}{c}{} & 0.092 & 0.054 && 0.027 & 0.023 \\ 
& $\eta =0.49$ & 0.199 & 0.083 && 0.105 & 0.053 & 
\multicolumn{1}{c}{} & 0.157 & 0.058 && 0.066 & 0.031 & 
\multicolumn{1}{c}{} & 0.093 & 0.049 && 0.025 & 0.019\\ 
& $\eta =0.5$ & 0.199 & 0.084 && 0.107 & 0.054 & 
\multicolumn{1}{c}{} & 0.159 & 0.059 && 0.068 & 0.033  & 
\multicolumn{1}{c}{} & 0.095 & 0.052 && 0.027 & 0.021 \\[6pt]
\multicolumn{2}{l}{Panel C: serial dependence: $\rho ^{(e)}=0.9$} &  &  &  &

&  & \multicolumn{1}{c}{} &  &  &  &  &  & \multicolumn{1}{c}{} &  &  &  &

&  \\[2pt]
\multicolumn{2}{l}{\textit{No endogeneity: $\rho ^{(xe)}=0$}}
&  &  &  &  & 
& \multicolumn{1}{c}{} &  &  &  &  &  & \multicolumn{1}{c}{}
&  &  &  &  & 
\\[2pt]
$\rho ^{(x)}=0$ & $WW$--$IM$ & 0.376 & 0.143 && 0.315 & 0.119 & \multicolumn{1}{c}{} & 0.412 & 0.148 && 0.316 & 0.117 & \multicolumn{1}{c}{} & 0.382 & 0.140 & \multicolumn{1}{c}{} & 0.285 & 0.106   \\
& $\eta =0$ & 0.191 & 0.123 && 0.122 & 0.090 & 
\multicolumn{1}{c}{} & 0.222 & 0.112 && 0.128 & 0.078  & 
\multicolumn{1}{c}{} & 0.213 & 0.101 && 0.106 & 0.068
\\ 
& $\eta =0.45$ & 0.167 & 0.089 && 0.103 & 0.067 & 
\multicolumn{1}{c}{} & 0.193 & 0.079 && 0.106 & 0.055 & 
\multicolumn{1}{c}{} & 0.180 & 0.101 && 0.080 & 0.052\\ 
& $\eta=0.49$ & 0.174 & 0.091 && 0.109 & 0.068 & 
\multicolumn{1}{c}{} & 0.199 & 0.074 && 0.110 & 0.051  & 
\multicolumn{1}{c}{} & 0.178 & 0.071 && 0.079 & 0.047 \\ 
& $\eta =0.5$ & 0.175 & 0.093 && 0.111 & 0.070  & 
\multicolumn{1}{c}{} & 0.199 & 0.076 && 0.112 & 0.053 & 
\multicolumn{1}{c}{} & 0.180 & 0.074 && 0.081 & 0.049 \\[2pt]
$\rho ^{(x)}=0.5$ & $WW$--$IM$ & 0.320 & 0.122 && 0.237 & 0.093 & \multicolumn{1}{c}{} & 0.331 & 0.119 && 0.236 & 0.090 & \multicolumn{1}{c}{} & 0.301 & 0.112 && 0.208 & 0.082 \\
& $\eta =0$ & 0.124 & 0.094 && 0.060 & 0.065 & 
\multicolumn{1}{c}{} & 0.143 & 0.085 && 0.068 & 0.055 &
\multicolumn{1}{c}{} & 0.125 & 0.076 && 0.053 & 0.044 \\ 
& $%
\eta =0.45$ & 0.103 & 0.070 && 0.043 & 0.044  & 
\multicolumn{1}{c%
}{} & 0.123 & 0.060 && 0.050 & 0.036  & 
\multicolumn{1}{c}{} & 0.098 & 0.058 && 0.030 & 0.029 \\ 
& $\eta =0.49$ & 0.108 & 0.071 && 0.047 & 0.046  & 
\multicolumn{1}{c}{} & 0.126 & 0.057 && 0.052 & 0.033  & 
\multicolumn{1}{c}{} & 0.098 & 0.054 && 0.029 & 0.025 \\ 
& $\eta =0.5$ & 0.110 & 0.072 && 0.049 & 0.048  & 
\multicolumn{1}{c}{} & 0.129 & 0.058 && 0.054 & 0.034 & 
\multicolumn{1}{c}{} & 0.100 & 0.056 && 0.030 & 0.027 \\ 
\multicolumn{2}{l}{\textit{Endogeneity: $\rho ^{(xe)}=0.5$}} &  &  &  &  &

& \multicolumn{1}{c}{} &  &  &  &  &  & \multicolumn{1}{c}{} &  &  &  &  &

\\ 
$\rho ^{(x)}=0$ & $WW$--$IM$ & 0.353 & 0.127 && 0.283 & 0.106 & \multicolumn{1}{c}{} & 0.364 & 0.133 && 0.275 & 0.103 & \multicolumn{1}{c}{} & 0.341 & 0.126 && 0.248 & 0.094  \\
& $\eta =0$ & 0.193 & 0.116 && 0.109 & 0.084 & 
\multicolumn{1}{c}{} & 0.212 & 0.101 && 0.109 & 0.070  & 
\multicolumn{1}{c}{} & 0.191 & 0.092 && 0.084 & 0.062 \\ 
& $\eta=0.45$ & 0.166 & 0.087 && 0.091 & 0.063 & 
\multicolumn{1}{c}{} & 0.187 & 0.071 && 0.089 & 0.049 & 
\multicolumn{1}{c}{} &  0.160 & 0.072 && 0.060 & 0.047\\ 
& $\eta =0.49$ & 0.171 & 0.088 && 0.097 & 0.064 & 
\multicolumn{1}{c}{} & 0.189 & 0.068 && 0.093 & 0.045 & 
\multicolumn{1}{c}{} & 0.162 & 0.065 && 0.059 & 0.042\\ 
& $\eta =0.5$ & 0.174 & 0.089 && 0.098 & 0.065 & 
\multicolumn{1}{c}{} & 0.191 & 0.070 && 0.094 & 0.046 & 
\multicolumn{1}{c}{} & 0.160 & 0.068 && 0.060 & 0.044 \\[2pt]
$%
\rho ^{(x)}=0.5$ & $WW$--$IM$ & 0.308 & 0.115 && 0.219 & 0.084 & \multicolumn{1}{c}{} & 0.303 & 0.110 &&  0.213 & 0.081 & \multicolumn{1}{c}{} & 0.268 & 0.099 && 0.184 & 0.073\\
& $\eta =0$ & 0.132 & 0.088 && 0.060 & 0.058 & 
\multicolumn{1}{c}{} & 0.129 & 0.078 && 0.060 & 0.048 & 
\multicolumn{1}{c}{} & 0.105 & 0.067 && 0.044 & 0.037\\ 
& $\eta=0.45$ & 0.113 & 0.065 && 0.043 & 0.039 & 
\multicolumn{1}{c}{} & 0.110 & 0.054 && 0.043 & 0.030 & 
\multicolumn{1}{c}{} & 0.078 & 0.050 && 0.021 & 0.023 \\ 
& $\eta =0.49$ & 0.117 & 0.065 && 0.047 & 0.041 & 
\multicolumn{1}{c}{} & 0.114 & 0.051 && 0.045 & 0.026 & 
\multicolumn{1}{c}{} & 0.079 & 0.046 && 0.020 & 0.021\\ 
& $\eta =0.5$ & 0.119 & 0.067 && 0.048 & 0.042 & 
\multicolumn{1}{c}{} & 0.117 & 0.052 && 0.046 & 0.028 & 
\multicolumn{1}{c}{} & 0.081 & 0.048 && 0.021 & 0.021\\[2pt]\hline
\end{tabular}
}
\label{tab:Table2b}
\end{table}
\end{landscape}

\topmargin1.0cm
\textwidth14.25cm
\textheight20.0cm
\oddsidemargin1cm
\evensidemargin1cm

\begin{table}[H]
\centering
\caption{Empirical rejection frequencies under $H_{A,2}$.}
{\tiny
\begin{tabular}{llllllllll}
&  &  &  &  &  &  &  &  & 
\\[-6pt] \hline \\
&  & \multicolumn{2}{c}{$T=100$} &  & \multicolumn{2}{c}{$%
T=200$} &  & 
\multicolumn{2}{c}{$T=400$} \\[2pt]
&  & $m=25$ & $m=50$ & 
& $m=50$ & $m=100$ &  & $m=100$ & $m=200$ \\%
[2pt]\cline{3-4}\cline{6-7}%
\cline{9-10}
&  &  &  &  &  &  &  &  &  \\[-6pt] 
\multicolumn{2}{l}{Panel
A: no serial dependence: $\rho ^{(e)}=0$} &  &  & 
&  &  &  &  &  \\[2pt]
\multicolumn{2}{l}{\textit{No endogeneity: $\rho ^{(xe)}=0$}} &  &  &  &  &

&  &  &  \\[2pt]
$\rho ^{(x)}=0$  & $WW$--$IM$ & 0.556 & 0.445 && 0.787 & 0.699 && 0.925 & 0.907 \\
& $\eta =0$ & $0.946$ & $0.564$ &  & $0.995$ & $0.774$ &  & $0.999$ & $0.948$ \\ 
& $\eta =0.45$ & $0.940$ & $0.589$ &  & $0.994$ & $0.793$ &  & $0.999$ & $0.950$ \\ 
& $\eta =0.49$ & $0.933$ & $0.581$ &  & $0.992$ & $0.785$ &  & $0.999$ & $0.949$ \\ 
& $\eta =0.5$ & $0.931$ & $0.574$ &  & $0.992$ & $0.780$ &  & $0.999$ & $0.947$ \\[2pt]
$\rho ^{(x)}=0.5$ & $WW$--$IM$ & 0.637 & 0.510 && 0.820 & 0.728 && 0.936 & 0.918 \\
& $\eta =0$ & 0.948 & 0.569 && 0.994 & 0.773 && 0.999 & 0.947 \\ 
& $\eta =0.45$ & 0.941 & 0.591 && 0.993 & 0.790 && 0.999 & 0.950 \\ 
& $\eta =0.49$ & 0.934 & 0.584 && 0.993 & 0.784 && 0.999 & 0.949  \\ 
& $\eta =0.5$ & 0.934 & 0.580 && 0.993 & 0.780 && 0.999 & 0.947 \\ 
\multicolumn{2}{l}{\textit{Endogeneity: $\rho ^{(xe)}=0.5$}}
&  &  &  &  & 
&  &  &  \\ 
$\rho ^{(x)}=0$ & $WW$--$IM$ & 0.535 & 0.430 && 0.784 & 0.697 && 0.923 & 0.906 \\
& $\eta =0$ & 0.875 & 0.452 && 0.971 & 0.651 && 0.998 & 0.875 \\ 
& $\eta =0.45$ & 0.873 & 0.483 && 0.970 & 0.674 && 0.996 & 0.888 \\ 
& $\eta=0.49$ & 0.869 & 0.476 && 0.969 & 0.666 && 0.996 & 0.876 \\ 
& $\eta =0.5$ & 0.865 & 0.469 && 0.968 & 0.660 && 0.996 & 0.876 \\[2pt]
$\rho ^{(x)}=0.5$ & $WW$--$IM$ & 0.591 & 0.460 && 0.791 & 0.700 && 0.927 & 0.910  \\
& $\eta =0$ & 0.833 & 0.378 && 0.957 & 0.573 && 0.994 & 0.792 \\ 
& $\eta =0.45$ & 0.832 & 0.407 && 0.946 & 0.601 && 0.993 & 0.807\\ 
& $\eta =0.49$ & 0.826 & 0.400 && 0.940 & 0.592 && 0.993 & 0.799 \\ 
& $\eta =0.5$ & 0.824 & 0.393 && 0.937 & 0.589 && 0.993 & 0.794 \\[6pt] 
\multicolumn{2}{l}{Panel B: serial dependence: $\rho
^{(e)}=0.5$} &  &  &  & 
&  &  &  &  \\[2pt]
\multicolumn{2}{l}{\textit{No
endogeneity: $\rho ^{(xe)}=0$}} &  &  &  &  & 
&  &  &  \\[2pt]
$\rho^{(x)}=0$ & $WW$--$IM$ & 0.519 & 0.371 && 0.675 & 0.555 && 0.851 & 0.782\\
& $\eta =0$ & 0.793 & 0.321 && 0.911 & 0.480 && 0.980 & 0.696 \\ 
& $\eta =0.45$ & 0.788 & 0.351 &&  0.908 & 0.511 && 0.978 & 0.713  \\ 
& $\eta =0.49$ & 0.780 &  0.343 && 0.902 & 0.504 && 0.977 & 0.710  \\ 
& $\eta =0.5$ & 0.775 & 0.338 && 0.901 & 0.497 && 0.976 & 0.699  \\[2pt]
$\rho^{(x)}=0.5$ & $WW$--$IM$ & 0.548 & 0.370 && 0.669 & 0.548 && 0.845 &  0.778 \\
& $\eta =0$ & 0.808 & 0.324 && 0.914 & 0.482 && 0.980 & 0.695  \\ 
& $\eta =0.45$ & 0.802 & 0.358 && 0.911 & 0.510 && 0.978 & 0.713 \\ 
& $\eta =0.49$ & 0.792 & 0.350 && 0.905 & 0.505 && 0.978 & 0.710 \\ 
& $\eta =0.5$ & 0.790 & 0.341 && 0.902 & 0.502 && 0.977 & 0.699 \\  
\multicolumn{2%
}{l}{\textit{Endogeneity: $\rho ^{(xe)}=0.5$}} &  &  &  &  & 
&  &  &  \\ 
$\rho ^{(x)}=0$ & $WW$--$IM$ & 0.503 & 0.389 && 0.692 & 0.573 && 0.861 & 0.797 \\
& $\eta =0$ & 0.742 & 0.258 && 0.883 & 0.408 && 0.970 & 0.624  \\ 
& $\eta =0.45$ & 0.740 & 0.288 && 0.877 & 0.447 && 0.967 & 0.646  \\ 
& $\eta =0.49$ & 0.730 & 0.282 && 0.875 & 0.437 && 0.964 & 0.643 \\ 
& $\eta =0.5$ & 0.727 & 0.280 && 0.871 & 0.432 && 0.963 & 0.631 \\[2pt]
$\rho^{(x)}=0.5$ & $WW$--$IM$ & 0.524 & 0.348 && 0.654 & 0.533 && 0.839 & 0.773  \\
& $\eta =0$ & 0.709 & 0.222 && 0.852 & 0.377 && 0.957 & 0.566  \\ 
& $\eta =0.45$ & 0.710 & 0.259 && 0.850 & 0.411 && 0.954 & 0.594  \\ 
& $\eta =0.49$ & 0.703 & 0.252 && 0.841 & 0.407 && 0.950 & 0.588 \\ 
& $\eta =0.5$ & 0.700 & 0.248 && 0.841 & 0.403 && 0.950 & 0.576 \\[6pt]
\multicolumn{2}{l}{Panel C: serial dependence: $\rho
^{(e)}=0.9$} &  &  &  & 
&  &  &  &  \\[2pt]
\multicolumn{2}{l}{\textit{No
endogeneity: $\rho ^{(xe)}=0$}} &  &  &  &  & 
&  &  &  \\[2pt]
$\rho^{(x)}=0$ & $WW$--$IM$ & 0.478 & 0.339 && 0.479 & 0.292 && 0.534 &  0.404\\
& $\eta =0$ & 0.602 & 0.113 && 0.609 & 0.097 && 0.678 & 0.186 \\ 
& $\eta =0.45$ & 0.607 & 0.151 && 0.604 & 0.144 && 0.685 & 0.217  \\ 
& $\eta =0.49$ & 0.598 & 0.145 && 0.589 & 0.147 && 0.678 & 0.223 \\ 
& $\eta =0.5$ & 0.595 & 0.143 && 0.582 & 0.143 && 0.675 & 0.217 \\[2pt]
$\rho^{(x)}=0.5$ & $WW$--$IM$ & 0.504 & 0.375 && 0.491 & 0.310 && 0.540 &    0.400  \\
& $\eta =0$ & 0.645 & 0.127 && 0.623 & 0.103 && 0.685 & 0.187 \\ 
& $\eta =0.45$ & 0.640 & 0.166 && 0.619 & 0.151 && 0.691 & 0.223  \\ 
& $\eta =0.49$ & 0.623 & 0.161 && 0.600 & 0.152 && 0.690 & 0.227 \\ 
& $\eta =0.5$ & 0.621 & 0.157 && 0.592 & 0.146 && 0.685 & 0.219\\ 
\multicolumn{2%
}{l}{\textit{Endogeneity: $\rho ^{(xe)}=0.5$}} &  &  &  &  & 
&  &  &   \\[2pt]
$\rho^{(x)}=0$ & $WW$--$IM$ & 0.488 & 0.364 && 0.491 & 0.320 && 0.552 & 0.420\\
& $\eta =0$ &  0.597 & 0.105 && 0.583 & 0.085 && 0.666 & 0.182 \\ 
& $\eta =0.45$ & 0.601 & 0.143 && 0.584 & 0.135 && 0.674 & 0.216  \\ 
& $\eta =0.49$ & 0.589 & 0.141 && 0.570 & 0.138 && 0.672 & 0.221  \\ 
& $\eta =0.5$ & 0.588 & 0.137 && 0.560 & 0.132 && 0.665 & 0.212  \\[2pt]
$\rho ^{(x)}=0.5$ & $WW$--$IM$ & 0.485 & 0.368 && 0.493 & 0.319 && 0.548 & 0.410 \\
& $\eta =0$ & 0.600 & 0.108 && 0.585 & 0.088 && 0.660 & 0.172 \\ 
& $\eta =0.45$ & 0.599 & 0.145 && 0.586 & 0.135 && 0.666 & 0.205  \\ 
& $\eta =0.49$ & 0.593 & 0.142 && 0.575 & 0.139 && 0.656 & 0.211 \\ 
& $\eta =0.5$ & 0.592 & 0.139 && 0.573 & 0.135 && 0.651 & 0.203 \\[2pt]\hline
\end{%
tabular}
}
\label{tab:Table3a}
\end{table}

\begin{table}[H]
\centering
\caption{%
Detection delay under $H_{A,2}$.}
\tiny
\begin{tabular}{llllllllll}

&  &  &  &  &  &  &  &  &  \\[-6pt] \hline \\
&
& \multicolumn{2}{c}{$T=100$} &  & \multicolumn{2}{c}{$T=200$} &  & 
\multicolumn{2}{c}{$T=400$} \\[2pt]
&  & $m=25$ & $%
m=50$ &  & $m=50$ & $%
m=100$ & 
& $m=100$ & $m=200$ \\%
[2pt]\cline{3-4}%
\cline{6-7}\cline{9-10}
&  &  &  &  &  &  &  &  &  \\[-6pt] 
\multicolumn{2%
}{l}{Panel A: no serial dependence: $\rho ^{(e)}=0$} &  &  & 
&  &  &  &  &
\\[2pt]
\multicolumn{2}{l}{\textit{No endogeneity: $\rho ^{(xe)}=0$}} &  & 
&  &  & 
&  &  &  \\[2pt]
$\rho ^{(x)}=0$ & $WW$--$IM$ & 0.644 & 0.240 && 0.540 & 0.211 && 0.407 & 0.165 \\
& $\eta =0$ & $0.375$ & $0.215$ &  & $0.304$ & $0.193$ &  & $0.215$ & $0.159$ \\ 
& $\eta =0.45$ & $0.349$ & $0.175$ &  & $0.283$ & $0.163$ &  & $0.186$ & $0.138$ \\ 
& $\eta=0.49$ & $0.356$ & $0.178$ &  & $0.287$ & $0.157$ &  & $0.186$ & $0.133$ \\ 
& $\eta =0.5$ & $0.357$ & $0.179$ &  & $0.291$ & $0.159$ &  & $0.188$ &  $0.137$ \\[2pt]

$\rho ^{(x)}=0.5$ & $WW$--$IM$ & 0.592 & 0.229 && 0.516 & 0.204 && 0.392 & 0.160 \\
& $\eta =0$ & 0.360 & 0.214 && 0.302 & 0.193 && 0.214 & 0.159 \\ 
& $\eta =0.45$ & 0.332 & 0.175 && 0.281 & 0.162 && 0.186 & 0.138  \\ 
& $\eta =0.49$ & 0.339 & 0.177 && 0.287 & 0.157 && 0.186 & 0.133 \\ 
& $\eta =0.5$ & 0.343 & 0.179 && 0.291 & 0.159 && 0.188 & 0.137  \\ 
\multicolumn{2}{l}{\textit{Endogeneity: $\rho ^{(xe)}=0.5$}}
&  &  &  &  & 
&  &  &  \\ 
$\rho ^{(x)}=0$ & $WW$--$IM$ &  0.656 & 0.241 && 0.546 & 0.212 && 0.409 & 0.165 \\
& $\eta =0$ & 0.406 & 0.233 && 0.356 & 0.210 && 0.262 & 0.182 \\ 
& $\eta =0.45$ & 0.378 & 0.185 && 0.334 & 0.172 && 0.231 & 0.160\\ 
& $\eta=0.49$ & 0.388 & 0.186 && 0.341 & 0.166 && 0.231 & 0.153 \\ 
& $\eta =0.5$ & 0.389 & 0.188 && 0.344 & 0.169 && 0.234 & 0.157 \\[2pt]
$\rho ^{(x)}=0.5$ & $WW$--$IM$ & 0.620 & 0.237 && 0.533 & 0.209 && 0.404 & 0.163\\
& $\eta =0$ & 0.418 & 0.243 && 0.399 & 0.222 && 0.300 & 0.194  \\ 
& $\eta =0.45$ & 0.393 & 0.185 && 0.370 & 0.179 && 0.268 & 0.169 \\ 
& $\eta =0.49$ & 0.402 & 0.186 && 0.374 & 0.171 && 0.268 & 0.161 \\ 
& $\eta =0.5$ & 0.404 & 0.188 && 0.376 & 0.175 && 0.272 & 0.165  \\[6pt] 
\multicolumn{2}{l}{Panel B: serial dependence: $\rho
^{(e)}=0.5$} &  &  &  & 
&  &  &  &  \\[2pt]
\multicolumn{2}{l}{\textit{No
endogeneity: $\rho ^{(xe)}=0$}} &  &  &  &  & 
&  &  &  \\[2pt]
$\rho^{(x)}=0$ & $WW$--$IM$ & 0.617 & 0.220 && 0.572 & 0.222 && 0.481 & 0.190 \\
& $\eta =0$ & 0.449 & 0.245 && 0.426 & 0.228 && 0.337 & 0.202 \\ 
& $\eta =0.45$ & 0.416 & 0.177 && 0.400 & 0.177 && 0.304 & 0.174  \\ 
& $\eta =0.49$ & 0.423 & 0.178 && 0.404 & 0.169 && 0.303 & 0.165 \\ 
& $\eta =0.5$ & 0.422 & 0.181 && 0.408 & 0.171 &&  0.306 & 0.168  \\[2pt]
$\rho^{(x)}=0.5$ & $WW$--$IM$ & 0.582 & 0.212 && 0.568 & 0.222 && 0.482 &   0.190  \\
& $\eta =0$ &  0.433 & 0.243 && 0.424 & 0.228 && 0.336 & 0.202  \\ 
& $\eta =0.45$ & 0.400 & 0.178 && 0.399 & 0.177 && 0.302 & 0.173\\ 
& $\eta =0.49$ & 0.406 & 0.179 && 0.403 & 0.169 && 0.302 & 0.165 \\ 
& $\eta =0.5$ & 0.407 & 0.180 && 0.405 & 0.172 && 0.305 & 0.168 \\ 
\multicolumn{2%
}{l}{\textit{Endogeneity: $\rho ^{(xe)}=0.5$}} &  &  &  &  & 
&  &  &  \\ 
$\rho ^{(x)}=0$ & $WW$--$IM$ & 0.617 & 0.214 && 0.562 & 0.215 && 0.465 & 0.186  \\
& $\eta =0$ & 0.466 & 0.247 && 0.453 & 0.235 && 0.366 & 0.209 \\ 
& $\eta =0.45$ & 0.430 & 0.164 && 0.425 & 0.177 && 0.330 & 0.177 \\ 
& $\eta =0.49$ & 0.439 & 0.165 && 0.433 & 0.165 && 0.328 & 0.168 \\ 
& $\eta =0.5$ & 0.440 & 0.169 && 0.434 & 0.168 && 0.332 & 0.171  \\[2pt]
$\rho^{(x)}=0.5$ & $WW$--$IM$ & 0.600 & 0.219 && 0.577 & 0.223 && 0.489 & 0.191    \\
& $\eta =0$ & 0.445 & 0.249 && 0.470 & 0.242 && 0.391 & 0.213   \\ 
& $\eta =0.45$ & 0.413 & 0.159 && 0.441 & 0.176 && 0.354 & 0.179  \\ 
& $\eta =0.49$ & 0.422 & 0.159 && 0.445 & 0.166 && 0.352 & 0.167  \\ 
& $\eta =0.5$ & 0.423 & 0.163 && 0.450 & 0.170 && 0.357 & 0.171  \\[6pt]
\multicolumn{2}{l}{Panel C: serial dependence: $\rho
^{(e)}=0.9$} &  &  &  & 
&  &  &  &  \\[2pt]
\multicolumn{2}{l}{\textit{No
endogeneity: $\rho ^{(xe)}=0$}} &  &  &  &  & 
&  &  &  \\[2pt]
$\rho^{(x)}=0$ & $WW$--$IM$ & 0.461 & 0.121 && 0.525 & 0.156 && 0.566 &  0.204 \\
& $\eta =0$ & 0.357 & 0.196 && 0.495 & 0.256 && 0.519 & 0.246  \\ 
& $\eta =0.45$ & 0.318 & 0.055 && 0.443 & 0.071 && 0.460 & 0.145  \\ 
& $\eta =0.49$ & 0.323 & 0.050 && 0.443 & 0.042 && 0.456 & 0.118 \\ 
& $\eta =0.5$ & 0.324 & 0.056 && 0.445 & 0.047 && 0.460 & 0.129 \\[2pt]
$\rho^{(x)}=0.5$ & $WW$--$IM$ & 0.430 & 0.120 && 0.513 & 0.152 && 0.554 & 0.200 \\
& $\eta =0$ & 0.328 & 0.188 && 0.471 & 0.251 && 0.514 & 0.245 \\ 
& $\eta =0.45$ & 0.280 & 0.061 && 0.423 & 0.078 && 0.456 & 0.149\\ 
& $\eta =0.49$ & 0.280 & 0.060 && 0.420 & 0.048 && 0.457 & 0.121\\ 
& $\eta =0.5$ & 0.282 & 0.062 && 0.419 & 0.049 && 0.460 & 0.130\\ 
\multicolumn{2%
}{l}{\textit{Endogeneity: $\rho ^{(xe)}=0.5$}} &  &  &  &  & 
&  &  &  \\ 
$\rho ^{(x)}=0$ & $WW$--$IM$ & 0.455 & 0.107 && 0.494 & 0.144 && 0.547 & 0.197\\
& $\eta =0$ & 0.378 & 0.199 && 0.501 & 0.256 && 0.529 & 0.251 \\ 
& $\eta =0.45$ & 0.337 & 0.049 && 0.451 & 0.060 && 0.470 & 0.150 \\ 
& $\eta =0.49$ & 0.341 & 0.051 && 0.452 & 0.029 && 0.470 & 0.122 \\ 
& $\eta =0.5$ & 0.346 & 0.055 && 0.450 & 0.030 && 0.472 & 0.130\\[2pt]
$\rho^{(x)}=0.5$ & $WW$--$IM$ & 0.458 & 0.122 && 0.518 & 0.148 && 0.551 & 0.197 \\
& $\eta =0$ & 0.350 & 0.197 && 0.497 & 0.259 && 0.532 & 0.252 \\ 
& $\eta =0.45$ & 0.304 & 0.048 && 0.447 & 0.061 && 0.471 & 0.144 \\ 
& $\eta =0.49$ & 0.313 & 0.050 && 0.452 & 0.033 && 0.465 & 0.116 \\ 
& $\eta =0.5$ & 0.317 & 0.055 && 0.457 & 0.039 && 0.469 & 0.125 \\[2pt]\hline
\end{%
tabular}
\label{tab:Table3b}
\end{table}

\newpage

\topmargin1.0cm
\textwidth14.25cm
\textheight20.0cm
\oddsidemargin1cm
\evensidemargin1cm

\begin{figure}[]
\caption{Histograms of estimated break dates for $WW$--$IM$, $\eta=0$ and $\eta=0.45$ test procedures under $H_{A,1}$}
\label{histHA1}
\centering
\begin{subfigure}{\textwidth}
\centering
\includegraphics[width=0.85\textwidth]{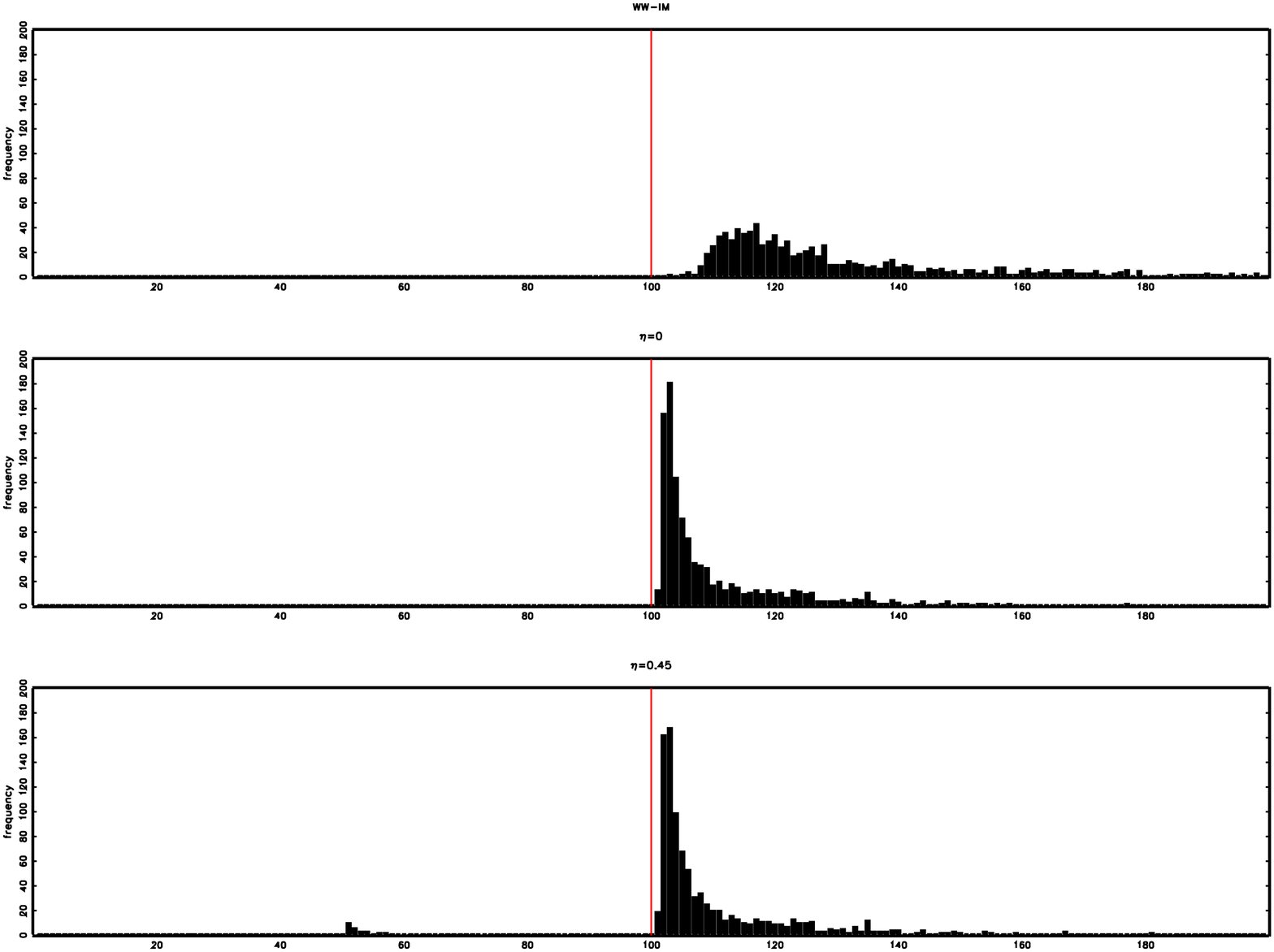}
\caption{$T=200$, $m=\frac{T}{4}$, $\Delta_{\beta}=1$}
\label{HA1_200}
\end{subfigure}
\par
\vspace{3.00mm}
\par
\begin{subfigure}{\textwidth}
\centering
\includegraphics[width=0.85\textwidth]{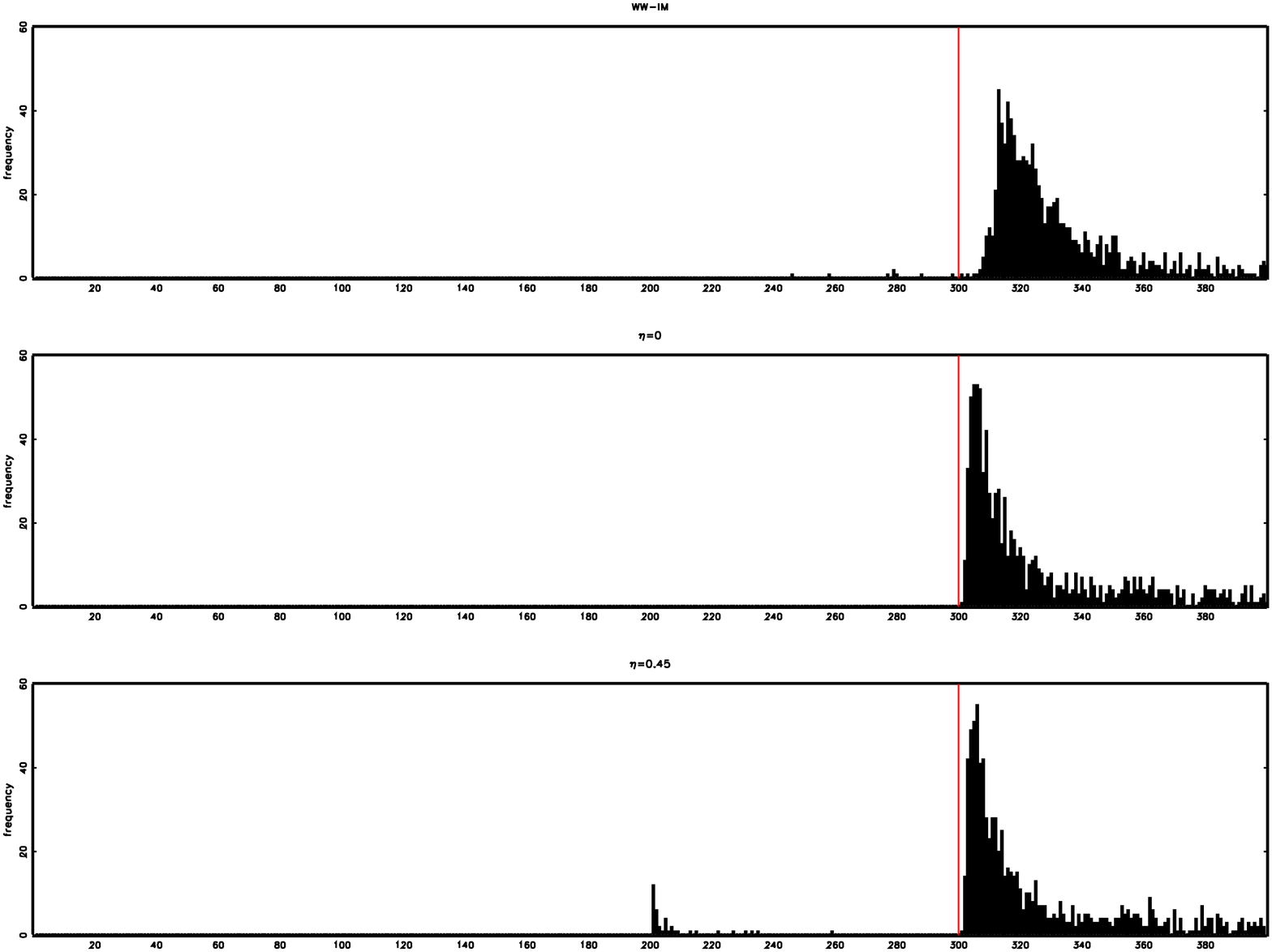}
\caption{$T=400$, $m=\frac{T}{2}$, $\Delta_{\beta}=0.5$}
\label{HA1_400}
\end{subfigure}
\captionsetup{justification=centering}	
\caption*{\textcolor{red}{---} $k^{*}$}
\end{figure}

\begin{figure}[]
\caption{Histograms of estimated break dates for $WW$--$IM$, $\eta=0$ and $\eta=0.45$ test procedures under $H_{A,2}$}
\label{histHA2}
\centering
\begin{subfigure}{\textwidth}
\centering
\includegraphics[width=0.85\textwidth]{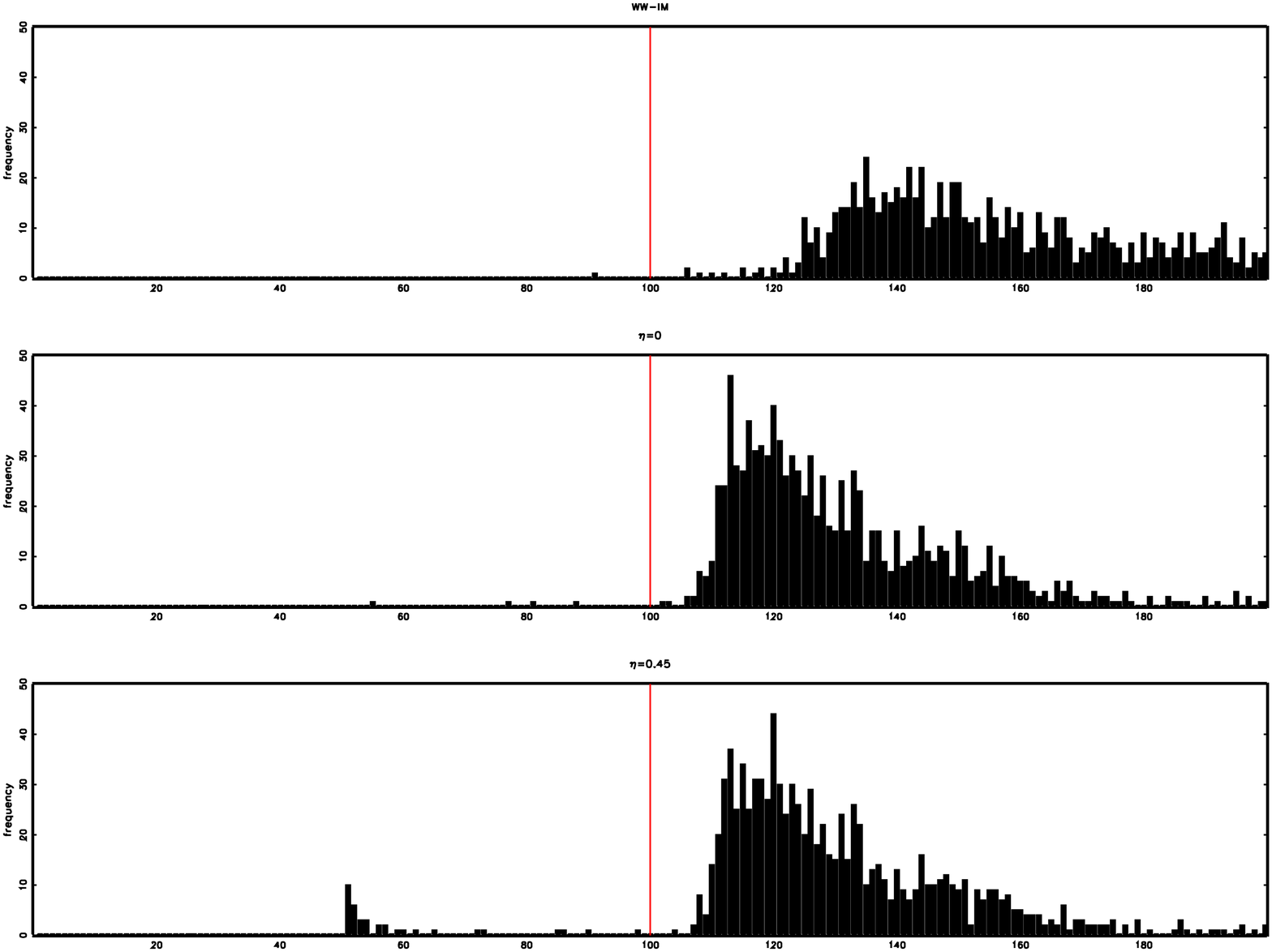}
\caption{$T=200$, $m=\frac{T}{4}$}
\label{HA2_200}
\end{subfigure}
\par
\vspace{3.00mm}
\par
\begin{subfigure}{\textwidth}
\centering
\includegraphics[width=0.85\textwidth]{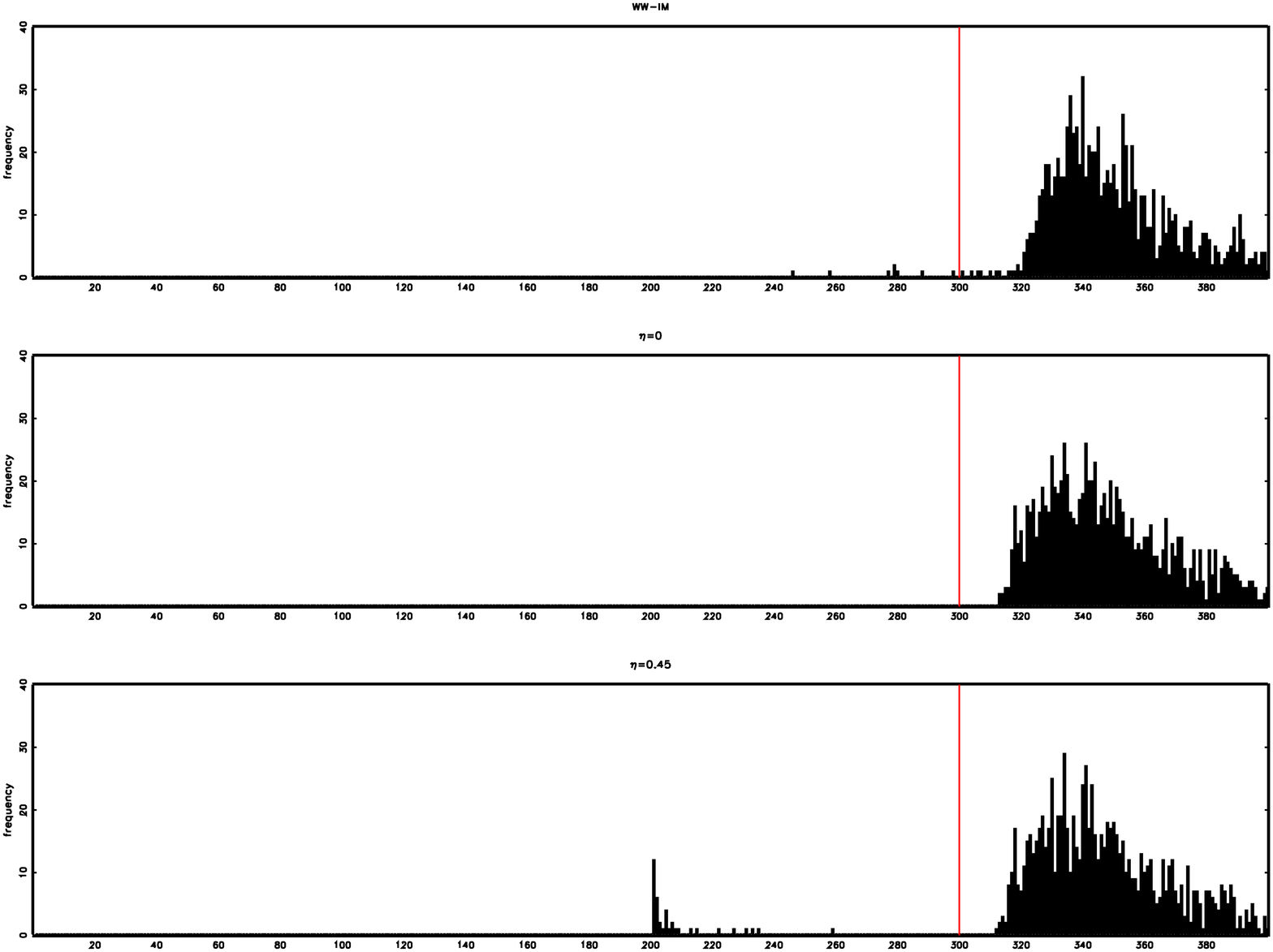}
\caption{$T=400$, $m=\frac{T}{2}$}
\label{HA2_400}
\end{subfigure}
\captionsetup{justification=centering}	
\caption*{\textcolor{red}{---} $k^{*}$}
\end{figure}

\begin{figure}[]
\caption{Residuals from price-to-rent and inverted demand models of housing with break date estimates}%
\label{empexgraphs}
\centering
\begin{subfigure}{.49\linewidth}
	\centering
		\includegraphics[width=\textwidth]{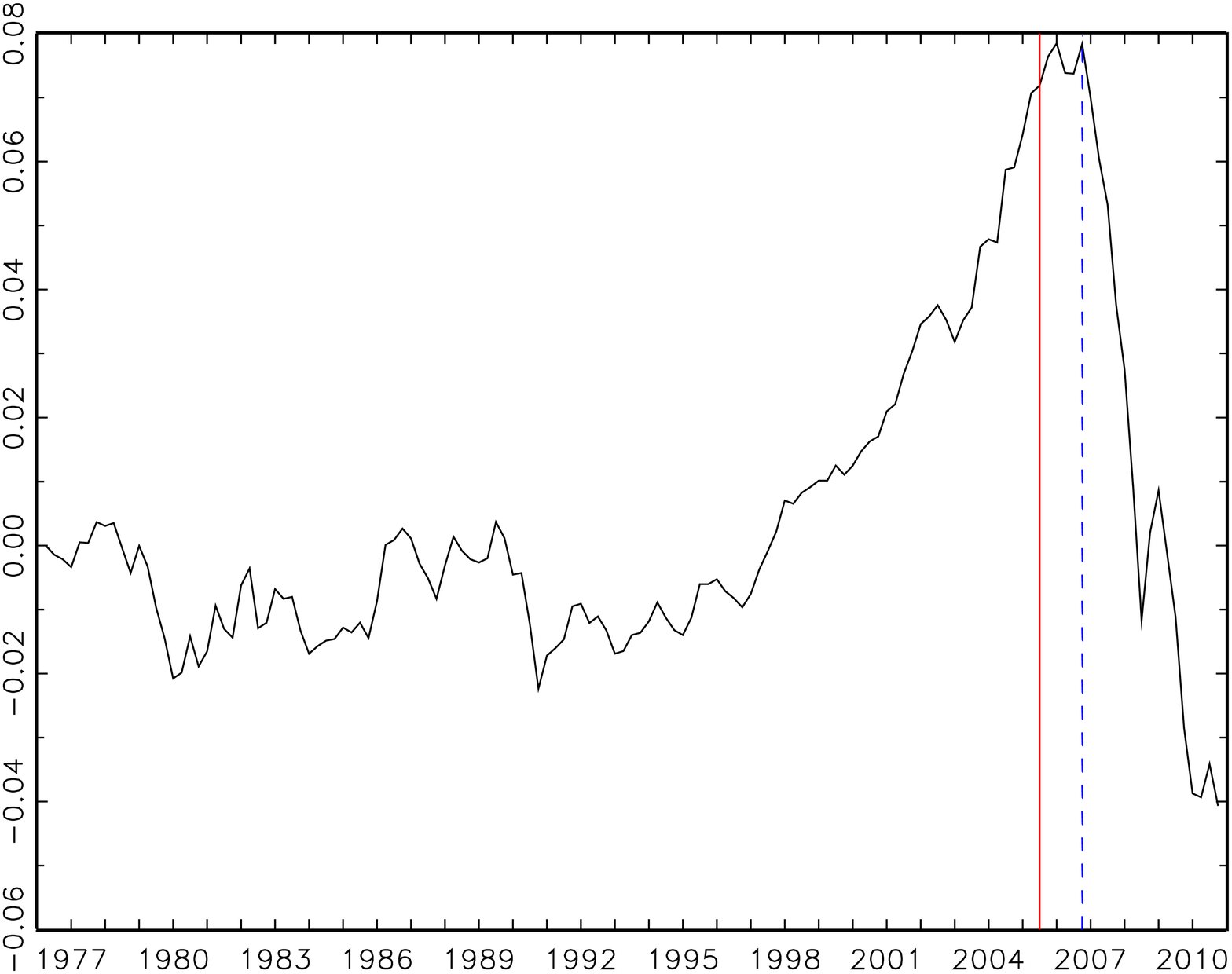} 
		\caption{Price-to-rent model}
		\label{ptr}
	\end{subfigure} 
\begin{subfigure}{.49\linewidth}
	\centering	
		\includegraphics[width=\textwidth]{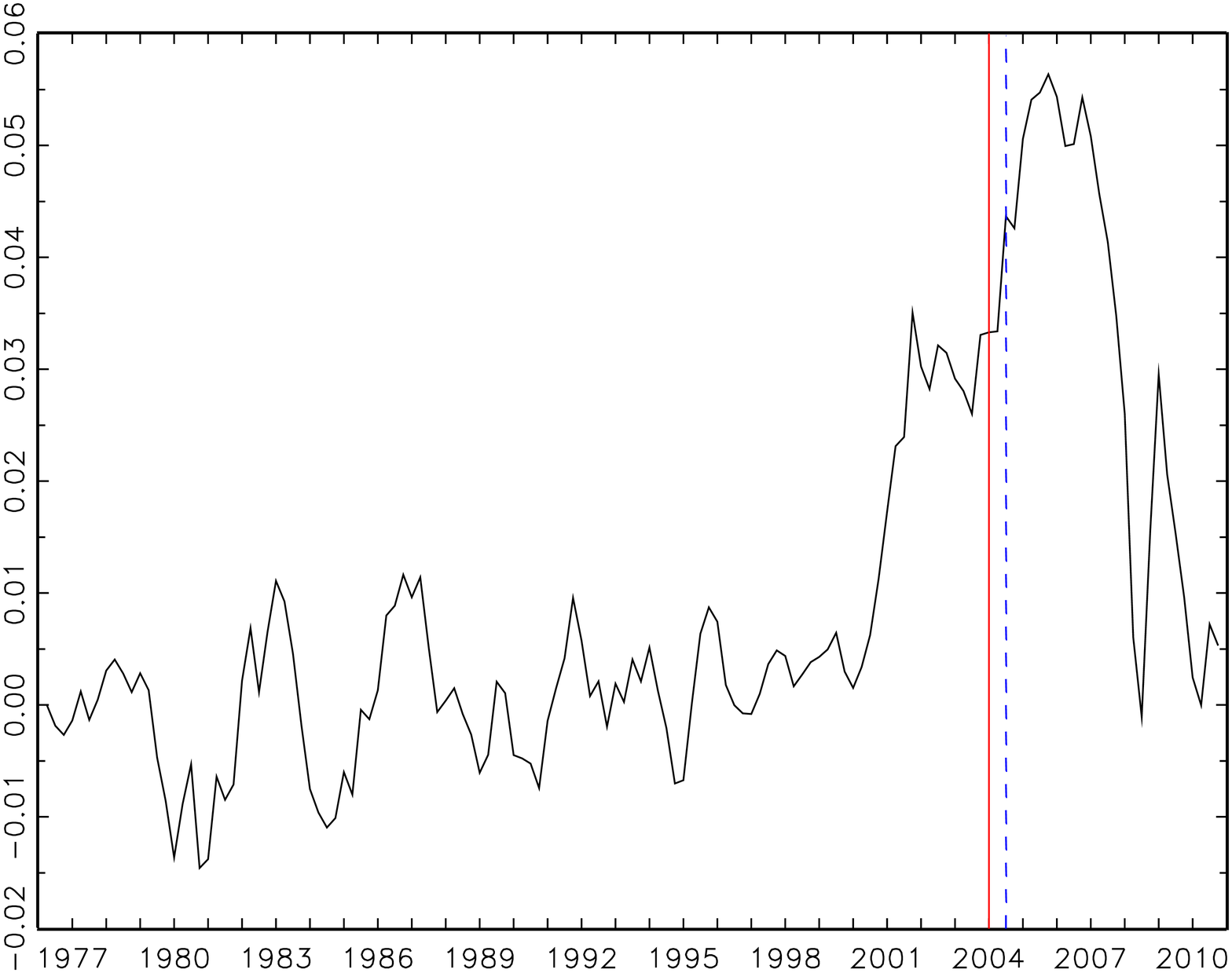}
		\caption{Inverted demand model}
		\label{id}
		\end{subfigure}
\captionsetup{justification=centering}	
\caption*{\textcolor{red}{---} $\widehat{k}_{m}$, \textcolor{blue}{- -} $WW$--$IM$ break date estimate\\
}
\end{figure}

\end{document}